\documentclass[11pt]{article}

\usepackage{odonnell}

\makeatletter
\newenvironment{proofof}[1]{\par
  \pushQED{\qed}%
  \normalfont \topsep6\p@\@plus6\p@\relax
  \trivlist
  \item[\hskip\labelsep
\emph{Proof of #1\@addpunct{.}}]\ignorespaces
}{%
  \popQED\endtrivlist\@endpefalse
}
\makeatother

\newcommand{\rodnote}[1]{\footnote{\color{blue}Ryan: {#1}}}

\newcommand{\rasnote}[1]{\footnote{\color{red}Rocco: {#1}}}
\newcommand{\kanenote}[1]{{{#1}}}

\newcommand{\MollifierChecks}{\textsc{MollifierChecks}}
\newcommand{\AnalysisChecks}{\textsc{AnalysisChecks}}

\newcommand{\epsprg}{\eps_{\textnormal{PRG}}}
\newcommand{\len}{L}
\newcommand{\degree}{d}
\newcommand{\noisegeneric}{\lambda}
\newcommand{\noisefixed}{\ul{\lambda}}

\newcommand{\hypergeneric}{R}
\newcommand{\hypergenerictwo}{{S}}
\newcommand{\hyperfixed}{{\ul{R}}}

\newcommand{\wise}{k_{\mathrm{indep}}}

\newcommand{\allsoftchecks}{\mathrm{Mollifier}_{\p}}
\newcommand{\Indplus}{I_+} \newcommand{\Indminus}{I_-} \newcommand{\Indplusorminus}{I_\pm}
\newcommand{\stats}{\mathfrak{S}}
\newcommand{\ineq}{\textsc{Ineq}}
\newcommand{\softness}{\delta}

\newcommand{\restr}[3]{{#1}_{#2 \mid #3}}
\newcommand{\sqrtnoise}{\sqrt{\noisegeneric}}

\newcommand{\bq}{\boldsymbol{q}}

\newcommand{\anticoncgap}{\hat{\lambda}}
\newcommand{\horizclose}{\delta_{\mathrm{horz}}}

\newcommand{\horizanal}{\delta_{a}} \newcommand{\horzanal}{\horizanal}

\newcommand{\taylor}{T}

\newcommand{\softcheck}{\mathrm{SoftCheck}}
\newcommand{\chk}{\textsc{Check}}
\newcommand{\st}{s}
\newcommand{\bst}{\bs}

\newcommand{\nderiv}[3]{\mathop{\triangle_{#1,#2}^{(#3)}}}
\newcommand{\zoom}[3]{#1_{#2 \mid #3}}
\newcommand{\HV}{\mathbf{HyperVar}} \newcommand{\HyperVar}{\HV} \newcommand{\HVar}{\HV}

\newcommand{\sign}{\mathrm{sign}}

\newcommand{\W}{\mathrm{W}}

\newcommand{\U}{\mathrm{U}}

\newcommand{\bup}{{\boldsymbol{\upsilon}}}

\newcommand{\Stab}{\mathrm{Stab}}

\newcommand{\maxcol}{D}

\newcommand{\Prx}{\mathop{{\bf Pr}\/}}

\newcommand{\p}{p}

\newcommand{\bits}{\{-1,1\}}
\newcommand{\bn}{\bits^n}

\renewcommand{\eps}{\varepsilon}
\renewcommand{\epsilon}{\varepsilon}
\usepackage{mathrsfs}

\def\colorful{0}

\ifnum\colorful=1
\newcommand{\violet}[1]{{\color{violet}{#1}}}

\newcommand{\red}[1]{{\color{red} {#1}}}

\newcommand{\gray}[1]{{\color{gray}{#1}}}

\fi
\ifnum\colorful=0
\newcommand{\violet}[1]{{{#1}}}

\newcommand{\red}[1]{{{#1}}}

\newcommand{\gray}[1]{{{#1}}}

\fi

\begin{document}

\title{Fooling Gaussian PTFs via Local Hyperconcentration\footnote{A preliminary version of this paper~\cite{OST20} appeared in the proceedings of the 52nd Annual ACM Symposium on Theory of Computing (STOC 2020).}\vspace*{15pt}}


\author{\hspace{-30pt}Ryan O'Donnell \vspace{3pt}\\
\hspace{-38pt} \small{\sl Carnegie Mellon University} \and \hspace{5pt}Rocco A. Servedio \vspace{3pt} \\ \small{\sl Columbia University} \vspace*{25pt} \\ with an appendix by \vspace{15pt} \\
Daniel Kane \vspace{3pt} \\
\small{\sl University of California, San Diego} \vspace*{15pt}
\and Li-Yang Tan \vspace{3pt}  \\ \small{\sl Stanford University}\vspace*{15pt}
}

\date{\small{\today}}
\maketitle



\begin{abstract}
We give a pseudorandom generator that fools degree-$d$ polynomial threshold functions over $n$-dimensional Gaussian space with seed length
$\violet{\poly(d)}\cdot \log n$.   All previous generators had a seed length with at least a $2^d$ dependence on $d$.

The key new ingredient  is \violet{a} {\sl Local Hyperconcentration Theorem}, which shows that every degree-$d$ Gaussian polynomial is hyperconcentrated almost everywhere at scale \violet{$d^{-O(1)}$}.
\end{abstract}

\thispagestyle{empty}
\newpage 
\thispagestyle{empty}

\tableofcontents










\thispagestyle{empty}
\newpage 

\setcounter{page}{1}


\section{Introduction}

This paper is about pseudorandom generators (PRGs) for polynomial threshold functions (PTFs) over Gaussian space.  Let us explain what this means.  Let $\mathscr{C}$ be a class of functions from $\R^n$ to $\R.$  A distribution $\mathcal{G}$ over $\R^n$ is an \emph{$\epsilon$-PRG for $\mathscr{C}$ over Gaussian space} if for every function $f \in \mathscr{C}$,
\[
    \abs*{\Ex_{\bz\sim\calG}[f(\bz)] - \Ex_{\bx\sim\normal(0,1)^n}[f(\bx)] } \le \eps,
\]
where $\normal(0,1)^n$ is the standard $n$-dimensional Gaussian distribution.   We equivalently say that $\calG$ \emph{$\eps$-fools $\mathscr{C}$ over Gaussian space}. If a draw $\bz \sim \calG$ can be deterministically generated from a source of $s$~independent uniformly random bits, we say that the \emph{seed length} of $\bz$ is~$s$. If furthermore the generation can be performed by a computationally efficient  algorithm, we say the PRG is \emph{explicit}.


 A \emph{degree-$d$ polynomial threshold function} (PTF) is a function
$f(x) = \sign(p(x))$ where $p : \R^n \to \R$ is a real polynomial of total degree at most $d$.  Now we can state the main theorem of this paper:

\begin{theorem}
\label{thm:main}
For all $n, d \in \N$ and $\eps \in (0,1)$, there is an explicit PRG with seed length $\violet{\poly(d/\eps)}\cdot \log n$
that $\eps$-fools the class of all degree-$d$ PTFs over $n$-dimensional Gaussian space.
\end{theorem}

The \violet{polynomial} dependence on $d$ here is a substantial improvement over previous PRGs, all of which had at least $2^{O(d)}$ dependence or worse.  We view this as notable, as there are few prior works concerning structural properties of $n$-dimensional Gaussian or Boolean PTFs that are nontrivial for $d \gg \log n$.


%


\subsection{Prior work}
\label{sec:prior}

There has been significant work on PRGs for PTFs.   Their study was initiated by Meka and Zuckerman~\cite{MZstoc10,MZ13}, who gave a PRG with seed length\footnote{They state $O(1/\eps)^{O(d)} \cdot \log n$ just after \cite[Thm.~5.18]{MZ13}, but they have appear to have dropped a factor of~$d$ when citing their Thm.~5.2 at the end of Lem.~5.20's proof. Correcting this leads to the seed length $(d/\eps)^{O(d)} \cdot \log n$.} $(d/\eps)^{O(d)} \cdot \log n$ that fools degree-$d$ PTFs over the more general setting of \emph{Boolean space}, $\{\pm 1\}^n$.  PRGs over Boolean space can be shown to also yield PRGs over Gaussian space, thanks to the fact that $\bx_1 + \cdots + \bx_m$ has a nearly Gaussian distribution when $\bx \sim \{\pm 1\}^m$ is uniformly random (see the discussion in \Cref{sec:bits-gaussians}), and the fact that degree-$d$ PTFs are closed under taking linear combinations of inputs.  Since the work of~\cite{MZstoc10,MZ13}, there have been several works that focus just on fooling PTFs over Gaussian space, which we now discuss.

First, Kane~\cite{Kane11ccc} showed that limited independence (see \Cref{def:wise-independence}) suffices to fool Gaussian PTFs.  The amount of independence required was $O_d(\eps^{-2^{O(d)}})$, which translates into $O_d(\eps^{-2^{O(d)}}) \cdot \log n$ in seed length.  Using a different generator (one that is not based only on limited independence), Kane~\cite{kane11focs} then gave a PRG for Gaussian PTFs with seed length $2^{O(d)} \cdot \poly(1/\eps)\cdot \log n$.  Note that this seed length strictly improves upon that in~\cite{MZ13}, albeit only in the Gaussian setting.

Towards further improving the seed length dependence on~$\eps$, Kane~\cite{Kane12} gave a PRG with seed length $A(d,1/c)\cdot (1/\eps)^{2+c} \cdot \log n$ for any $c > 0$, where $A(\cdot,\cdot)$ is a variant of the Ackermann function.\footnote{In fact, it seems that correcting a typo in \cite[Proof~of~Prop.~12]{Kane11ccc}, where a ``$\theta$'' factor should be ``$\theta^2$'', already leads to seed length $2^{O_c(d)} \cdot (1/\eps)^{2+c} \cdot \log n$.} This was  improved to $A(d,1/c)\cdot (1/\eps)^c \cdot \log n$ in~\cite{Kane14-subpoly}; while the seed length now has subpolynomial dependence on $1/\eps$, its dependence on $d$ limits its applicability to PTFs of constant (or very slightly superconstant) degree.

For degree-$2$ PTFs, Kane gives a PRG with seed length $O(\log^6(1/\eps) \log\log(n/\eps)\log n)$~\cite{Kane15}; Diakonikolas, Kane, and Nelson \cite{DKN10} showed that $\tilde{O}(1/\eps^9)$-wise independence suffices to fool degree-$2$ PTFs over both Boolean and Gaussian space.  For degree-$1$ PTFs (i.e.~halfspaces), the current best PRG is due to Kothari and Meka \cite{KothariMeka15}, who achieve a near-optimal seed length of $O(\log(1/\eps)\log\log(1/\eps) + \log n)$.

Summarizing the prior state of the art, previous PRGs were either specific to $d = 1,2$, or else had seed length with at least an exponential dependence on $d$.   Consequently, there were no PRGs that could fool PTFs of degree $d = \log n$, even just to constant accuracy~$\eps$.   \Cref{thm:main} therefore represents the first PRG that is able to fool PTFs of degree $d \ge \log n$; our seed length remains nontrivial for $d$ as large as \violet{$n^{\Omega(1)}$}. Please see~\Cref{prior}.

\vspace{5pt}
\begin{table}[h!]
  \captionsetup{width=.95\linewidth}
\renewcommand{\arraystretch}{1.7}
\centering
\begin{tabular}{|c|l|l|}
\hline
  Reference   & \multicolumn{1}{c|}{Seed length}  & \multicolumn{1}{c|}{Allowable / nontrivial  range of $d$'s}\\ \hline
\cite{DKN10} & $\tilde{O}(1/\eps^9) \cdot \log n$  & $d \le 2$ \\ [.2em] \hline
\cite{MZ13,MZstoc10} & $(d/\eps)^{O(d)} \cdot \log n$  & $d \le O(\log n / \log \log n)$ \\ [.2em] \hline
\cite{Kane11ccc} &
$O_d\big(\eps^{-2^{O(d)}}\big)\cdot \log n$\quad &
 $d \le \text{slightly superconstant}$\\ [.2em] \hline
\cite{kane11focs} &
$2^{O(d)}\cdot \poly(1/\eps)\cdot \log n$ &
 $d \le O(\log n)$ \\ [.2em] \hline
\cite{Kane12}  &
$A(d,\frac1{c})\cdot (1/\eps)^{2+c} \cdot \log n$ \text{ for any $c > 0$}  & $d \le \text{slightly superconstant}$~ \\ [.2em]  \hline
\cite{Kane14-subpoly}  &
$A(d,\frac1{c})\cdot (1/\eps)^c \cdot \log n$ \text{ for any $c > 0$}  & $d \le \text{slightly superconstant}$ \\ [.2em] \hline
\cite{Kane15} &
$O(\log^6(1/\eps) \log\log(n/\eps)\log n)$ & $d \le 2$ \\ [.2em] \hline
\cite{KothariMeka15} &
$O(\log(1/\eps)\log\log(1/\eps) + \log n)$ & $d=1$  \\ [.2em]  \hline \hline
  {\bf This work}  & $\violet{\poly(d/\eps)} \cdot \log n$  & $d\le \violet{n^{\Omega(1)}}$ \\ [.2em] \hline
\end{tabular}
\caption{Our work and the prior results on fooling degree-$d$ Gaussian PTFs.  The last column indicates the range of values of $d$'s for which the seed length of the corresponding PRG is nontrivial (i.e.~$o(n)$). The generators of \cite{MZ13,DKN10} work for the more general setting of Boolean space, and~\cite{DKN10,Kane11ccc}'s analyses show that limited independence suffices.}
\label{prior}
\vspace{-7pt}
\end{table}

\subsection{Motivations}

\paragraph{Geometric content.}  We now give a geometric perspective on the problem of constructing PRGs for Gaussian PTFs.  Suppose one is given a set~$F \subseteq \R^n$ and one wishes to approximately compute its Gaussian volume, $\Pr_{\bx \sim \normal(0,1)^n}[\bx \in F]$. There is an obvious Monte Carlo approach: picking $O(1/\eps^2)$ Gaussian vectors~$\bx$ at random and outputting the fraction that fall into~$F$ will, with high probability, give an $\eps$-accurate estimate.    Our question is to what extent randomness is necessary for this problem.

The extent to which derandomization is possible depends on the ``complexity'' of the sets~$F$ we allow.  If $F$ is only given via a black-box membership oracle then no derandomization is possible.  So we need to assume an ``explicit description'' of $F$ is given, and in this paper we focus on the case that $F$ is the set of points satisfying a polynomial inequality of degree at most~$d$ (i.e., $F$ is the set indicated by a degree-$d$ PTF).  Thus the $d = 1$ case allows halfspaces, the $d = 2$ case allows ellipsoids and hyperboloids, etc.  
For illustrative purposes, \Cref{fig:ptf-eg} shows an example with $n = 2$ and $d = 5$, although we generally think of $n \gg d$.

\begin{figure}
  \captionsetup{width=.95\linewidth}
\centering
\includegraphics[width=.4\textwidth]{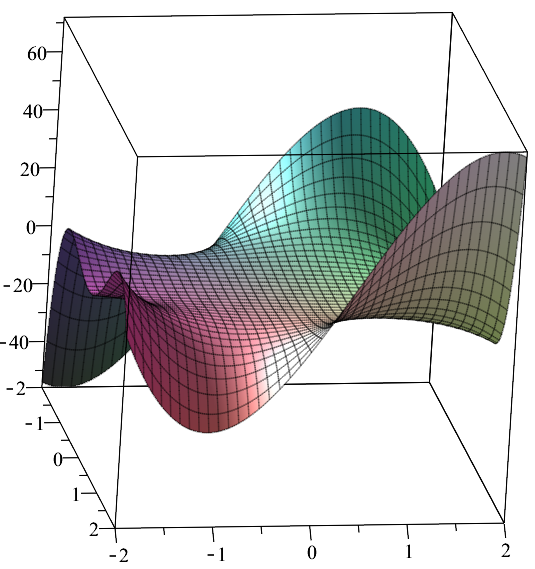} \quad \includegraphics[width=.4\textwidth]{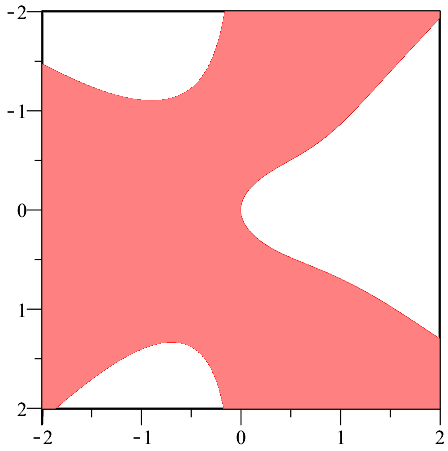}
\caption{On the left, a plot of the degree-$5$ polynomial $p(x,y) = x^2y^3-2x^3y^2+3xy^4-x+2y^2$.  On the right, the threshold set $F = \{(x,y) \in \R^2 : p(x,y) \geq 0\}$.}
 \label{fig:ptf-eg}
\end{figure}

One natural approach to this  volume-approximation problem is the following:  First, define some kind of explicit (nonrandom) finite ``grid'' of discrete points in~$\R^n$; second, show that the Gaussian volume of any degree-$d$ PTF set~$F$ is closely approximated by the fraction of grid points in~$F$.  A naive gridding scheme would use at least an exponential-in-$n$ number of grid points (even for $d = O(1)$); the question is whether we can use a subexponential-in-$n$ number of gridpoints, when $d \ll n$.  Our \Cref{thm:main} provides such a solution; by enumerating all seeds (essentially, taking the support of~$\calG$), we get an explicit set of just \violet{$n^{\poly(d)}$} ``grid points'' that gives a high-quality volume approximation for any degree-$d$ polynomial threshold set; this is nontrivial for $d$ up to some $\violet{n^{\Omega(1)}}$.  Also note that this kind of ``PRG solution'' is stronger than just being an ``volume-approximation'' algorithm of the type ``given~$F$, approximate $\text{vol}(F)$''; as it is PRG-based, it gives \emph{one} fixed, deterministic ``grid'' that simultaneously works to approximate the volume of all degree-$d$ polynomial threshold sets~$F$. 

\paragraph{Boolean complexity theory.} As mentioned earlier, the problem of PRGs (or deterministic volume approximation) for Gaussian polynomial threshold functions is a special case of the problem of PRGs (or approximate-counting) for \emph{Boolean} polynomial threshold functions.  This, in turn, is a very special case of the problem of derandomization for general Boolean circuits.  Recall that the $\mathsf{BPP}$ vs.~$\mathsf{P}$ problem is roughly equivalent to asking whether there is a deterministic polynomial-time algorithm that, given the explicit description of a subset $F \subseteq \{0,1\}^n$ in the form of a $\poly(n)$-gate Boolean circuit~$C$ computing the indicator function of $F$, computes a $0.1$-accurate approximation to its ``volume'', $\Pr_{\bx \sim \{0,1\}^n}[C(\bx) = 1]$.  Given how far we are from answering this question, the field of pseudorandomness has focused on special classes of circuits, of restricted depth and gate-types; the case of Boolean PTFs corresponds to depth-$2$ circuits~$C$ with a threshold gate on top and AND gates of width at most~$d$ at the bottom.



\subsection{Our key new tool: the Local Hyperconcentration Theorem}
For large~$d$, the best prior PRG for degree-$d$ Gaussian PTFs is Kane's~\cite{kane11focs}, which has seed length $2^{O(d)}/\poly(\eps) \cdot \log n$.  In this section we describe the most important new ingredient we introduce to Kane's framework, which lets us reduce the seed length's dependence on $d$ down to~\violet{$\poly(d)$}.  In the next section we will give an overview of the constructions of \cite{MZstoc10,MZ13,kane11focs}, putting our new tool into context.

We call our main new tool the {\sl Local Hyperconcentration Theorem}.  To explain it, suppose $p : \R^n \to \R$ is a degree-$d$ polynomial.  Since $p$ has high degree, it might fluctuate quite wildly near a given point~$x \in \R^n$, causing $\sign(p(x))$ to rapidly switch between~$\pm 1$ in small neighborhoods.  However, we might hope that for \emph{most} points~$x$, the value of~$p$ in a local neighborhood of~$x$ is almost always within a $1\pm \delta$ multiplicative factor of~$p(x)$, and hence is almost always of constant sign.

The right definition of a ``local neighborhood of~$x$'' is to choose a small scale parameter~$\lambda > 0$, and then to consider a Gaussian~$\wt{\bx}$, centered at $\sqrt{1-\lambda} x$, with variance~$\lambda$ in each coordinate.\footnote{The $\sqrt{1-\lambda}$ factor is included so that when we look at a \emph{typical} $\bx$ chosen from $\normal(0,1)^n$, the resulting ``random point in the neighborhood''~$\wt{\bx}$ also has distribution~$\normal(0,1)^n$.}  Now if $\Var[p(\wt{\bx})] \ll \E[p(\wt{\bx})]^2$, we may say that~$p$ is \emph{(multiplicatively) concentrated} in this $\lambda$-local neighborhood of~$x$; and indeed, the second moment method (Chebyshev's inequality) tells us that~$p(\wt{\bx})$ almost always has the same sign (namely, the sign of~$\E[p(\wt{\bx})]$).  The most important ingredient in Kane's work, \cite[Cor.~10+Lem.~11]{kane11focs}, 
establishes this sort of result:
\begin{theorem}[The key technical theorem of \cite{kane11focs}, simplified]
\label{thm:kane-technical}
    Let $p : \R^n \to \R$ be a degree-$d$ polynomial. Provided $\lam \leq 2^{-O(d)}$, with high probability over $\bx \sim \normal(0,1)^n$ we have
    \begin{equation} \label[ineq]{ineq:kane-var}
        \Var[p(\wt{\bx})] \ll \E[p(\wt{\bx})]^2, \quad \text{where } \wt{\bx} \sim \sqrt{1-\lambda}\bx + \normal(0,\lambda)^n.
    \end{equation}
%
%
\end{theorem}
We may say that Kane shows degree-$d$ polynomials have \emph{local concentration at scale $\lambda = 2^{-O(d)}$}, almost everywhere.  The value $L = 1/\lambda = 2^{O(d)}$ ends up becoming the dominant factor in Kane's PRG's seed length.  At a high level, this is because the PRG has the form $\bz  =  \bw_1 + \bw_2 + \cdots + \bw_{\len}$, where the $\bw_i$'s are independent random vectors with $O(d)$-wise-independent $\normal(0,\lambda)^n$ distributions.\\

By way of contrast, our new Local Hyperconcentration Theorem (stated in simplified form below) shows \emph{local \textit{\textbf{hyper}}concentration at scale $\lambda =\violet{d^{-O(1)}}$}.  For a high level sketch of the proof, see \Cref{sec:HZL-idea}.
\begin{theorem}[Simplified Local Hyperconcentration Theorem, see \Cref{thm:RZL} \violet{and~\Cref{main theorem}}]
\label{thm:local-hyperconcentration-informal}
    Let $p : \R^n \to \R$ be a degree-$d$ polynomial. Provided $\lam \leq \violet{d^{-O(1)}}$, with high probability over $\bx \sim \normal(0,1)^n$ we have
    \begin{equation}    \label[ineq]{ineq:our-hypercon}
        \HV_\hypergeneric[p(\wt{\bx})] \ll \E[p(\wt{\bx})]^2, \quad \text{where } \wt{\bx} \sim \sqrt{1-\lambda}\bx + \normal(0,\lambda)^n,
    \end{equation}
    for any large constant~$R$ (indeed, for any $R \leq \violet{\poly(d)}$).
\end{theorem}

\violet{ 
\begin{remark} 
The conference version of this paper~\cite{OST20} proved a quantitatively weaker version of the Local Hyperconcentration Theorem, showing local hyperconcentration at scale $\lambda = d^{-O(\log d)}$.  This led to a seed length of $(d/\eps)^{O(\log d)} \cdot \log n$.  Subsequently, Kane improved the Local Hyperconcentration Theorem to show local hyperconcentration at scale $\lambda = d^{-O(1)}$, which yields the current seed length of $\poly(d/\eps)\cdot \log n$.  Kane's proof is given in~\Cref{app:kane} and subsumes~\Cref{sec:anticoncentration} of this paper. 
After the initial appearance of this work on the ArXiV \cite{OSTK21arxiv-v1} we were informed by R. Meka that he and Z. Kelley had independently and concurrently obtained results similar to the main result of this work \cite{KelleyMeka21}.
\end{remark} 
}

We will define ``hypervariance'' $\HV_\hypergeneric[\cdot]$ later
(see \Cref{def:hypervar}); here we only note that it is a stronger notion than variance, in the sense that $\HV_\hypergeneric[p]$ is always at least as  large as $\Var[p]$ for all $\hypergeneric \ge 1.$
Whenever the theorem's conclusion holds for an outcome $x$ of $\bx,$ the value of $p$ in the $\lambda$-local neighborhood of $x$ is ``hyperconcentrated'' (see \Cref{lem:attenuated-hyperconcentrated}), meaning that for any large constant~$q$,
\[
    \E[|p(\wt{\bx}) - \mu|^q] \ll |\mu|^q, \quad \text{where $\mu = \E[p(\wt{\bx})]$.}
\]
The case $q = 2$ here is precisely the ``concentration'' conclusion in the theorem of \cite{kane11focs}. Our hyperconcentration is a stronger conclusion: e.g., taking $q = 4$ lets us use the ``fourth moment method'', and in fact we'll eventually use $q = 8$.

To summarize, our theorem has two important improvements over \cite{kane11focs}.\ignore{\rasnote{Here and in an number of places in the intro I did trivial rephrasings so as to refer to ``\cite{kane11focs}'' rather than ``Kane''. This is b/c with Kane's theorem in the appendix, the references to ``Kane'' I thought might be confusing/ambiguous.}}  First, it shows concentration at a much larger scale, $\lambda = \violet{d^{-O(1)}}$, rather than $2^{-O(d)}$.  This crucially gives us the \emph{potential} to get our seed's dependence on~$d$ to be $1/\lambda = \violet{\poly(d)}$.  This is far from automatic, though, because there are several other places in the \cite{kane11focs} construction that ``lose'' a factor of~$2^{O(d)}$.  In all but one of these cases\footnote{Namely, our ``noise insensitivity extension lemma'' \Cref{lem:magicker-lemma}, where we eliminate a factor of $2^{O(d)}$ from the analogous result of Kane~\cite[Cor.~16]{kane11focs}.}, it's because in \cite{kane11focs} the variance bound \Cref{ineq:kane-var} is bootstrapped using the \emph{hypercontractivity inequality}  in order to get control over~$p$'s behavior in various local neighborhoods.  This hypercontractive inequality for degree-$d$ polynomials inherently loses $2^{O(d)}$ factors (see \Cref{thm:hypercon}).  By contrast, since our theorem already establishes the stronger \emph{hyperconcentration} conclusion \Cref{ineq:our-hypercon} (this is the second key improvement, bounding hypervariance rather than variance), we are able to provide argumentation that eliminates all of these $2^{O(d)}$ factors.

\ignore{
\gray{
It should be noted that our \Cref{thm:local-hyperconcentration-informal} has one disadvantage compared to Kane's \Cref{thm:kane-technical}; the ``with high probability over~$\bx$'' in our theorem is worse than in Kane's, and this leads to our seed length's worse dependence on~$\eps$.}

\gray{ 
We conclude this section by conjecturing that our Local Hyperconcentration Theorem remains true even at the larger scale of $\noisegeneric = 1/\poly(d)$.  If this were proven, it would immediately yield PRGs for Gaussian PTFs  with seed length $\poly(d)\cdot \log n$.   In fact, it seems possible to us that our Local Hyperconcentration Theorem holds not only in the context of degree-$d$ polynomials on Gaussian space, but even in the more general context of degree-$d$ polynomials on Boolean space $\{\pm 1\}^n$.
   If this could be proven, it seems
   hopeful that our PRG would also extend to yield Boolean PRGs.}

}

\subsection{Overview of the PRG framework we use}

We use the same PRG for Gaussian PTFs as in the prior works of Meka--Zuckerman PRG~\cite{MZstoc10,MZ13} and Kane~\cite{Kane11ccc}, namely
\begin{equation}    \label{eqn:foolio}
    \bz =  \sqrt{\noisegeneric}\bz_1 + \sqrt{\noisegeneric} \bz_2 + \cdots + \sqrt{\noisegeneric} \bz_{\len},
\end{equation}
where the key parameter $\noisegeneric$ is a small function of $d$ and $\epsilon$, where $\len = 1/\noisegeneric$, and where $\bz_1, \dots, \bz_{\len}$ are independent random vectors, each having an $O(d)$-wise independent $n$-dimensional Gaussian distribution.  This leads to a seed length of essentially $O(d^2 \len \cdot \log n)$ (see \Cref{thm:kanes-prg}), and hence all the effort goes into finding the \emph{largest} $\lambda = \lambda(d,\epsilon)$ such that \Cref{eqn:foolio} $\eps$-fools degree-$d$ Gaussian PTFs.

Here we review the Meka--Zuckerman and Kane works; our own analysis is heavily based on Kane's framework.

\paragraph{Meka--Zuckerman.}  The work of Meka and Zuckerman~\cite{MZstoc10} gave PRGs for degree-$d$ Gaussian PTFs with seed length $(d/\eps)^{O(d)} \cdot \log n$.  In fact, they also extended their results to Boolean PTFs, but we do not review that extension here.  At a high level, their construction followed a basic two-part paradigm used both in the proof of Central Limit Theorems and in PRG construction: \emph{mollification + local low-degree behavior}.  To explain this, recall that we are trying to design a PRG~$\calG$ with
\[
    \abs*{\E_{\bz \sim \calG}[f(\bz)] - \E_{\bx \sim \normal(0,1)^n}[f(\bx)]} \leq \eps,
\]
where $f = \sign(p)$, with $p$ a degree-$d$ polynomial.  Suppose first that we did not have the discontinuous ``$\sign$'' function, but rather we just wanted the above inequality for $f = p$.  In that case, it would suffice for the components of the random vector $\bz \sim \calG$ to be ``$d$-wise independent,'' and in fact this would achieve $\eps = 0$.  Furthermore, there are standard techniques to produce an appropriate ``$d$-wise independent'' $\calG$ with seed length $\poly(d) \cdot \log n$, which would be an excellent bound for us.

Of course, when we return to the actual scenario of $f = \sign(p)$, the function $f$ is not even a polynomial, let alone a low-degree one.  The \emph{mollification} portion of Meka and Zuckerman's work is to replace the $\sign$ function with a smooth approximator~$\Phi$, which is equal to $\sign$ outside some interval $[-\lambda, \lambda]$.  Because the $\sign$ function is scale-invariant ($\sign(ty) = \sign(y)$ for $t > 0$), we may normalize $p$ so that its variance $\Var[p(\bx)]$ is~$1$.  Then one chooses the parameter $\lambda = \Theta(\eps/d)^d$.  The smooth mollifier $\Phi$ will have derivatives of all orders, with $\taylor$th derivative $\Phi^{(\taylor)}$ bounded in magnitude by $O(1/\lambda^\taylor)$.  The replacement of~$\sign$ by~$\Phi$ leads to a \emph{mollication error} of $O(\eps)$, essentially due to the well-known anticoncentration bound for degree-$d$ Gaussian polynomials due to Carbery and Wright~\cite{CW:01}: $\Pr[|p(\bx)| \leq (\eps/d)^d] \leq O(\eps)$.  (Note also that thanks to a trick, this only needs to hold for~$\bx$, and not the pseudorandom~$\bz$.)  With the mollifier in place, Meka and Zuckerman can try to bound
\[
    \abs*{\E_{\bz \sim \calG}[\Phi(p(\bz))] - \E_{\bx \sim \normal(0,1)^n}[\Phi(f(\bx))]} \leq O(\eps).
\]
Now although $\Phi$ is not a polynomial, it is ``locally a low-degree polynomial'' (say, of degree~$4$), thanks to Taylor's theorem.  The error in this statement scales like the $4$th derivative bound $\|\Phi^{(4)}\|_\infty \leq \poly(1/\lambda)$, times the ``locality scale''.  Thus as long as we substitute $O(1)$-wise independent Gaussians for true Gaussians at a ``scale'' of $\lambda^{\Theta(1)}$, we will not incur more than $O(\eps)$ error.
This sort of argumentation allows Meka and Zuckerman to show that the PRG in \Cref{eqn:foolio} $\eps$-fools degree-$d$ PTFs with  $\lambda = \Theta(\eps/d)^d$, which leads to their seed length of $(d/\eps)^{O(d)} \cdot \log n$.

\paragraph{Kane.}  To repeat, our PRG analysis closely follows the structure of Kane's, which we now describe. Kane~\cite{kane11focs} shows that the PRG in \Cref{eqn:foolio} succeeds with the  improved (larger) value of $\lambda = 2^{-O(d)} \cdot \poly(\eps)$, leading to his seed length of $2^{O(d)} \cdot \poly(1/\eps) \cdot \log n$.  His ``local concentration theorem'' (\Cref{thm:kane-technical}) plays a central role in this, but he still needs to develop a complex framework (which we also employ) in order to complete the analysis.

Kane's \Cref{thm:kane-technical} allows him to begin a new strategy for designing~$\calG$; rather than mollifying the $\sign$ function and taking $p(\bx)$ as a ``black box'' random variable, Kane instead mollifies the polynomial~$p$ itself. Roughly speaking, Kane's strategy begins by replacing $p$ with $p \cdot \text{Check}_1$, where $\text{Check}_1(\wt{x})$ is a smoothed indicator function for the event that \Cref{ineq:kane-var} holds at~$\wt{x}$.  The ``with high probability over~$\bx$'' in Kane's \Cref{thm:kane-technical} is in fact probability $1-\eps$ provided $\lambda \leq 2^{-O(d)}\cdot \poly(\eps)$, and this implies that the replacement of $p$ by $p \cdot \text{Check}_1$ only incurs error~$\eps$.  Now we may hope that the construction from \Cref{eqn:foolio} will work; roughly, this requires that in a $\noisegeneric$-scale neighborhood of \emph{every} point~$x$, say $\wt{\bx} = \sqrt{1-\noisegeneric} x + \sqrt{\noisegeneric} \bw$, the function $\sign(p) \cdot \text{Check}_1$ is essentially determined by low-degree moments of~$\bw$.  There are two cases. If $x$ is well into the region where $\text{Check}_1(\wt{x})$ is $0$, then $p \cdot \text{Check}_1$ is  essentially~$0$ and $\sign(p) \cdot \text{Check}_1$ is essentially constant.  Otherwise, if $x$ is near the region where  $\text{Check}_1(\wt{x})$ is~$1$, then by definition $\Var[p(\wt{\bx})]$ is very small.  Thus $p$ is not varying very much in a neighborhood of~$x$, and Taylor's theorem will tell us that low-degree moments suffice to essentially determine $p \cdot \text{Check}_1$ in this neighborhood of~$x$.

There are two catches here.  First, the use of Taylor's theorem out to, say, degree~$4$ forces one to bound not just the expected \emph{squared} deviation of~$p$ from $|p(x)|$ in the $\noisegeneric$-neighborhood of~$\wt{x}$; it requires one to control, say, the \emph{$4$th-power} deviation.  This is where  Kane uses the standard hypercontractivity-based fact that higher-power deviations can be controlled by the $2$nd-power deviation (i.e., $\Var[p(\wt{\bx})]$) at the expense of $2^{O(d)}$ losses. Kane is losing such factor anyway, since he takes $\lambda = 2^{-O(d)} \cdot (1/\eps)$.  (This is one place where our analysis takes advantage of the local \textit{\textbf{hyper}}concentration we prove in \Cref{thm:local-hyperconcentration-informal}.)

The second catch is that Taylor's theorem needs to be applied not just to $p$ but to $\text{Check}_1$ itself.  Now $\text{Check}_1$ is concerned with the variance of $p$ in a $\sqrt{\noisegeneric}$-neighborhood of~$x$.  In order to control the Taylor error here, one needs to control the \emph{variance of the variance}!  Kane handles this by further mollifying~$p$.  He uses a \emph{generalization} of \Cref{thm:kane-technical} to show that at most points~$x$, the variance of the variance in the neighborhood of~$x$ is small.  (We must prove a similar generalization of our Local Hyperconcentration Theorem; see \Cref{thm:HZL}.) Thus $p$ can be further mollified to $p \cdot \text{Check}_1 \cdot \text{Check}_2$ at only small loss.  Now we have three cases to consider when analyzing $p \cdot \text{Check}_1 \cdot \text{Check}_2$; if $x$ is well into the region where $\text{Check}_2$ is~$0$, then the mollified function is essentially~$0$ on the $\noisegeneric$-neighborhood.  Else, the variance of the variance of~$p$ in the neighborhood is suitably small.  Next, if $x$ is well into the region where $\text{Check}_1$ is~$0$, then the mollified function is again essentially~$0$ on the neighborhood; otherwise, the variance of~$p$ in the neighborhood is suitably small.  In this third case, we are again in good shape to apply Taylor to $p$, and $\text{Check}_1$\dots but to handle Taylor error for $\text{Check}_2$, we need to introduce \emph{another} check that the variance of the variance of the variance is small.  Indeed, Kane's final mollifier needs not only this ``descending'' sequence of checks (that we will picture ``vertically''), but for technical reasons needs additional ``horizontally proliferating'' checks (which, to avoid further lengthening this description, we will not discuss here).

Luckily, all of these proliferating checks eventually ``bottom out''.  The vertically descending checks bottom out because the ``$i$-fold variance'' is a polynomial of degree $d-i$, and hence the $(d+1)$-fold variance is constantly~$0$.  The horizontally proliferating checks may eventually be terminated due to the fact that a degree-$d$ polynomial is determined by its values at~$d+1$ points.  (Actually, one needs a quantitative version of this fact.  Kane provides one involving another factor of $2^{O(d)}$; we eliminate this factor in \Cref{lem:magicker-lemma}.)

Ultimately, Kane's mollifier multiplies~$p$ by $\poly(d)$ ``$\text{Check}_{i,j}$'' functions: one needs a generalization of \Cref{thm:kane-technical} and another theorem to show that the mollification is close to~$p$ at almost all points; and, when using Taylor's theorem at~$x$, one needs a $\poly(d)$-case analysis looking at the ``deepest'' check (if any) that ``fails.''  If any check ``fails,'' then the mollified function is essentially~$0$; otherwise, if they all pass, then in the $\noisegeneric$-neighborhood of~$z$, the variance of $p$, and the variance of the variance, and the variance of the variance of the variance, etc., are all suitably small for use in Taylor's theorem.

\section{The high-level structure of our proof}
\label{sec:high-level}

Throughout this paper $\p : \R^n \to \R$ is a nonzero polynomial of degree at most~$\degree$, and we are interested in the degree-$d$ polynomial threshold function $f(x)=\sign(p(x))$. For a given $0 < \epsprg < 1$, we  determine a small value\ignore{\rasnote{Redefined $\noisefixed$ to be the right value $\parens*{{\frac \epsprg \degree}}^{O(1)}$ here, and in \Cref{sec:anticoncentration} now we say that there we give the original version of the Local Hyperconcentration Theorem and point people to the appendix for the newer/stronger version. So the body of the paper is now written using this right value (needed for the proof of  \Cref{thm:main}) $\parens*{{\frac \epsprg \degree}}^{O(1)}$ of $\noisefixed$ throughout, except in \Cref{sec:anticoncentration}. }}

\begin{equation} \label{eq:noisefixed}
    \violet{\noisefixed \coloneqq \parens*{{\frac \epsprg \degree}}^{O(1)}} 
   \ignore{\rasnote{Was ``$\noisefixed \coloneqq \parens*{{\frac \epsprg \degree}}^{O(\log \degree)}$''}}
    \end{equation}
and we also let

\begin{equation} \label{eq:len-wise}
    \len \coloneqq 1/\noisefixed, \quad \quad \quad \quad \wise \coloneqq \Theta(\degree).
\end{equation}


Our main goal is:
\begin{theorem}  [Main result:  sum of $\wise$-wise independent Gaussians fools degree-$\degree$ PTFs]                                  \label{thm:1}
    Let $\bx \sim \normal(0,1)^n$ be a standard $n$-dimensional Gaussian random vector, and let $\bz_1, \dots, \bz_{\len}$ be independent $\wise$-wise independent $n$-dimensional Gaussian random vectors. Write
    \[
        \bZ \coloneqq \sqrt{\noisefixed}\bz_1 + \sqrt{\noisefixed} \bz_2 + \cdots + \sqrt{\noisefixed} \bz_{\len}.
    \]
    Then
    \[
        \abs{\E[\sign(f(\bZ))] - \E[\sign(f(\bx))]} \leq \epsprg.
    \]
\end{theorem}

\bigskip

To prove \Cref{thm:1}, we will construct a certain function
\[
    \allsoftchecks : \R^n \to [0,1],
\]
which is a smoothed indicator function for a collection of events (related to local hyperconcentration of~$p$) that are expected to almost always occur.  We then show the following:
\begin{theorem}
[Mollification error theorem\kanenote{, analogue of Lemma~17 of \cite{kane11focs}}]
\label{thm:hard-to-soft}
\[\displaystyle
    \Pr_{\bx \sim \normal(0,1)^n}[\allsoftchecks(\bx) \neq 1] \leq \epsprg/4.
\]
\end{theorem}
We then extend the mollifier to take into account the sign of $\p$:
\begin{definition}
    Define $\Indplus : \R^n \to [0,1]$ by
    \[
        \Indplus(x) = \allsoftchecks(x) \cdot \bone[\sign(\p(x)) = 1],
    \]
    and define $\Indminus$ similarly as $ \Indminus(x) = \allsoftchecks(x) \cdot \bone[\sign(\p(x)) = -1].$

\end{definition}

The main thing we prove about $\Indplusorminus$ is the following:
\begin{theorem}
[One step of the Replacement Method\kanenote{, analogue of Lemma~19 of \cite{kane11focs}}]
                                   \label{thm:hybridstep}
    Fix any  $x \in \R^n$, and assume the $\R^n$-valued random vectors~$\bz,\bz'$ are each $\wise$-wise independent $n$-dimensional Gaussian vectors.  Then we have
    \[
    \abs*{
    \E\bracks*{\Indplus\parens*{\sqrt{1-\noisefixed}  x + \sqrt{\noisefixed}\bz}}
-
    \E\bracks*{\Indplus\parens*{\sqrt{1-\noisefixed} x + \sqrt{\noisefixed}\bz'}}
}
 \leq {\frac \epsprg {4\len}}.
    \]
 The analogous statement for~$\Indminus$ also holds.
\end{theorem}
From this, a ``Replacement Method'' argument easily yields the following:
\begin{corollary}                                     \label{cor:softerror}
    For $\bZ$ as in \Cref{thm:1} and $\bx\sim \normal(0,1)^n$ we have
   \[
        \abs{\E[\Indplus(\bZ))] - \E[\Indplus(\bx)]} \leq \epsprg/4,
   \]
    and similarly for $\Indminus$
\end{corollary}
\begin{proof}
    We may view $\bx \sim \normal(0,1)^n$ as
    \[
        \bx = \sqrt{\noisefixed}\bx_1 + \sqrt{\noisefixed} \bx_2 + \cdots + \sqrt{\noisefixed} \bx_{\len},
    \]
    where $\bx_1, \dots, \bx_{\len} \sim \normal(0,1)^n$ are independent. For $0 \leq t \leq \len$, write
    \[
        \bw_t = \sqrt{\noisefixed}\bx_1 + \cdots + \sqrt{\noisefixed}\bx_t + \sqrt{\noisefixed}\bz_{t+1} + \cdots + \sqrt{\noisefixed}\bz_\len,
    \]
    so $\bw_0 = \bZ$ and $\bw_\len = \bx$.  Thus by telescoping,
    \begin{equation}    \label[ineq]{ineq:telescope}
        \abs{\E[\Indplus(\bZ))] - \E[\Indplus(\bx)]} \leq \sum_{t = 1}^\len \abs*{\E[\Indplus(\bw_{t-1}))] - \E[\Indplus(\bw_t))]}.
    \end{equation}
    For a fixed $1 \leq t \leq \len$, if we write
    \[
        \bv = \sqrt{\tfrac{\noisefixed}{1-\noisefixed}}\bx_1 + \cdots + \sqrt{\tfrac{\noisefixed}{1-\noisefixed}}\bx_{t-1} + \sqrt{\tfrac{\noisefixed}{1-\noisefixed}}\bz_{t+1} + \cdots + \sqrt{\tfrac{\noisefixed}{1-\noisefixed}}\bz_{\len},
    \]
    then
    \begin{equation}    \label{onesteperr}
        \abs*{\E[\Indplus(\bw_{t-1}))] - \E[\Indplus(\bw_t))]}
        = \abs*{\E\bracks*{\Indplus\parens*{\sqrt{1-\noisefixed} \cdot \bv + \sqrt{\noisefixed} \bz_t}} - \E\bracks*{\Indplus\parens*{\sqrt{1-\noisefixed} \cdot \bv + \sqrt{\noisefixed} \bx_t}}}.
    \end{equation}
    Since $\bz_t$ and $\bx_t$ are each $\wise$-wise independent $n$-dimensional Gaussian vectors, \Cref{thm:hybridstep} implies that
    $
        \eqref{onesteperr} \leq \epsprg/(4\len).
    $
    Putting this into \Cref{ineq:telescope} completes the proof.
\end{proof}

With the above ingredients in place, \Cref{thm:1} follows almost immediately:
\begin{proof}[Proof of \Cref{thm:1}]
    Since $\sign(\p) \leq 1-2\Indminus$ pointwise,
    \[
        \E[\sign(\p(\bZ))] \leq \E[1-2\Indminus(\bZ)] \leq \E[1-2\Indminus(\bx)] + \epsprg/2 \leq \E[\sign(\p(\bx))] + \epsprg,
    \]
    where the second inequality is thanks to \Cref{cor:softerror} and the third is thanks to \Cref{thm:hard-to-soft}.  The reverse direction, which lower bounds $\E[\sign(\p(\bZ))]$ by $\E[\sign(\p(\bx))] - \epsprg$ using $\Indplus$, is similar.
\end{proof}

\Cref{thm:1} shows that a scaled sum of $\wise$-wise independent Gaussians fools degree-$d$ PTFs, but such a random variable is not quite the desired PRG since perfectly generating even a single Gaussian random variable 
formally requires infinitely many random bits. However, the following construction of Kane tells us that for fooling degree-$d$ Gaussian PTFs, it essentially suffices to find the least~$L$ such that they are fooled by sums of~$L$ independent $k$-wise Gaussians; then, one gets an explicit PRG with seed length $O(k L \cdot d \log n)$.

\begin{theorem}   [Section~6 of \cite{kane11focs}]  \label{thm:kanes-prg}
  Let $n, d \in \N$, $0 < \eps < 1$.  Suppose that for some~$k, \len \in \N$,  degree-$d$ Gaussian PTFs are $(\eps/2)$-fooled by $\sqrt{\noisegeneric} \bz_1 + \cdots + \sqrt{\noisegeneric} \bz_\len$, where $\noisegeneric = 1/\len$ and $\bz_1, \dots, \bz_\len$ are $k$-wise independent $n$-dimensional Gaussians.  Then there is an explicit PRG for $\eps$-fooling degree-$d$ $n$-dimensional Gaussian PTFs with seed length
    \[
        O(k \len \cdot d \log (d \len n / \eps)),
    \]
    which is simply $O(k \len \cdot d \log n)$ under the reasonable assumptions that $d, 1/\eps, \len \leq \poly(n)$.
\end{theorem}

As \cite{kane11focs} does not quite explicitly state \Cref{thm:kanes-prg}, we outline a proof in \Cref{app:omitted} for completeness.  \Cref{thm:main} follows immediately from \Cref{thm:1} and \Cref{thm:kanes-prg}.


\medskip

The remaining tasks are to define $\allsoftchecks$ and prove \Cref{thm:hard-to-soft,thm:hybridstep}.
We define $\allsoftchecks$ in \Cref{sec:defining-allsoftchecks} and prove \Cref{thm:hard-to-soft} and
\Cref{thm:hybridstep} in \Cref{sec:hard-to-soft,sec:hybridstep} respectively.


\section{Probabilistic preliminaries}

In this section we introduce notation and collect several probabilistic facts we will use.  Throughout, \textbf{boldface} is used to indicate random variables, $\normal(0,1)$ denotes the standard Gaussian (normal) distribution, and $\normal(0,1)^n$ is the associated $n$-dimensional product distribution.

\subsection{Bits, Gaussians, and $k$-wise independence} \label{sec:bits-gaussians}
Although this work is mainly concerned with Gaussian random variables, many (but not all) of the tools in it ``generalize'' to Boolean $\pm 1$ random variables. In order to illustrate this, we will  provide some definitions and notations in this section that work in both cases.  However the Boolean results are never strictly needed in this work, and the reader may prefer to ignore them and focus only on the Gaussian case.

The fact that PTFs over Boolean space generalize PTFs over Gaussian space holds because,  for large~$M$ and $\bx^{(1)}, \dots, \bx^{(n)} \sim \{\pm 1\}^M$ uniform and independent,
\begin{equation}    \label{eqn:bits-to-gaussians}
    \by = \parens*{\tfrac{1}{\sqrt{M}} \sum_{i=1}^M \bx^{(1)}_i, \dots, \tfrac{1}{\sqrt{M}} \sum_{i=1}^M \bx^{(n)}_i}
\end{equation}
is ``close'' to having an $\normal(0,1)^n$ distribution, and because a degree-$d$ polynomial $p(\by)$ is also a degree-$d$ polynomial in the $\bx^{(i)}_j$'s.  One sense of ``closeness'' here is that each $\by_i$ may be coupled with a true Gaussian $\bz_i \sim \normal(0,1)$ in such a way that $|\by_i - \bz_i| \leq \frac{1}{\sqrt{M}}$ except with probability at most $O(\frac{1}{\sqrt{M}})$.

\begin{definition} \label{def:wise-independence}
    Let $\calD$ be a probability distribution on $\R$.  We say that a random vector $\bz$ on $\R^n$ has a \emph{$k$-wise independent $\calD$~distribution} if each $\bz_i$ has distribution~$\calD$, and  for all choices of $k$~indices $1 \leq i_1, \dots, i_k \leq n$, the random variables $\bz_{i_1}, \dots, \bz_{i_k}$ are independent.  Examples include $\calD$ being the uniform distribution on $\{\pm 1\}$ (``$k$-wise independent bits'') and the main concern in this paper, $\calD$ being $\normal(0,1)$ (``$k$-wise independent Gaussians'').
\end{definition}
\begin{remark}  \label{rem:why-wise}
    The main way we use $k$-wise independence is to say that if $\bx$ is $n$-wise independent, $\bz$~is $k$-wise independent, and $p : \R^n \to \R$ is a polynomial of degree at most~$k$, then $\E[p(\bz)] = \E[p(\bx)]$.
\end{remark}

\subsection{Polynomial expansions}
We recall standard facts and notation from analysis of Boolean functions and Hermite polynomials; see, e.g.,~\cite{ODbook} for a reference, and in particular \cite[Ch.~11.2]{ODbook} for Hermite analysis.

Every function $g : \{\pm 1\}^n \to \R$ can be represented by a multilinear polynomial,
\[
    g(x) = \sum_{\alpha \in \{0,1\}^n} \wh{g}(\alpha) x^\alpha,
\]
where each $\wh{g}(\alpha) \in \R$ and we use the standard multi-index notation $x^\alpha = \prod_{i=1}^n x_i^{\alpha_i}$ and $|\alpha| = \sum_i \alpha_i$.  In ``Gaussian space'' the only functions we will ever analyze are polynomials; every degree-$d$ polynomial $g : \R^n \to \R$ can be written in \emph{Hermite polynomial decomposition} as
\[
    g(x) = \sum_{\substack{\alpha \in \N^n \\ |\alpha| \leq d}} \wh{g}(\alpha) h_\alpha(x),
\]
where each $\wh{g}(\alpha) \in \R$, and the \emph{multivariate Hermite polynomial} polynomial $h_\alpha$ is given by $h_\alpha(x) =h_{\alpha_1}(x_1) \cdots h_{\alpha_n}(x_n)$, where $h_k = \frac{1}{\sqrt{k!}} H_k$ is a normalized version of the univariate degree-$k$ ``probabilists' Hermite polynomial''~$H_k$. The multivariate Hermite polynomials $h_\alpha$ are orthonormal under $\normal(0,1)^n$. Also, in the notation of \Cref{eqn:bits-to-gaussians},
\begin{equation} \label{eqn:hermite-convergence}
    \sum_{\substack{\alpha \in \{0,1\}^n \\ |\alpha| = k}} \parens*{\tfrac{1}{\sqrt{M}}\bx^{(i)}}^\alpha \xrightarrow{M \to \infty} h_k(\bz_i), \quad \bz \sim \normal(0,1).
\end{equation}

Let $g$ denote either an $n$-variate Boolean or Gaussian polynomial. We use standard notation $\E[g]$ for its mean (that is, $\E[g(\bx)]$ for $\bx \sim \{\pm 1\}^n$ in the former case, $\bx \sim \normal(0,1)^n$ in the latter), $\|g\|_r = \E[|g(\bx)|^r]^{1/r}$ for its $r$-norm ($r \geq 1$), and $\Var[g] = \E[g^2] - \E[g]^2$ for its variance.  It holds that
\[
    \E[g] = \wh{g}(0), \quad \E[g^2] = \sum_{\alpha} \wh{g}(\alpha)^2, \quad \text{hence } \Var[g] = \sum_{\alpha \neq 0^n} \wh{g}(\alpha)^2.
\]

    We write $g^{< k} = \sum_{|\alpha| < k} \wh{g}(\alpha) h_\alpha$ for $k \in \N$, and similarly write $g^{= k}$ and $g^{\geq k}$.  We also write $\mathrm{W}^{< k}[g] = \E[(g^{< k})^2] = \sum_{|\alpha| < k} \wh{g}(\alpha)^2$ for the ``weight of $g$ below level $k$'', and similarly write $\mathrm{W}^{=k}[g]$ and $\mathrm{W}^{\geq k}[g]$.

\subsection{Noise and zooms}

A basic fact about Gaussians is that if $\bx, \by \sim \normal(0,1)^n$ are independent and $0 \leq \noisegeneric \leq 1$, then $\sqrt{1-\noisegeneric} \bx + \sqrt{\noisegeneric} \by$ is also distributed as $\normal(0,1)$.  In this work, $\noisegeneric$ typically denotes a ``small'' quantity; for fixed $x \in \R^n$ we view $\sqrt{1-\noisegeneric} x + \sqrt{\noisegeneric} \by$ as a ``$\noisegeneric$-noisy'' version of~$x$, and we view changing a polynomial~$g$'s input from $\bx \sim \normal(0,1)^n$ to $\sqrt{1-\noisegeneric} x + \sqrt{\noisegeneric} \by$ as ``zooming into $g$ at~$x$ with scale~$\noisegeneric$''.  We make a precise definition:

\begin{definition}
    For $g$ an $n$-variate Gaussian polynomial, $0 \leq \noisegeneric \leq 1$, and $x \in \R^n$, we define the function $\zoom{g}{\lambda}{x}$ by
    \[
        \zoom{g}{\lambda}{x}(y) = g\parens*{\sqrt{1-\lambda} x + \sqrt{\lambda} y}.
    \]
    The function $\zoom{g}{\lambda}{x}(y)$ is a polynomial in~$y$ of the same degree as~$g$, and we (nonstandardly) refer to it as the \emph{$\noisegeneric$-zoom of $g$ at $x$}.
\end{definition}

\begin{remark} \label{rem:zoom}
    Referring again to \Cref{eqn:bits-to-gaussians}, one may verify that a $\noisegeneric$-zoom of $g$ at a random~$x$ is the Gaussian analogue of a standard Boolean concept: a \emph{random restriction} of a function $g : \{\pm 1\}^n \to \R$ at $x \in \{\pm 1\}^n$, meaning a subfunction obtained by proceeding through each coordinate~$i$, and either fixing the $i$th input to be $x_i$ with probability~$1-\noisegeneric$, or else leaving it unfixed (``free'') with probability~$\noisegeneric$.
\end{remark}

The fact that random restrictions of a Boolean function interact well with its polynomial expansion is well known; e.g.~\cite[Prop.~4.17]{ODbook} gives a formula for the expected square of any Fourier coefficient of a Boolean function under a random restriction.  Carefully taking the ``Gaussian special case'' of this (using \Cref{eqn:hermite-convergence}) yields the below analogue for random zooms.  For completeness, we give a self-contained proof of this analogue in \Cref{app:omitted}.

\begin{proposition}   \label{prop:2}
    For $g : \R^n \to \R$ a polynomial,  $0 \leq \noisegeneric \leq 1$, and $\beta \in \N^n$,
    \[
        \Ex_{\bx \sim \normal(0,1)^n}\bracks*{\wh{\restr{g}{\noisegeneric}{\bx}}(\beta)^2} 
        = \sum_{\gamma \geq \beta} \Pr[\textnormal{Bin}(\gamma,\noisegeneric) = \beta]\,\wh{g}(\gamma)^2,
    \]
    where $\textnormal{Bin}(\gamma,\noisegeneric)$ denotes an $n$-dimensional random vector with independent components, the $j$th of which is distributed as the binomial random variable $\textnormal{Bin}(\gamma_j,\noisegeneric)$.
\end{proposition}
Summing the above proposition over all multi-indices $\beta$ of a given weight $|\beta| = m$ immediately yields the following useful corollary:
\begin{corollary}
\label{cor:attenuate0}
    For $g : \R^n \to \R$ a polynomial,  $0 \leq \noisegeneric \leq 1$, and $m \in \N$,
    \[
        \Ex_{\bx \sim \normal(0,1)^n}\bracks*{\mathrm{W}^{= m}[\restr{g}{\noisegeneric}{\bx}]} = \sum_{M} \Pr[\textnormal{Bin}(M,\noisegeneric) = m]\cdot \mathrm{W}^{= M}[g].
    \]
\end{corollary}

\subsection{Noise operator and hypercontractivity}
Considering the mean of the zoom of a polynomial leads to the ``Gaussian noise'' (or ``Ornstein--Uhlenbeck'') operator (see, e.g., \cite[Def.~11.12]{ODbook}):
\begin{definition} \label{def:gaussian-noise-operator}
    Given $0 < \rho \leq 1$, the operator $\U_\rho$ acts on Gaussian polynomials $g : \R^n \to \R$ via
    \[
        (\U_\rho g)(x) = \Ex_{\by \sim \normal(0,1)^n}\bracks*{g\parens*{\rho x + \sqrt{1-\rho^2} \by)}} = \Ex_{\by \sim \normal(0,1)^n}\bracks*{ \zoom{g}{(1-\rho^2)}{x}(\by)}.
    \]
    It is well known that $\U_\rho$ acts diagonally in the Hermite polynomial basis $(h_\alpha)_{\alpha \in \N^n}$:
    \begin{equation}    \label{eqn:U-formula}
        \U_\rho g = \sum_{\alpha \in \N^N} \rho^{|\alpha|} \wh{g}(\alpha) h_\alpha.
    \end{equation}
    In particular, if $g$ is a degree-$d$ polynomial, so too is $\U_\rho g$.

    We may also write $\U_\rho$ for the analogous Boolean noise operator (more usually denoted $\mathrm{T}_{\rho}$, see \cite[Def.~2.46]{ODbook}), definable for $g : \{\pm 1\}^n \to \R$ either through \Cref{eqn:U-formula}, or by stipulating that $\U_\rho g(x)$ is the mean of a random restriction of~$g$ at~$x$ with $\rho$-probability of fixing a coordinate.

    Finally, somewhat unusually, we will need to extend the definition of $\U_\rho$ to $\rho > 1$, which we can do via the formula \Cref{eqn:U-formula}; equivalently, by stipulating that $\U_{\rho^{-1}} = \U_{\rho}^{-1}$. For $\rho > 1$ this operator no longer has a ``probabilistic interpretation'', but it still maps degree-$d$ polynomials to degree-$d$ polynomials.
\end{definition}
\begin{remark}
    We will several times use the ``semi-group property'', $\U_{\rho_1} \U_{\rho_2} = \U_{\rho_1 \rho_2}$, which is immediate from \Cref{eqn:U-formula}.
\end{remark}

At one point in our analysis we will also need the notion of \emph{Gaussian noise stability}:
\begin{definition} \label{def:gaussian-noise-stability}
    For $g: \R^n \to \R$ and $\rho > 0$,
    \[
        \Stab_\rho[g] = \Ex_{\by \sim \normal(0,1)^n}[g(\by) \cdot \U_\rho g(\by)] = \sum_{\alpha} \rho^{|\alpha|} \wh{g}(\alpha)^2,
    \]
    where the last equality is by \Cref{eqn:U-formula} and orthonormality of Hermite polynomials.
\end{definition}

\paragraph{Hypercontractivity.}
A nontrivial and highly useful property of the Boolean/Gaussian noise operator~$\U_\rho$ is \emph{hypercontractivity} (see, e.g., \cite[Secs.~9.2,~11.1]{ODbook}):
\begin{theorem}     [$(2,q)$-hypercontractive inequality]                                 \label{thm:hypercon}
    Let $g$ be a Gaussian or Boolean polynomial.  Then $\|\U_{1/\sqrt{q-1}} g\|_q \leq \|g\|_2$ holds for any $q > 2$.
\end{theorem}
Hypercontractivity has the following consequences (see \cite[Thms.~9.22,~9.23]{ODbook}):
\begin{theorem} \label{thm:aobf-9.22}
    Let $g$ be a Gaussian or Boolean polynomial of degree at most $k$.  Then $\|g\|_2 \leq e^k \|g\|_1$.
\end{theorem}
\begin{theorem}  \label{thm:tail-bound}
    Let $g$ be a Gaussian or Boolean polynomial of degree at most~$k$.  Then for any $t \ge \sqrt{2e}^k$,
    \[
        \Prx_{\bx}[|g(\bx)| \ge t \| g\|_2] \le\exp\left(-\frac{k}{2e}t^{2/k} \right).
    \]
\end{theorem}

\subsection{Hyperconcentration: our key tool}  \label{sec:hyperconcentration}
The ideas in this section, though technically standard, are part of the conceptual contribution of this work.

Very often we will need to show that a random variable is tightly concentrated around its mean in a \emph{multiplicative} sense.  Let us start with some notation.
\begin{notation}
    We use the following notation to denote that two reals $a,b > 0$ are multiplicatively close:  For $\nu \geq 0$,
    \[
        a \approx_\nu b \quad \iff \quad e^{-\nu} \leq a/b \leq e^\nu.
    \]
    Note that this condition is indeed symmetric in $a$ and $b$.  We extend the notation to all $a,b \in \R$ by stipulating that $a \approx_\nu b$ if: $ab > 0$ and the above condition holds; or, $a = b = 0$.
\end{notation}

Given a real random variable~$\bw$ with mean~$\mu$, a standard way to show that $\bw \approx \mu$ with high probability is to first establish $\stddev[\bw] \leq \eta |\mu|$ and then use Chebyshev's inequality.  When this holds we informally say that $\bw$ \emph{concentrates} around its mean.  In this work, a crucial concept will be improving this concentration using higher norms.

\begin{definition} \label{def:hyperconcentrated}
    Let $q > 2$ and $\eta \geq 0$ be real numbers.  We say a real random variable $\bw$ with mean~$\mu$ is \emph{$(q,\eta)$-hyperconcentrated} if
    \[
        \norm{\bw - \mu}_q = \E[|\bw - \mu|^q]^{1/q} \leq \eta |\mu|.
    \]
\end{definition}
The utility of this definition is that it gives an improvement to the Chebyshev inequality:
\begin{proposition}                                     \label{prop:hypermarkov}
    Suppose $\bw$ with mean $\mu$ is $(q,\eta)$-hyperconcentrated.  Then for any $t > 0$, except with probability at most $(\eta/t)^q$ we have $|\bw - \mu| \leq t|\mu|$ (and in particular $\bw \approx_{2t} \mu$ if $t \leq 1/2$).
\end{proposition}
\begin{proof}
    Apply Markov's inequality to the random variable $|\bw - \mu|^q$.
\end{proof}
We'll also need the following simple consequence of hyperconcentration:
\begin{lemma}                                       \label{lem:hyperholder}
    Suppose $\bz$ is an $\R^m$-valued random vector with  $(q,\eta)$-hyperconcentrated components, and write $\mu = \E[\bz]$.  Then for any multi-index $\alpha \in \N^m$ with $|\alpha| \leq q$,
    \[
        \E\bracks*{|\bz - \mu|^\alpha} \leq \eta^{|\alpha|} |\mu|^{\alpha}.
    \]
\end{lemma}
\begin{proof}
We have
\begin{align*}
\E\bracks*{|\bz - \mu|^\alpha} =
\E\bracks*{\prod_{i=1}^n |\bz_i - \mu_i|^{\alpha_i}} \le
\prod_{i=1}^n \E\bracks*{|\bz_i - \mu_i|^{\alpha_i \cdot {\frac {|\alpha|}{\alpha_i}}}}^{{\frac {\alpha_i}{|\alpha|}}} \le
\prod_{i=1}^n (\eta \mu_i)^{\alpha_i} = \eta^{|\alpha|} |\mu|^{\alpha},
\end{align*}
where the first inequality is from H\"{o}lder's inequality and the second is from \Cref{def:hyperconcentrated}.
\end{proof}

The random variables we'll show hyperconcentration for will be Gaussian polynomials. We will do this by bounding a quantity that we term their ``hypervariance'', and that plays a central role in our work:

\begin{definition} \label{def:hypervar}
    Let $g$ be a Gaussian or Boolean polynomial.  Then for $\hypergeneric > 1$, we define the \emph{$\hypergeneric$-hypervariance} of~$g$ to be
    \[
        \HV_\hypergeneric[g] \coloneqq \Var[\U_\hypergeneric g] = \sum_{\alpha \neq 0} \hypergeneric^{2|\alpha|} \wh{g}(\alpha)^2.
    \]
    (For $\hypergeneric = 1$, this reduces to the usual variance of~$g$.)
\end{definition}
\begin{lemma}                                       \label{lem:hypervar-to-hypercon}
    Let $g$ be a Gaussian or Boolean polynomial.  Write $\mu = \E[g]$ and assume $\HV_{\hypergeneric}[g] \leq \theta \mu^2$.  Then the random variable $g(\bx)$ is $(1+\hypergeneric^2,\sqrt{\theta})$-hyperconcentrated.
\end{lemma}
\begin{proof}
    Writing $\ol{g} = g-\mu$, our hypothesis is that $\norm{\U_\hypergeneric \ol{g}}_2^2 \leq \theta \mu^2$.
    By hypercontractivity, we have $\norm{\ol{g}}_{1+\hypergeneric^2} = \norm{\U_{1/\hypergeneric} \U_\hypergeneric \ol{g}}_{1+\hypergeneric^2} \leq \norm{\U_\hypergeneric \ol{g}}_2$.  Thus $\norm{\ol{g}}_{1+\hypergeneric^2} \leq \sqrt{\theta} |\mu|$, as needed.
\end{proof}

The hypothesis in \Cref{lem:hypervar-to-hypercon}, that $g$'s hypervariance is small compared to its squared-mean, will be an important one for us.  It is essentially the same as the hypothesis that $g$'s hypervariance is small compared to its squared-$2$-norm (since squared-$2$-norm equals squared-mean plus variance, and hypervariance is at least variance for all $\hypergeneric >1$).  It will be slightly more convenient in our Local Hypervariance Theorem to work with the latter hypothesis, so we codify it here and establish the analogue of \Cref{lem:hypervar-to-hypercon}.

\begin{definition}  \label{def:attenuated-prelim}
    Let $g$ be a Gaussian or Boolean polynomial, and let $\hypergeneric > 1$.  We say that $g$ is \emph{$(\hypergeneric,\eps)$-attenuated} if $\HV_{\hypergeneric}[g] \leq \eps \|g\|_2^2$.
\end{definition}

\begin{remark} \label{remark:attenuated}
Intuitively, a polynomial $g$ is attenuated if for each $i \geq 1$, the amount of Hermite weight it has at level $i$ is ``very small'' compared with the total Hermite weight (squared 2-norm) of~$g$.  Crucially, the precise quantitative definition of ``very small'' in the preceding sentence depends on the weight level $i$, and gets exponentially stronger (smaller) as $i$ gets larger.  An intuition which may possibly be helpful is to think of an attenuated polynomial as a polynomial which is ``morally constant'' over Gaussian space.
\end{remark}

Returning to hyperconcentration, we have the following:

\begin{lemma}                                       \label{lem:attenuated-hyperconcentrated}
    Let $g$ be a Gaussian or Boolean polynomial that is $(\hypergeneric,\theta)$-attenuated, with $\hypergeneric \geq \sqrt{2}$ and $\theta \leq 1$.  Then the random variable $g(\bx)$ is $(1+\tfrac12 \hypergeneric^2,\sqrt{\theta})$-hyperconcentrated.
\end{lemma}
\begin{proof}
    Using the notation $\mu$ and $\ol{g}$ again, and starting with the $(\hypergeneric, \theta)$-attenuation assumption, we have
    \[
        \sum_{j \geq 1} \hypergeneric^{2j} \|g^{=j}\|_2^2  \leq \theta \parens*{\mu^2 + \sum_{j \geq 1} \|g^{=j}\|_2^2}
        \implies \sum_{j \geq 1} (\hypergeneric^{2j} - \theta) \|g^{=j}\|_2^2 \leq \theta \mu^2
        \implies \sum_{j \geq 1} \parens*{\tfrac{\hypergeneric}{\sqrt{2}} }^{2j} \|g^{=j}\|_2^2 \leq \theta \mu^2,
    \]
    where the last step used $\hypergeneric \geq \sqrt{2}$ and $\theta \leq 1$.  But this is equivalent to $\HV_{\frac{\hypergeneric}{\sqrt{2}}}[g(\bx)] \leq \theta \mu^2$, so the result follows from \Cref{lem:hypervar-to-hypercon}.
\end{proof}
Combining this with \Cref{lem:attenuated-hyperconcentrated} and \Cref{prop:hypermarkov} yields the following useful result, which informally says that ``attenuated polynomials are very likely to take values multiplicatively close to their means'':
\begin{proposition}                                     \label{prop:atten-hyperconcy}
    Let $g$ be a Gaussian or Boolean polynomial that is $(\hypergeneric,\theta)$-attenuated, with $\hypergeneric \geq \sqrt{2}$ and $\theta \leq 1$. Write $\mu = \E[g]$.  Then assuming $0 < \gamma \leq 1$, we have $g(\bx) \approx_{\gamma} \mu$ except with probability at most $(2\sqrt{\theta}/\gamma)^{\frac12 \hypergeneric^2 + 1}$.
\end{proposition}

\subsection{Special properties of Gaussian random variables} \label{sec:special}

All of the results in this section so far have applied equally well to Gaussian or Boolean polynomials.  We now give the two results we will use that are specific just to Gaussian polynomials.  The first is a well known result of Carbery and Wright~\cite{CW:01} on anticoncentration (see e.g.~\cite[Lem.~23]{kane11focs}, \cite[Sec.~11.6]{ODbook}):

\begin{theorem} [Gaussian Carbery--Wright] \label{thm:CW}  There is a univeral constant $C$ such that for any degree-$d$ polynomial $g : \R^n \to \R$  and any $0 \leq \delta \leq 1$,
    \[
        \Pr_{\bx \sim \normal(0,1)^n}\bracks*{|g(\bx)| < \parens*{\tfrac{\delta}{C \cdot d}}^d \cdot \|g\|_2} \leq \delta.
    \]
\end{theorem}
The second Gaussian-specific result we use is a key lemma from Kane's work \cite[Lemma~9]{kane11focs}.  This lemma was the essential ingredient he used to prove his ``local concentration'' result \Cref{thm:kane-technical}.  At first glance, it may look much stronger than \Cref{thm:kane-technical}, because it gives a nontrivial kind of concentration result even for $\lambda$ as large as $1/\poly(d)$.  However the concentration one gets in (almost all) local neighborhoods is somewhat weak: one gets that $\zoom{g}{\noisegeneric}{\bx}$'s values are with high probability near a specific value, but this is not enough to even conclude that $\stddev[\zoom{g}{\noisegeneric}{\bx}]$ is small compared to that value.  Kane uses hypercontractivity to bootstrap this to control over the variance when he obtains \Cref{thm:kane-technical}, and this loses a~$2^{O(d)}$ factor.  When we employ \Cref{lem:kane-lemma9-alt} below, we will already be working with hyperconcentrated functions, which means we will not lose much when similarly bootstrapping.
\begin{lemma} \label{lem:kane-lemma9-alt}
    (\cite[Lemma~9]{kane11focs} with parameters renamed.) Let $g : \R^n \to \R$ be a degree-$\degree$ polynomial and let $0 < \beta < 1$. Then for $\bx, \by \sim \normal(0,1)^n$ independent, except with probability $\beta$ we have
    \[
        \phantom{\quad \text{for } \gamma = O(\frac{d^2}{\beta})\cdot \sqrt{\noisegeneric}.}
        \zoom{g}{\noisegeneric}{\bx}(\by) \approx_{\nu} g(\bx) \quad \text{for } \nu = O(d^2/\beta)\cdot \sqrt{\noisegeneric}
    \]
    (provided $\noisegeneric$ is small enough that $\nu \leq 1$).
\end{lemma}
Kane's proof of this lemma (seemingly) crucially relies on the rotational invariance of $n$-dimensional Gaussians.


\section{Defining $\allsoftchecks$} \label{sec:defining-allsoftchecks}

The definition of $\allsoftchecks$ involves a collection $\stats$ of ``statistics'' of the polynomial $\p$.  Each statistic $\st \in \stats$ will be a certain nonnegative polynomial $\st : \R^n \to \R^{\geq 0}$, defined in terms of $\p$, of degree at most $2\degree$.

The definition also involves a collection $\MollifierChecks$ of ``mollifier checks''.  Each mollifier check~$\chk \in \MollifierChecks$ will consist of two ingredients:

\begin{equation} \label{check}
    \chk = (\ineq, \softness), \quad \text{where } \ineq=(\gamma, \st_u, \st_v) \text{ means ``}\st_u \geq \gamma \st_v\text{''}
\end{equation}
for some statistics $\st_u, \st_v \in \stats$ and some nonnegative value~$\gamma$, and where
$
    \softness \geq 0$ is a ``softness'' parameter.
The intuitive meaning of~$\chk$ applied at a point $x \in \R^n$ is that it is ``softly'' checking that $\st_u(x) \geq \gamma \st_v(x)$, up to a multiplicative factor of roughly~$1 \pm \softness$.  More precisely:
\begin{definition} \label{def:sigma}
    Let $\sigma : \R \to [0,1]$ be a smooth function satisfying
    \[
        \sigma(t) = \begin{cases}
                                0 & \text{if $t \leq -1$,} \\
                                1 & \text{if $t \geq +1$,}
                          \end{cases}
    \]
    and which is such that for all $j \leq d$, the magnitude of $\sigma$'s $j$-th derivative is everywhere bounded by~$j^{O(j)}$. (This is easily achieved by standard constructions such as taking $\sigma$ to be a suitable polynomial of degree $O(d)$ on the interval $[-1,1].$)  Also, given a mollifier check~$\chk \in \MollifierChecks$ as in~\eqref{check}, define
    \[
        \softcheck_\chk : \R^n \to [0,1]
    \]
    by
    \begin{equation} \label{eq:softcheck-definition}
        \softcheck_\chk(x) = \sigma\parens*{\softness^{-1} \ln\parens*{\frac{\phantom{\gamma} \st_u(x)}{\gamma \st_v(x)}}},
    \end{equation}
    where we take $0/0 = +\infty$.  We remark that
    \begin{align*}
        \softcheck_\chk(x) &= 1 \text{ if } \st_u(x) \geq \exp(\delta) \cdot \gamma \st_v(x), \\
        \softcheck_\chk(x) &= 0 \text{ if } \st_u(x) \leq \exp(-\delta) \cdot \gamma  \st_v(x).
    \end{align*}
\end{definition}

The function $\allsoftchecks$ is the product of all the mollifier checks:
\begin{definition} \label{def:allsoftchecks}
$\displaystyle
    \allsoftchecks(x) = \prod_{\chk \in \MollifierChecks} \softcheck_\chk(x).
$
\end{definition}

To complete the definition of $\allsoftchecks$ we need to: (i)~define the statistics in $\stats$; and, (ii)~define the collection $\MollifierChecks$ of mollifier checks.  We do each of these in turn  below.

\subsection{The statistics in $\stats$}

\subsubsection{Noisy derivatives of amplified polynomials} Now we arrive at a novel definition in this work which plays a key role in our results.  We will make extensive use of the following notion, which can be thought of as a sort of ``noisy derivative of the $R$-amplified version of $g$ at $x$ in directions $y$ and $y'$.''  This is a variant of one of the key definitions of \cite{kane11focs} (the second definition in Section~3 of that paper) but with the crucial difference that now we consider the ``$R$-amplified'' version of $g$ in place of just $g$ itself as was the case in the corresponding definition in \cite{kane11focs}:

\begin{definition} \label{def:amplified-derivative}
    Given two vectors $y, y' \in \R^n$, define the operator $\nderiv{y}{y'}{\hyperfixed,\noisefixed}$ on polynomials $g : \R^n \to \R$ via
    \[
        (\nderiv{y}{y'}{\hyperfixed,\noisefixed} g)(x) \coloneqq \frac{\U_\hyperfixed \zoom{g}{\noisefixed}{x}(y) - \U_\hyperfixed \zoom{g}{\noisefixed}{x}(y')}{\sqrt{2}}.
    \]
    We remark that $\U_\hyperfixed \zoom{g}{\noisefixed}{x}(y)$ is parenthesized as $(\U_\hyperfixed (\zoom{g}{\noisefixed}{x}))(y)$.\ignore{ and that the dependence of the notation $\nderiv{y}{y'}{\hypergeneric}$ on $\noisefixed$ is suppressed, as it will be fixed throughout.\rodnote{Honestly, this doesn't make sense; $\noisefixed$ is no more ``fixed throughout'' than $\hyperfixed$ is going to be.  We actually hardly use this notation $\triangle$ that much, so maybe we should cram both params in there.{\red{{\bf Rocco:} Sure, let's cram it in, give the reader a real treat.}}}}
\end{definition}

The following is easily verified:

\begin{fact}                                        \label{fact:deriv-deg}
    If $g : \R^n \to \R$ is a polynomial of degree at most~$d$, then for every $y,y' \in \R^n$, the function $\nderiv{y}{y'}{\hyperfixed,\noisefixed} g$ is a polynomial of degree at most~$d-1$.
\end{fact}

The following simple but crucial fact connects the derivative notion from \Cref{def:amplified-derivative} to the hypervariance notion from \Cref{def:hypervar}:

\begin{fact}                                        \label{fact:deriv-hypervar}
    For fixed $x \in \R^n$ and independent $n$-dimensional Gaussians $\by, \by' \sim \normal(0,1)^n$, we have that
    \[
        \Ex_{\by,\by'}\bracks*{\parens*{\nderiv{\by}{\by'}{\hyperfixed,\noisefixed} g(x)}^2} = \HVar_\hyperfixed\bracks*{\zoom{g}{\noisefixed}{x}}.
    \]
\end{fact}

\begin{proof}
The left-hand side is equal to
\begin{align*}
&{\frac 1 2} \E \bracks*{
\parens*{
\sum_{\alpha} \hyperfixed^{|\alpha|} \wh{\zoom{g}{\noisefixed}{x}}(\alpha) \left(
h_\alpha(\by)-h_\alpha(\by')
\right)
}^2
}\\
&= {\frac 1 2} \E \bracks*{
\parens*{
\sum_{\alpha \neq 0} \hyperfixed^{|\alpha|} \wh{\zoom{g}{\noisefixed}{x}}(\alpha) \left(
h_\alpha(\by)-h_\alpha(\by')
\right)
}^2
}\tag{since $h_{\alpha}(y) \equiv 1$ for $\alpha=0$}\\
&= {\frac 1 2}
\sum_{\alpha,\beta \neq 0}
\hyperfixed^{|\alpha|+|\beta|}
\wh{\zoom{g}{\noisefixed}{x}}(\alpha)
\wh{\zoom{g}{\noisefixed}{x}}(\beta)
\E  \bracks*{
(h_\alpha(\by)-h_\alpha(\by'))
(h_\beta(\by)-h_\beta(\by'))
}\\
& = \sum_{|\alpha|>0} \hyperfixed^{2|\alpha|} \wh{\zoom{g}{\noisefixed}{x}}(\alpha)^2
        =  \HV_\hyperfixed[\zoom{g}{\noisefixed}{x}]
\end{align*}
where the penultimate equality is by orthonormality of the Hermite polynomials and independence of $\by,\by'$.
\end{proof}

\subsubsection{The statistics in $\stats$} \label{sec:the-statistics}

Fix a parameter
\begin{equation} \label{eq:hyperfixed}
\hyperfixed \coloneqq \Theta(\taylor)
\end{equation}
where $\taylor$ is a parameter we will set later.  (See \Cref{eq:taylor}; it will in fact be an absolute constant.)
We can now define the set of statistics,~$\stats$.  Each statistic is doubly indexed by a pair of natural numbers; there are $(\degree+1)(\maxcol+1)$ many statistics, $\{\st_{i,j}\}_{i \in \{0,\dots,\degree\}, j \in \{0,\dots,\maxcol\}},$ where
\begin{equation} \label{eq:maxcol}
\maxcol \coloneqq (2\degree+1)^2.
\end{equation}
It is convenient for us to view the elements of $\stats$ as being arranged in a grid where the $(i,j)$-th statistic $\st_{i,j}$ is in row $i$ and column $j$ (we will often use terminology of this sort).  We remark that our statistic $\st_{i,j}$ will closely correspond to the functions called $|p_\theta^{(\ell),m} (X)|_2^2$ in~\cite{kane11focs} ($i = \ell, j = m$), except that as mentioned we use \emph{amplified} noisy derivatives where Kane just had noisy derivatives.

All statistics are defined in terms of the underlying degree-$\degree$ polynomial $\p$.
We first define the 0$^{th}$ column of statistics:
\begin{definition} [0$^{th}$ column of statistics] \label{def:si0}
    For $i \in \{0,\dots,\degree\}$, we define
    \[
        \st_{i,0}(x) \coloneqq \E\bracks*{\parens*{\nderiv{\by_i}{\by_i'}{\hyperfixed,\noisefixed} \nderiv{\by_{i-1}}{\by_{i-1}'}{\hyperfixed,\noisefixed} \cdots \nderiv{\by_1}{\by_1'}{\hyperfixed,\noisefixed} \p(x)}^2},
    \]
    where $\by_1, \by'_1, \dots, \by_i, \by'_i \sim \normal(0,1)^n$ are independent.
\end{definition}
\begin{remark}
    By \Cref{fact:deriv-deg}, $\st_{\degree,0}$ is a constant function and $\st_{\degree+1,0}$ is identically zero; this is why we consider $\st_{i,0}$ only for $i \leq \degree$.
\end{remark}

For  technical reasons, we will also need to use slight variants of the statistics $\st_{i,0}$, which correspond to taking an average over mildly noisy versions of the input:

\begin{definition}   [The remaining statistics] \label{def:sij}
  For $i \in \{0,\dots,\degree\}$ and $j \in \{1,\dots,\maxcol\}$ we define  $\st_{i,j}$ by
    \[
        \st_{i,j}(x) \coloneqq \U_{\sqrt{1-\noisefixed}} \st_{i,j-1}(x) = \E_{\by \sim \normal(0,1)^n} \bracks*{\zoom{(\st_{i,j-1})}{\noisefixed}{x}(\by)}.
    \]
 \end{definition}

\begin{remark} \label{rem:semigroup} Using the semigroup property, we have that
\[
  \st_{i,j}(x) = \U_{(1-\noisefixed)^{j/2}} \st_{i,0}(x) =
  \E_{\by \sim \normal(0,1)^n} \bracks*{\zoom{(\st_{i,0})}{1 - (1-\noisefixed)^j}{x}(\by)}.
\]
\end{remark}

This completes the formal definition of the statistics in $\stats$; however, it will be very useful for us to view the statistics from a different perspective based on distributions of polynomials.  We introduce this perspective in the next subsection.

\subsection{A distributional view on the statistics} \label{sec:distributional}

For the sake of probabilistic technicalities, we will need to define a notion of a distribution of polynomials being  ``nice'':

\begin{definition} \label{def:nice}
We say that a distribution $(g_{\bup})_{\bup \sim \Upsilon}$ of $n$-variable polynomials of degree at most $d$ is \emph{nice} if
\begin{itemize}

\item $\Upsilon$ is the normal distribution $\normal(0,1)^I$ for some natural number $I$, and

\item for each $\upsilon$ in the support of $\Upsilon$, the coefficients of $p_\upsilon$ are polynomials of degree at most~$d$ in $\upsilon.$
\end{itemize}
\end{definition}

It will be convenient for us to view the statistics $\st_{i,j}$ as averages of squares of polynomials drawn from various nice distributions. To do this, we inductively define a grid of nice distributions of polynomials ${\cal F}_{i,j}$ as follows.  The base distribution, $\calF_{0,0}$ is just the probability distribution with a single outcome, namely the polynomial~$\p$. (Note that this corresponds to a nice distribution with $I=0.$)  Next, we inductively define the distributions $\calF_{i,0}$ as follows:

\begin{itemize}

\item {\bf To make a draw from $\calF_{i,0}$, where $i > 0$:}  First draw $\boldf \sim \calF_{i-1,0}$.
 Then draw $\by, \by' \sim \normal(0,1)^n$.
 Then output the polynomial $\nderiv{\by}{\by'}{\hyperfixed,\noisefixed} \boldf$.

\end{itemize}
\newcommand{\opnoise}{\mathrm{N}}
The following ``operator notation'' for zooms (local to this section and \Cref{sec:sick}) will be convenient: for a polynomial~$p$ and a vector $y \in \R^n$, we define the notation
\begin{equation} \label{eq:opnoise}
\opnoise_{y} p \coloneqq \zoom{p}{1-\noisefixed}{y}.
\end{equation}

With the distributions of polynomials $\calF_{i,0}$ defined as above, we inductively define the distributions
 $\calF_{i,j}$, where $j > 0$ as follows:

\begin{itemize}

\item {\bf To make a draw from $\calF_{i,j}$, where $j > 0$:}
First draw $\boldf \sim \calF_{i,j-1}$.
 Then draw $\by \sim \normal(0,1)^n$.
  Then output the polynomial $\opnoise_{\by} \boldf$.

\end{itemize}

It is immediate from these definitions that each ${\cal F}_{i,j}$ is a nice distribution of polynomials of degree at most $d$.  It is also immediate, comparing the above definition against \Cref{def:si0} and \Cref{def:sij}, that for each $i \in \{0,\dots,d\}$ and each $j \in \{0,\dots,\maxcol\}$ we have that
\begin{equation} \label{eq:alt-sij}
    \st_{i,j} = \E_{\boldf \sim \calF_{i,j}}[\boldf^2].
\end{equation}
Finally, it is straightforward to check from these definitions (using also \Cref{fact:deriv-hypervar}) that
\begin{equation} \label{eq:stats-as-zooms}
    \st_{i+1,0}(x)  = \E_{\boldf \sim \calF_{i,0}}[\HV_\hyperfixed[\zoom{\boldf}{\noisefixed}{x}]],
    \qquad
    \st_{i,j+1}(x)  = \E_{\boldf \sim \calF_{i,j}}[\|\zoom{\boldf}{\noisefixed}{x}\|_2^2].
\end{equation}

These characterizations will be useful when we analyze the statistics later.


\ignore{
%
%
%
}

\subsection{Defining the mollifier checks} \label{sec:MollifierChecks-def}

\noindent {\bf Intuition.}  In this subsection we define the collection $\MollifierChecks$ of mollifier checks. Before formally defining these checks, we give some useful intuition concerning them. We will show in \Cref{sec:hard-to-soft} that except with very small failure probability over $\bx \sim \normal(0,1)^n$, the statistics $\st_{i,j}(\bx)$ satisfy the following properties, where $\anticoncgap,\horizclose>0$ are suitable small parameters:

\begin{enumerate}

\item {\bf Local hyperconcentration:} For each $i \in \{0,\dots,\degree-1\}$,
\begin{equation}    \label[ineq]{ineq:anticoncgap}
 \st_{i+1,0}(\bx) \leq \anticoncgap \st_{i,1}(\bx).
\end{equation}

\item {\bf Insensitivity under noise:} For each $i \in \{0,\dots,\degree\}, j \in \{0,\dots,\maxcol-{2}\}$,
\begin{equation}    \label[ineq]{ineq:noise-insens}
 \st_{i,j}(\bx) \approx_{\horizclose} \st_{i,j+1}(\bx).
\end{equation}

\end{enumerate}

The parameter settings we require will turn out to be the following:
\begin{equation}  \label{eq:anticoncgap-horizclose}
\anticoncgap \text{ satisfying } \poly(\taylor d)^\taylor \cdot \anticoncgap^{\taylor/2} = \noisefixed \epsprg,
\quad \quad \quad \quad
\horizclose \coloneqq {\frac 1 {K \degree \maxcol}},
\end{equation}
where $K$ is a suitably large absolute constant and  $\taylor > 2$ is a constant that will be set later in \Cref{eq:taylor}  (for now, the most important thing to notice is that since $\taylor > 2$, the exponent on $\anticoncgap$ above is strictly greater than~$1$).
\ignore{\rasnote{See the ``CHECK/ENFORCE'' rasnotes in \Cref{sec:anticoncentration} and \Cref{sec:noise-insensitivity} for constraints that we'll need $\anticoncgap, \horizclose$ and $\noisefixed$ to satisfy together}
}

The mollifier checks are designed precisely to check that the above properties \Cref{ineq:anticoncgap,ineq:noise-insens} actually hold at~$x$, and thus \Cref{thm:hard-to-soft} corresponds to the fact that these properties hold with high probability for a random $\bx \sim \normal(0,1)^n$.

With the above intuition in place, we proceed to define the mollifier checks that check for each of the above types of properties (1) and (2).
Please see below for a figure depicting the mollifier checks and the grid of statistics.
\myfig{.75}{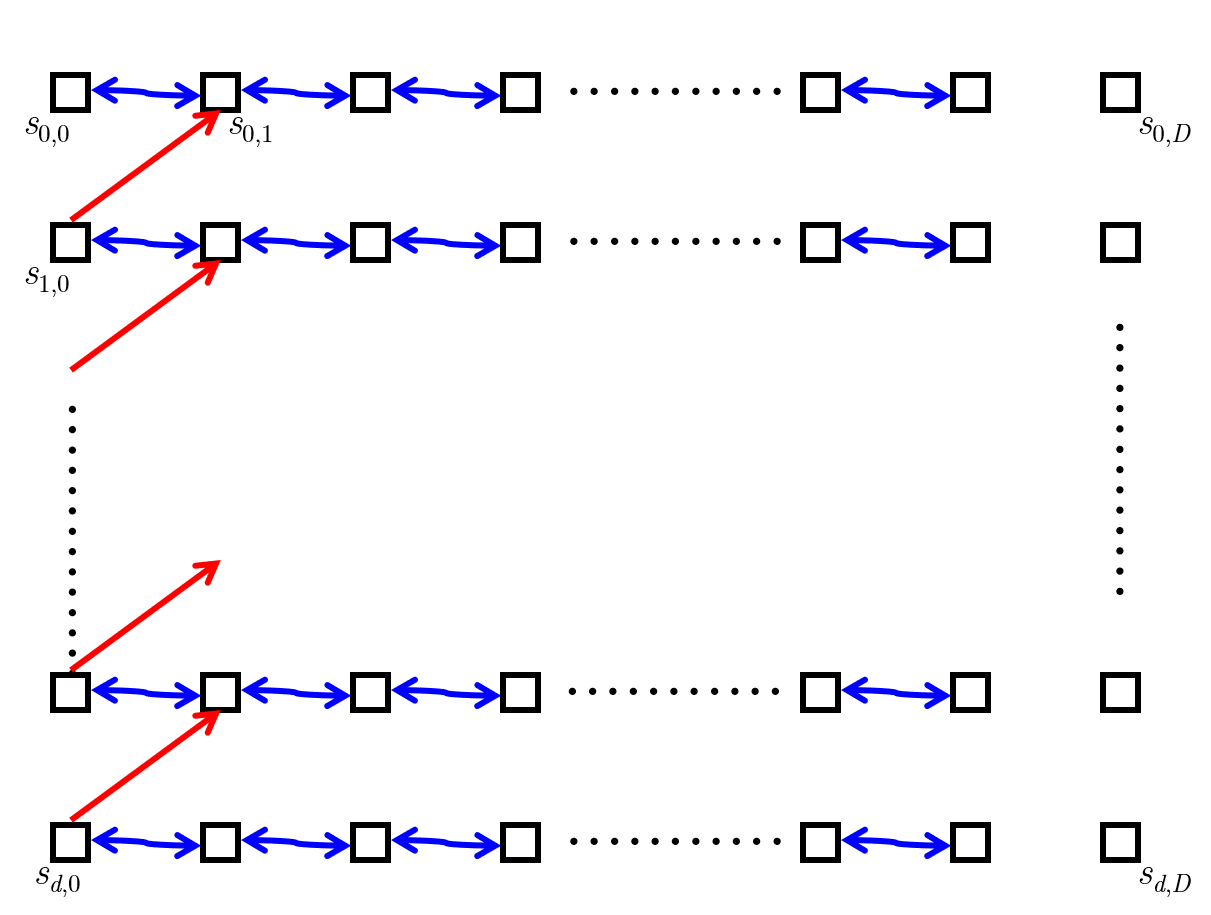}{The grid of statistics and the mollifier checks. Local hyperconcentration checks are depicted in red and noise insensitivity checks are depicted in blue.}{}

\subsubsection{Checking local hyperconcentration} \label{sec:checking-local-hyperconcentration}
For each $i \in \{0,\dots,\degree-1\}$, $\MollifierChecks$ contains a check corresponding to \Cref{ineq:anticoncgap}. The ``inequality'' portion of the check is
\[
\st_{i+1,0} \leq e\anticoncgap \st_{i,1}
\]
and the ``softness'' parameter of the check is $1$; so for this check $\chk$ in $\MollifierChecks,$ the associated ``soft check'' is
\begin{equation} \label{eq:diagonal-check}
\softcheck_\chk(x) = \sigma
\parens*{ \ln\parens*{\frac{e\anticoncgap \st_{i,1}}{\st_{i+1,0}}}}.
\end{equation}
We refer to these $\degree$ elements of $\MollifierChecks$ as ``anticoncentration checks'' or (recalling the grid) as ``diagonal checks.''

\subsubsection{Checking insensitivity under noise} \label{sec:checking-insensitivity}
For each $i \in \{0,\dots,\degree\},j \in \{0,\maxcol-2\},$ $\MollifierChecks$ contains a pair of checks corresponding to \Cref{ineq:noise-insens}. The ``inequality'' portion of the first (respectively, second) check of the pair is
\[
\st_{i,j}(x) \geq   \exp(2\horizclose) \cdot \st_{i,j+1}(x)
\quad \text{(respectively, }
\st_{i,j+1}(x) \geq  \exp(2\horizclose)\cdot \st_{i,j}(x)\text{)},
\]
and the ``softness'' parameter of each of these checks is $\horizclose$. So for these two elements $\chk,\chk'$ of $\MollifierChecks$ the associated ``soft checks'' are
\begin{align}
\softcheck_\chk(x) &= \sigma
\parens*{\horizclose^{-1} \ln\parens*{\frac{ \exp(2\horizclose) \cdot \st_{i,j}(x)}{\st_{i,j+1}(x)}}}
\quad \text{and} \label{eq:horizontal-check1}\\
\softcheck_{\chk'}(x) &= \sigma
\parens*{\horizclose^{-1} \ln\parens*{\frac {\exp(2\horizclose) \cdot \st_{i,j+1}}{\st_{i,j}}}}. \label{eq:horizontal-check2}
\end{align}
We refer to these $2\degree\maxcol$ checks as ``noise-insensitivity checks'' or  as ``horizontal checks.''

\medskip

This concludes the definition of $\MollifierChecks$, so recalling \Cref{def:allsoftchecks} the definition of $\allsoftchecks$ is now complete.  We turn to proving \Cref{thm:hard-to-soft}.


\subsection{Breaking down the mollification error for the proof of \Cref{thm:hard-to-soft}} \label{sec:hard-to-soft}

Recalling the definition of $\MollifierChecks$ from \Cref{sec:MollifierChecks-def}, the approach to proving \Cref{thm:hard-to-soft} is clear.  We will show that each of the local hyperconcentration (diagonal) checks passes ``with room to spare'' with high probability over $\bx$, and that likewise each of the noise-insensitivity (horizontal) checks passes with room to spare with high probability over $\bx$.  The two theorems stated below give the desired bounds:

\begin{theorem} [Local hyperconcentration\kanenote{, rough analogue of Corollary~10 of \cite{kane11focs}}]
\label{thm:local-hyperconcentration}
For each $i \in \{0,\dots,\degree-1\}$, except with probability at most $\epsprg/(8 \degree)$ over $\bx \sim \normal(0,1)^n$, \Cref{ineq:anticoncgap} holds, i.e.
\[
\st_{i+1,0}(\bx) \leq \anticoncgap \st_{i,1}(\bx).
\]
\end{theorem}

\begin{theorem} [Noise-insensitivity\kanenote{, analogue of Lemma~11 of \cite{kane11focs}}] \label{thm:noise-insensitivity}
 For each $i \in \{0,\dots,\degree\}$ and $j \in \{0,\dots,\maxcol-1\},$ except with probability at most
 $\epsprg/(8(\degree+1)\maxcol)$ over $\bx \sim \normal(0,1)^n$, \Cref{ineq:noise-insens} holds, i.e.
\[
 \st_{i,j}(\bx) \approx_{\horizclose} \st_{i,j+1}(\bx).
\]
\end{theorem}
\begin{proofof}{\Cref{thm:hard-to-soft} using \Cref{thm:local-hyperconcentration} and \Cref{thm:noise-insensitivity}.}  By a union bound over failure probabilities, we have that with probability at least $1 - \epsprg/4$ \Cref{ineq:anticoncgap} holds for all $i$ and \Cref{ineq:noise-insens} holds for all $i,j$. If \Cref{ineq:anticoncgap} holds for a given $i$ then $\ln {\frac{e\anticoncgap \st_{i,1}}{\st_{i+1,0}}} \geq 1$ and the diagonal check \Cref{eq:diagonal-check} evaluates to 1.
If \Cref{ineq:noise-insens} holds for a given $i,j$ then $\horizclose^{-1} \ln \parens*{\frac {\exp(2 \horizclose) \cdot s_{i,j}}{s_{i,j+1}}} \geq 1$ and the horizontal check \Cref{eq:horizontal-check1} evaluates to 1, and similarly $\horizclose^{-1} \ln\parens*{\frac {(\exp(2 \horizclose) \cdot  \st_{i,j+1}}{\st_{i,j}}} \geq 1$ and the horizontal check \Cref{eq:horizontal-check2} evaluates to 1.
\end{proofof}

It remains to prove \Cref{thm:local-hyperconcentration,thm:noise-insensitivity}.

%




\section{Local hyperconcentration: Proof of
\Cref{thm:local-hyperconcentration}} \label{sec:anticoncentration}

In this section we present the key new ingredient underlying our main result, the \emph{Local Hyperconcentration Theorem} for degree-$\degree$ polynomials. As alluded to in the Introduction, this result says that with high probability over a Gaussian $\bx \sim \normal(0,1)^n,$ the $\noisegeneric$-zoom of a degree-$d$ polynomial $p$ at~$\bx$ (i.e.~the polynomial $\zoom{p}{\noisegeneric}{\bx}$) is attenuated --- intuitively, it is ``very close to a constant polynomial''.  We refer to this result as a ``local hyperconcentration theorem'' since by \Cref{lem:attenuated-hyperconcentrated} attenuation of $\zoom{p}{\noisegeneric}{\bx}$ implies that the random variable $\zoom{p}{\noisegeneric}{\bx}(\by)$ (for $\by \sim \normal(0,1)^n$) is hyperconcentrated; this property will play a crucial role in our later technical arguments.

For technical reasons related to the definition of our statistics (essentially because each statistic $\st_{i,j}(x)$ is an average of polynomials --- recall \Cref{sec:distributional}), the actual statement we will need is one that is about a distribution of polynomials rather than a single polynomial.  However, for clarity of exposition we first state the ``one-polynomial'' version of \violet{the original local hyperconcentration theorem from \cite{OST20}} below:

\begin{theorem}[Local hyperconcentration theorem for a single polynomial]                                   \label{thm:RZL}
       Let $g: \R^n \to \R$ be a polynomial of degree at most $\degree$.   Fix parameters $\hypergeneric \geq 1$,
        $0 < \eps \leq 1$, $0 < \beta < 1$, and assume 
        \[
            \noisegeneric \leq \frac{\eps}{\hypergeneric^2} \cdot \parens*{\frac{\beta}{d}}^{C \log \degree} 
       \]
(where $C$ is a certain universal constant). Then for $\bx \sim \normal(0,1)^n$, except with probability at most~$\beta$ we have that
       the randomly zoomed polynomial $\zoom{g}{\noisegeneric}{\bx}$ is $(\hypergeneric,\eps)$-attenuated; i.e.,
        \[
            \HV_{\hypergeneric}[\zoom{g}{\noisegeneric}{\bx}] \leq \eps \cdot \|\zoom{g}{\noisegeneric}{\bx}\|_2^2.
        \]
\end{theorem}
Another way to phrase the conclusion is that for $\bh = \zoom{g}{\noisegeneric}{\bx}$, except with probability~$\beta$ we have that $\bh$ is such that
\[
    \HV_{\hypergeneric}[\bh] \leq (d/\beta)^{O(\log d)} \cdot  \hypergeneric^2 \noisegeneric \cdot \|\bh\|_2^2.
\]
Notice that the dependencies here on $\hypergeneric$ and $\noisegeneric$ are ``correct'' in the sense that if one intuitively thinks of $\noisegeneric$ as ``infinitesimal'', we  expect that~$\bh$ will have $\Theta(\noisegeneric)$ weight at level~$1$, negligible weight above level~$1$, and the definition of $\HV_{\hypergeneric}$ multiplies this $\Theta(\noisegeneric)$ level-$1$ weight by~$\hypergeneric^2$.  The ``error factor'' in this result, $(d/\beta)^{O(\log d)}$ with $\beta \sim \epsprg$,  essentially becomes our final seed length (divided by $\log n$).

Because of the need to analyze the statistics introduced in \Cref{sec:the-statistics}, we will often need to work with a distribution over polynomials rather than a single polynomial.  We therefore introduce the following generalization of \Cref{def:attenuated}, which captures the notion of a distribution over polynomials being attenuated on average:

\begin{definition}  [Nice distribution of polynomials is attenuated on average]\label{def:attenuated-on-average-prelim}
Let $(g_{\bup})_{\bupsilon \sim \Upsilon}$ be a nice (in the sense of \Cref{sec:distributional}) distribution of polynomials over $\R^n$\ignore{of degree at most $\degree$}.   For $\hypergeneric \geq 1$ and $0 < \eps \leq 1$, we say that the distribution $(g_{\bup})_{\bup \in \Upsilon}$ is \emph{$(\hypergeneric,\eps)$-attenuated on average} if
    \[
        \E_{\bup} \bracks*{\HV_\hypergeneric[g_{\bup}]} \leq \eps \cdot \E_{\bup} \bracks*{\|g_{\bup}\|_2^2}.
    \]
\end{definition}

The actual main result we prove in this section is \Cref{thm:HZL}, which generalizes \Cref{thm:RZL} to a nice distribution of polynomials \violet{and is the original local hyperconcentration theorem from \cite{OST20}}:

\begin{theorem} [Local hyperconcentration theorem for a nice distribution of polynomials]                                   \label{thm:HZL}
       Let $(g_{\bup})_{\bup \sim \Upsilon}$ be a nice (in the sense of \Cref{sec:distributional}) distribution of degree-$\degree$ polynomials.   Fix parameters $\hypergeneric \geq 1$, $0 < \eps \leq 1$, $0 < \beta < 1$, and assume
\begin{equation} \label{eq:HZL-conditions}
\noisegeneric \leq \frac{\eps}{\hypergeneric^2} \cdot \parens*{\frac{\beta}{d}}^{C \log \degree}              \end{equation}
(where $C$ is a certain universal constant). Then for $\bx \sim \normal(0,1)^n$, except with probability at most $\beta$ we have that the distribution $(\zoom{(g_\bup)}{\noisegeneric}{\bx})_{\bup \in \Upsilon}$ is $(\hypergeneric,\eps)$-attenuated on average; i.e.,
        \[
            \E_{\bup} \bracks*{\HV_{\hypergeneric}[\zoom{(g_{\bup})}{\noisegeneric}{\bx}]}
            \leq  \eps \cdot \E_{\bup} \bracks*{\|\zoom{(g_{\bup})}{\noisegeneric}{\bx}\|_2^2}.
        \]
\end{theorem}

\violet{In \Cref{app:kane} an improved version of \Cref{thm:HZL}, namely \Cref{main theorem},
 is proved,  which only requires an upper bound on $\lambda$ of $\eps\beta/(Rd^{O(1)})$.}  
\Cref{thm:local-hyperconcentration} follows from \Cref{main theorem}\ignore{\rasnote{Was ``follows from \Cref{thm:HZL}''; in the following ``Proof of \Cref{thm:local-hyperconcentration} using \Cref{thm:HZL}'' all the ``\Cref{thm:HZL}'''s were replaced by ``\Cref{main theorem}'''s}} directly by setting parameters as follows:

\begin{proofof}{\Cref{thm:local-hyperconcentration} using \Cref{main theorem}}
We  instantiate \Cref{main theorem}  with its nice distribution ``$(g_{\bup})_{\bup \sim \Upsilon}$'' being  ${\cal F}_{i,0}$,  its ``$\hypergeneric$'' parameter being set to $\hyperfixed$ defined in \Cref{eq:hyperfixed}, its ``$\noisegeneric$'' parameter being $\noisefixed$, its ``$\eps$'' parameter being $\anticoncgap$, and its ``$\beta$'' parameter being $\epsprg/(8 \degree)$. Recalling \Cref{eq:alt-sij} and \Cref{eq:stats-as-zooms} we have
\[
\st_{i,0}(x) = \Ex_{\bup}[g_{\bup}(x)^2], \quad \st_{i,1}(x) =  \E_{\bup}\bracks*{\|\zoom{(g_{\bup})}{\noisefixed}{x}\|_2^2}, \quad \text{and~}\st_{i+1,0}(x)= \E_{\bup} \bracks*{\HV_{\hyperfixed}[\zoom{(g_{\bup})}{\noisefixed}{x}]}.
\]
Recalling the settings of $\noisefixed$ and $\anticoncgap$ from \Cref{eq:noisefixed} and \Cref{eq:anticoncgap-horizclose}, we see that the bound required in \Cref{eq:HZL-conditions} indeed holds, and so we can apply \Cref{main theorem}, and its conclusion gives precisely the desired conclusion of \Cref{thm:local-hyperconcentration}.
\end{proofof}



In the rest of this section we prove \Cref{thm:HZL}. We first explain the high-level structure of the argument in \Cref{sec:HZL-idea} and then give the formal proof in the rest of the section.

\subsection{A useful definition, and the high-level argument underlying \Cref{thm:HZL}} \label{sec:HZL-idea}

Before we can give the high level idea of the proof of \Cref{thm:HZL} we need a refined notion of a polynomial being attenuated:

\begin{definition}  [Attenuated polynomial, refined notion]\label{def:attenuated}
Let $g: \R^n \to \R$ be a polynomial of degree at most $\degree.$
For $k \geq 0$, $\hypergeneric \geq 1$, and $0 < \eps \leq 1$, we say that the polynomial $g$ is \emph{$(k,\hypergeneric,\eps)$-attenuated} if
    \begin{equation} \label{eq:attenuated}
\sum_{|\beta| > k} \hypergeneric^{2|\beta|} \wh{g}(\beta)^2 = \HV_\hypergeneric[g^{> k}] \leq  \eps \cdot \|g\|_2^2.
    \end{equation}
Similarly, if $(g_{\bup})_{\bupsilon \sim \Upsilon}$ is a nice (in the sense of \Cref{sec:distributional}) distribution of polynomials over $\R^n$\ignore{of degree at most $\degree$},
we say that the distribution $(g_{\bup})_{\bup \in \Upsilon}$ is \emph{$(k,\hypergeneric,\eps)$-attenuated on average} if
    \[
        \E_{\bup} \bracks*{\HV_\hypergeneric[g_{\bup}^{> k}]} \leq \eps \cdot \E_{\bup} \bracks*{\|g_{\bup}\|_2^2}.
    \]
Note that $g$ being $(0,\hypergeneric,\eps)$-attenuated is the same as $g$ being $(\hypergeneric,\eps)$-attenuated as defined earlier (see \Cref{def:attenuated-prelim}), and likewise for $(g_{\bup})_{\bupsilon \sim \Upsilon}$ (see \Cref{def:attenuated-on-average-prelim}).
\end{definition}


With this refined notion of attenuation in hand we can explain the high level idea of   our local hyperconcentration theorem.  For ease of exposition, below we sketch the underlying ideas in the ``one-polynomial'' setting of \Cref{thm:RZL} (the same ideas drive the proof of  \Cref{thm:HZL}).

So, we are given a degree-$\degree$ polynomial $g$ and the goal is to argue that with high probability over a random point $\bx \sim \normal(0,1)^n$, the polynomial $\zoom{g}{\noisegeneric}{\bx}$ is $(\hypergeneric,\eps)$-attenuated, i.e.~$(0,\hypergeneric,\eps)$-attenuated.  A simple but crucial insight pointing the way is that \emph{random zooms compose}:  in more detail, if $0<\noisegeneric,\noisegeneric'<1$ are two noise rates and $\bx,\bx'$ are two independent $\normal(0,1)^n$ random variables, then the distribution of the composed random zoom $\zoom{(\zoom{g}{\noisegeneric}{\bx})}{\noisegeneric'}{\bx'}$ is identical to the distribution of $\zoom{g}{\noisegeneric \noisegeneric'}{\bz}$ where $\bz \sim \normal(0,1)^n$.  With this in mind, it is natural to view a random zoom at the small noise rate $\noisegeneric$ as a ``strong'' random zoom which is obtained by composing a sequence of $\log \degree$ many ``weaker'' random zooms at larger noise rates.\footnote{The idea of decomposing a ``strong''random zoom into multiple ``weak'' random zooms is due to Avi Wigderson.}  If we can prove that a ``weak'' random zoom with high probability causes a $(k,\hypergeneric,\eps)$-attenuated polynomial to become $(k/2,\hypergeneric,\eps)$-attenuated, then since any degree-$\degree$ polynomial is trivially $(\degree,\hypergeneric,\eps)$-attenuated, a simple union bound over $\log \degree$ many applications of this ``one-stage'' result yields the desired random zoom lemma for $g$.  This is precisely the high-level structure of our argument; see \Cref{thm:HZL-one-hit} for a formal statement of the one-stage result in the more general setting of a nice distribution of polynomials.

We proceed to give intuition for the proof of the one-stage result.  In this setting we are now given $g$ which is a $(k,\hypergeneric,\eps)$-attenuated polynomial; intuitively this means that the amount of Hermite weight it has at levels $k+1,k+2,\dots$ is very small compared to the total Hermite weight of $g$ at all levels $0,1,\dots$.  We must argue that with high probability over $\bx \sim \normal(0,1)^n$, after a random zoom at $\bx$ the polynomial $\bq \coloneqq\zoom{g}{\noisegeneric}{\bx}$ is $(k/2,\hypergeneric,\eps)$-attenuated, i.e.~the amount of Hermite weight $\bq$ has at levels $k/2 + 1,k/2 + 2,\dots$ is very small relative to the total Hermite weight of $\bq$ at levels $0,1,\dots$.  This is naturally done via a two part argument.  The first part is to argue \emph{two-norm retention}:  this amounts to showing that with high probability over $\bx$, the squared two-norm of $\bq$ does not become too small relative to the squared two-norm of $p$. The argument for this is based on the Carbery--Wright anticoncentration bound (\Cref{thm:CW}) and the tail bound for Gaussian polynomials (\Cref{thm:tail-bound}); see \Cref{sec:retention} for a precise statement and proof of this part.
The second part is to argue \emph{attrition} of the high-degree Hermite weight:  this amounts to showing that with high probability after a random zoom, the amount of Hermite weight at levels $k/2+1,k/2+2,\dots$ becomes very small relative to the squared two-norm of $p$. The argument for this is based on \Cref{cor:attenuate0} and Markov's inequality; see \Cref{sec:attrition} for a precise statement and proof.

\subsection{First part of the proof of the one-stage local hyperconcentration theorem: Retention} \label{sec:retention}

The main result of this section is \Cref{lem:HZL-2-norm-retention}. Its proof uses the following proposition:
\begin{proposition}                                     \label{prop:HZL-CW-corollary}
    Let $(g_{\bup})_{\bup \sim \Upsilon}$ be a nice distribution of polynomials of degree at most~$k$.  Then for $\bx \sim \normal(0,1)^n$, except with probability at most~$\beta'$ we have
    \[
\Ex_{\bup} [\|(g_{\bup})_{\lambda\mid\bx}\|_2^2] \geq \parens*{\frac{\beta'}{O(k)}}^{2k} \Ex_{\bup} [\|g_{\bup}\|_2^2].
    \]
\end{proposition}
\begin{proof}
    Let $r(x) \coloneqq \Ex_{\bup} [\|(g_{\bup})_{\lambda\mid x}\|_2^2].$ We observe that $r$ is a nonnegative degree-$2k$ polynomial with mean
\begin{align*}
    \E_{\bx \sim \normal(0,1)^n}[r(\bx)] &=
    \Ex_{\bup} \Ex_{\bx \sim \normal(0,1)^n}[\|(g_{\bup})_{\lambda\mid\bx}\|_2^2] =
     \Ex_{\bup} \Ex_{\bx,\bx' \sim \normal(0,1)^n}[g_{\bup}(\sqrt{1-\lambda} \bx + \sqrt{\lambda} \bx')]\\
     &=
     \Ex_{\bup} \Ex_{\bg \sim \normal(0,1)^n}[g_{\bup}(\bg)^2]=
     \Ex_{\bup}[\|g_{\bup}\|_2^2].
\end{align*}
The claimed result now follows immediately from the Carbery--Wright anticoncentration bound \Cref{thm:CW} applied to $\p = r$, since $\|r\|_2 \geq \|r\|_1 = \E[r]  =  \Ex_{\bup}[\|g_{\bup}\|_2^2]$.
\end{proof}

One way to think of the nice distribution of degree-$k$ polynomials in \Cref{prop:HZL-CW-corollary} is that it is ``$(k,+\infty,1)$-attenuated on average.''
\Cref{lem:HZL-2-norm-retention} relaxes this requirement and shows that a similar result holds for a nice distribution that is $(k,\hypergenerictwo,1)$-attenuated on average for a modestly large $\hypergenerictwo$.

\begin{lemma}                                     \label{lem:HZL-2-norm-retention}
    Let $(g_{\bup})_{\bup \sim \Upsilon}$ be a nice distribution of polynomials that is $(k,\hypergenerictwo,1)$-attenuated on average for some $k \in \N^+$.
    Fix a parameter $0 < \beta' < 1$.  Then for $\bx \sim \normal(0,1)^n,$ except with probability at most~$\beta'$ we have
    \[
        \Ex_{\bup}\bracks*{\|(g_{\bup})_{\lambda\mid\bx}\|_2^2} \geq \parens*{\frac{\beta'}{O(k)}}^{2k} \Ex_{\bup}\bracks*{\|g_{\bup}\|_2^2},
    \]
provided that (for a certain universal constant~$C$)\ignore{\rasnote{This was ``$        \hypergenerictwo \geq C\max\parens*{k/\beta', \sqrt{k \log(3/\beta')}/\beta'}$'' in Ryan's original writeup but I couldn't quite get the inequality at the end to work out with that --- maybe I'm making a silly error?  With what follows, which is slightly weaker (but not I think in a way that matters for us), it seems to work out with a little room to spare.  This change is why the lower bound on $\hypergenerictwo$ changed from what Ryan had earlier in \Cref{thm:RZL} and \Cref{thm:HZL}.}}
    \[
\hypergenerictwo \geq Ck \log(3/\beta')/\beta'.
    \]
\end{lemma}
\begin{proof}
\ignore{
In this proof we only use that $(g_{\bup})_{\bup \sim \Upsilon}$ is $(k,\hypergenerictwo,1)$-attenuated on average, which is an immediate consequence of its being $(k,\hypergenerictwo,\eps)$-attenuated on average since $\eps \leq 1.$}
We introduce the notation $\ell_\upsilon \coloneqq g_\upsilon^{\leq k}$ and $h^j_\upsilon \coloneqq g_\upsilon^{= k + j}$ for $j = 1, 2, 3, \dots$, so $g_\upsilon = \ell_\upsilon + \sum_{j \geq 1} h^j_\upsilon$.
We may assume without loss of generality that $\Ex_{\bup}\bracks*{\|g_{\bup}\|_2^2} = 1$, or equivalently, $\Ex_{\bup}\bracks*{\|\ell_\bup\|_2^2} + \sum_{j \geq 1} \Ex_{\bup}\bracks*{\|h^j_\bup\|_2^2}= 1$.  Since $(g_{\bup})_{\bup \sim \Upsilon}$ is $(k,\hypergenerictwo,1)$-attenuated on average, we have that
    \begin{align}
 \Ex_{\bup}\bracks*{\|g_{\bup}\|_2^2}  = 1 \geq \Ex_{\bup} \bracks*{\HV_\hypergenerictwo[g_{\bup}^{> k}]} &= \sum_{i > k} \hypergenerictwo^{2i}  \Ex_{\bup}\bracks*{ \|g_{\bup}^{=i}\|_2^2} \label{eqn:dual-use} \\
        &\geq  \hypergenerictwo^{2k}  \sum_{j \geq 1}\Ex_{\bup} \bracks*{ \|h^j_\bup\|_2^2} = \hypergenerictwo^{2k} \parens*{1-\Ex_{\bup}\bracks*{\|\ell_\bup\|_2^2}}. \nonumber
    \end{align}
We may deduce that
    \[
        \Ex_{\bup} \bracks*{\|\ell_\bup\|_2^2} \geq 1 - 1/\hypergenerictwo^{2k} \geq .99,
    \]
where the latter inequality holds assuming $C$ is large enough.  From \Cref{prop:HZL-CW-corollary}, we conclude that
    \begin{equation} \label[ineq]{eq:low-stay}
        \text{except with probability at most\ }\beta'/2 \text{ over }\bx, \qquad         \Ex_{\bup}\bracks*{\|(\ell_\bup)_{\lambda\mid\bx}\|_2^2} \geq \parens*{\frac{\beta'}{C_1 k}}^{2k},
    \end{equation}
    where $C_1$ is a universal (large) constant.  Our goal will be to establish the following: for all $j \geq 1$,
    \begin{equation}    \label[ineq]{eqn:HZL-high-death}
         \text{except with probability at most\ $(\beta'/3)/10^j$  over $\bx$}, \qquad         \Ex_{\bup}\bracks*{\|(h^j_\bup)_{\lambda\mid\bx}\|_2^2} \leq \parens*{\frac{\beta'}{C_1 k}}^{2(k+j)}.
    \end{equation}

    Before establishing \Cref{eqn:HZL-high-death}, we show how it yields the conclusion of the lemma.
Given \Cref{eq:low-stay,eqn:HZL-high-death}, summing over $j$ and taking a union bound, we get that except with probability at most~$\beta'$ over~$\bx$,
    \begin{equation} \label{eq:about-to-use-this}
        \sqrt{\Ex_{\bup}\bracks*{\|(\ell_\bup)_{\lambda\mid\bx}\|_2^2}} - \sum_{j \geq 1} \sqrt{\Ex_{\bup}\bracks*{\|(h^j_\bup)_{\lambda\mid\bx}\|_2^2}} \geq \frac12 \parens*{\frac{\beta'}{C_1 k}}^{k}.
    \end{equation}
The triangle inequality easily gives that for functions $a(\upsilon),b_1(\upsilon),b_2(\upsilon),\dots$, we have that
    \[
        \sqrt{\Ex_{\bup}\bracks*{(a(\bup) - \littlesum_j b_j(\bup))^2}} \geq \sqrt{\Ex_{\bup}\bracks*{a(\bup)^2}} - \sum_j \sqrt{\Ex_{\bup}\bracks*{b(\bup)^2}}.
    \]
    Applying this (for each outcome of~$\bx$) with $a(\bup) = \|(\ell_\bup)_{\lambda\mid\bx}\|_2$ and $b(\bup) = \|(h^j_\bup)_{\lambda\mid\bx}\|_2$, by \Cref{eq:about-to-use-this} we get that
    \[
        \sqrt{\Ex_{\bup}\bracks*{(\|(\ell_\bup)_{\lambda\mid\bx}\|_2 - \littlesum_{j \geq 1} \|(h^j_\bup)_{\lambda\mid\bx}\|_2)^2}} \geq \sqrt{\Ex_{\bup}\bracks*{\|(\ell_\bup)_{\lambda\mid\bx}\|_2^2}} - \sum_{j \geq 1} \sqrt{\Ex_{\bup}\bracks*{\|(h^j_\bup)_{\lambda\mid\bx}\|_2^2}} \geq {\frac 1 2}  \parens*{\frac{\beta'}{C_1 k}}^{k},
    \]
 from which we get (using $(\|f\| - \|g_1\| - \|g_2\| - \|g_3\| - \cdots )^2 \leq \|f +g_1 + g_2 + g_3 + \cdots\|^2$) that
    \[
        {\frac 1 4} \parens*{\frac{\beta'}{C_1 k}}^{2k} \leq \Ex_{\bup}\bracks*{(\|(\ell_\bup)_{\lambda\mid\bx}\|_2 - \littlesum_{j \geq 1} \|(h^j_\bup)_{\lambda\mid\bx}\|_2)^2} \leq \Ex_{\bup}\bracks*{\|(\ell_\bup)_{\lambda\mid\bx} + \littlesum_{j \geq 1} (h^j_\bup)_{\lambda\mid\bx}\|_2^2} = \Ex_{\bup}\bracks*{\|(g_{\bup})_{\lambda\mid\bx}\|_2^2},
    \]
which is the conclusion of the lemma since $\Ex_{\bup}\bracks*{\|g_{\bup}\|_2^2} = 1$.

    Thus it remains to establish \Cref{eqn:HZL-high-death}.  To do this, write $s_j(\bx) = \Ex_{\bup}\bracks*{\|(h^{j}_\bup)_{\lambda\mid\bx}\|_2^2}$  for brevity, and note that $s_j$ is a nonnegative polynomial of degree at most~$2(k+j)$.  Similar to the proof of \Cref{prop:HZL-CW-corollary}, we have that
    \[
        \|s_j\|_1  = \E_{\bx}[s_j(\bx)] =\E_{\bx} \bracks*{\Ex_{\bup}\bracks*{\|(h^{j}_\upsilon)_{\lambda\mid\bx}\|_2^2}}
        = \Ex_{\bup} \bracks*{\E_{\bx}[\|(h^{j}_\bup)_{\lambda\mid\bx}\|_2^2]}
        = \Ex_{\bup} \bracks*{\|h^{j}_\bup\|_2^2},
    \]
and by \Cref{eqn:dual-use} we have that $\Ex_{\bup} \bracks*{\|h^{j}_\bup\|_2^2} \leq 1/\hypergenerictwo^{2(k+j)}$.  Now using \Cref{thm:aobf-9.22} we get that $\|s_j\|_2  \leq (e/\hypergenerictwo)^{2(k+j)} $ and therefore (using \Cref{thm:tail-bound}, with its ``$t$'' set to $u^{2(k+j)}$) we get that
    \[
        \Pr\bracks*{s_j(\bx) \geq (ue/\hypergenerictwo)^{2(k+j)} } \leq \exp(-(k+j) u^2/e) = \exp(-u^2k/e) \exp(-u^2/e)^j
    \]
    for any choice of $u \geq \sqrt{2e}$.
    We will select $u = \frac{\beta' \hypergenerictwo}{C_1e k}$, so that the preceding inequality aligns with \Cref{eqn:HZL-high-death}; recalling the bound on $R$, this choice of $u$ is indeed at least $\sqrt{2e}$ provided~$C$ is taken at least $\sqrt{2}e^{3/2}C_1$.    Also taking $C$~sufficiently large in our assumption on~$\hypergenerictwo$, it is not hard to arrange for the error probability above to be at most~$(\beta'/3)(.1)^j$.  Thus \Cref{eqn:HZL-high-death} is established and the proof of \Cref{lem:HZL-2-norm-retention} is complete.
    \end{proof}

It is interesting to observe that both the results of this section, \Cref{prop:HZL-CW-corollary} and \Cref{lem:HZL-2-norm-retention}, hold with no dependence on the value of $0 < \noisegeneric < 1.$

\subsection{Second part of the proof of the one-stage local hyperconcentration theorem: Attrition}
\label{sec:attrition}

\ignore{
}

The attrition result we establish in this subsection, \Cref{lem:HZL-hi-degree-weight-low}, is a fairly direct consequence of \Cref{cor:attenuate0}. (Note that \Cref{lem:HZL-hi-degree-weight-low} does not  require that the nice distribution $(g_{\bup})_{\bup \sim \Upsilon}$ be \ignore{$(k,\hypergeneric,\eps)$-}attenuated on average --- it holds for any nice distribution of degree-$\degree$ polynomials.)

\begin{lemma} \label{lem:HZL-hi-degree-weight-low}
   Let $(g_{\bup})_{\bup \sim \Upsilon}$ be a nice distribution of polynomials of degree at most~$\degree$.  
   Fix parameters $0 < \beta' < 1$,  $\hypergeneric' \geq 1$, $m \in \N^+$, $0 < \eps \leq 1$, let $c > 0 $ be a sufficiently small constant, and assume $\lambda \leq \frac{c^4 \eps\beta'^5}{{\hypergeneric'}^2 m^3 \degree}$.\ignore{\rodnote{I don't see why you need $\eps^5$ here, don't you just need $\eps^1$?  I'm kind of into trying to get the good $\lambda$ vs.\ $\eps$ tradeoff, imagining all other factors are ``constant''}}  Then for~$\bx \sim \normal(0,1)^n$,
    \[
    {\frac 1 \eps}   \cdot \Ex_{\bup} \bracks*{\HV_{\hypergeneric'}[((g_{\bup})_{\lambda \mid \bx})^{\geq m}]} \leq \parens*{ \frac{O(c\beta')}{m}}^{{4}m} \Ex_{\bup} \bracks*{\|g_{\bup}\|_2^2}
    \]
    holds except with probability at most~$\beta'$.
\end{lemma}

\begin{proof}
    The expectation of the left-hand side is
    \begin{align*}
        \E_{\bx}\bracks*{\Ex_{\bup} \bracks*{\HV_{\hypergeneric'}[((g_{\bup})_{\lambda \mid \bx})^{\geq m}]}}
         &= \Ex_{\bup}\bracks*{ \sum_{i \geq m} {\hypergeneric'}^{2i} \E_{\bx}\bracks*{\W^{= i}\bracks*{(g_{\bup})_{\lambda \mid \bx}}}}  \\
         &=\Ex_{\bup}\bracks*{ \sum_{i \geq m} {\hypergeneric'}^{2i} \sum_{j \leq \degree} \Pr[\text{Bin}(j,\lambda) = i] \W^{= j}\bracks*{g_{\bup}}} \tag{\Cref{cor:attenuate0}}\\
         &=\sum_{j \leq \degree} \Ex_{\bup}\bracks*{ \W^{=j}[g_{\bup}]} \sum_{i \geq m} \Pr[\text{Bin}(j,\lambda) = i] {\hypergeneric'}^{2i}\\
        &\leq \sum_{j \leq \degree} \Ex_{\bup}\bracks*{ \W^{=j}[g_{\bup}]} \sum_{i \geq m} \parens*{\frac{ej}{i}}^i \lambda^i {\hypergeneric'}^{2i} \\
        &\leq \sum_{j \leq \degree} \Ex_{\bup}\bracks*{ \W^{=j}[g_{\bup}]} \parens*{\frac{2ej{\hypergeneric'}^2 \lambda}{m}}^{m} \tag*{(if $c$ small enough)}\\
        & \leq \parens*{\frac{2e\degree{\hypergeneric'}^2 \lambda}{m}}^{m}  \Ex_{\bup}\bracks*{ \|g_{\bup}\|_2^2} \\
        & \leq \parens*{\frac{2ec^4 \eps\beta'^5}{m^4}}^{m}  \Ex_{\bup}\bracks*{ \|g_{\bup}\|_2^2} \tag*{(by the bound on $\lambda$)} \\
        & \leq \beta' \eps \parens*{\frac{O(c \beta')}{m}}^{4m}  \Ex_{\bup}\bracks*{ \|g_{\bup}\|_2^2}.
    \end{align*}
    The result now follows by Markov's inequality.
\end{proof}

\subsection{Putting the pieces together:  Proof of the local hyperconcentration theorem}

Combining \Cref{lem:HZL-2-norm-retention} and \Cref{lem:HZL-hi-degree-weight-low} (with the ``$k$'' and ``$m$'' parameters satisfying $m = \lfloor k/2 \rfloor$, and adjusting constants),  we may deduce the following, which is our ``one-stage local hyperconcentration theorem:''

\begin{theorem} [One-stage local hyperconcentration theorem]                                    \label{thm:HZL-one-hit}
       Let $(g_{\bup})_{\bup \sim \Upsilon}$ be a nice distribution of polynomials of degree at most~$\degree$, and assume the distribution  is $(k,\hypergenerictwo,1)$-attenuated on average for some $1 \leq k \leq \degree$.  Fix parameters $\hypergeneric' \geq 1$, $0 < \beta' < 1$, $0 < \eps' \leq 1$, and assume\ignore{\rasnote{As before, earlier this was ``$\hypergenerictwo  \geq C \max\parens*{k/\beta', \sqrt{k \log(3/\beta')}/\beta'}$''}}
       \[
 \hypergenerictwo \geq Ck \log(3/\beta')/\beta',
       \qquad \lambda \le \frac{\eps'\beta'^5}{C{\hypergeneric'}^2 k^3 \degree }
       \]
for a suitably large universal constant $C$.
       Then except with probability at most~$\beta'$ over $\bx \sim \normal(0,1)^n$, the distribution $(\zoom{(g_{\bup})}{\lambda}{\bx})_{\bup \sim \Upsilon}$ is $(\lfloor k/2\rfloor,\hypergeneric',\eps')$-attenuated on average.
\end{theorem}

Note that a nice distribution of degree-$\degree$ polynomials is $(\degree,\hypergenerictwo,1)$-attenuated on average for any~$\hypergenerictwo$.  We can take $\hypergenerictwo = C\degree \log(3/\beta')/\beta'$ and perform a first application of \Cref{thm:HZL-one-hit} on $(g_{\bup})_{\bup \sim \Upsilon}$ with its $\hypergeneric'$ parameter set to $\hypergenerictwo$ and its $\eps'$ parameter set to 1, and infer that except with failure probability at most $\beta'$ the distribution
$(\zoom{(g_{\bup})}{\lambda}{\bz})_{\bup \sim \Upsilon}$ is $(\lfloor \degree/2\rfloor,\hypergenerictwo,1)$-attenuated on average.  Repeating this a total of $\lceil \log d \rceil$ times, with each repetition having its $\beta'$ parameter set to $\beta/(\lceil \log d \rceil + 1)$, its $\hypergeneric'$ parameter set to $\hypergenerictwo =C\degree \log(3/\beta')/\beta'$ (for simplicity),  and its $\eps'$ parameter set to~$1$, we get that
except with probability $\beta \cdot (1 - {\frac 1 {\lceil \log d \rceil + 1 }})$,
the distribution $(\zoom{(g_\upsilon)}{\noisegeneric'}{\bz})_{\upsilon \in \Upsilon}$ is $(1,\hypergenerictwo,1)$-attenuated on average, where $\noisegeneric' = (\frac{\beta'^6}{C^2 d^6 \log^2(3/\beta')})^{\lceil \log d \rceil}$. Finally, we perform one last application of \Cref{thm:HZL-one-hit} with its $k$ parameter set to 1, its $\noisegeneric'$ parameter set to $\noisefixed$, its $\eps'$ parameter set to the ``$\eps$'' of \Cref{thm:HZL}, and its $\beta'$ parameter set to $\beta/(\lceil \log d \rceil + 1)$ and its $\hypergeneric'$ parameter set to $\hypergeneric$. We get the conclusion of \Cref{thm:HZL} as stated at the beginning of this section, and the proof of the local hyperconcentration theorem is complete.\qed

\ignore{

\rasnote{If we wanted we could get a somewhat tighter bound on $\noisegeneric$ but I'm not sure it makes a difference for us overall. \Cref{lem:HZL-2-norm-retention} doesn't care about the ``$\eps$'' of $(k,\hypergeneric,\eps)$-attenuated --- as far as that is concerned we could always have it just be 1.  And \Cref{lem:HZL-hi-degree-weight-low} could also be run with the ``$\eps$'' parameter being 1 for each of the first $\log_2(d)-1$ stages to get $k$ down to 1 --- and then we could just once at the last stage run it with $\eps$.  So overall instead of having $ \noisegeneric \leq \parens*{\frac{\eps\beta}{d \hypergeneric}}^{c \log d}$ as we currently have, we could I think instead have it be $ \noisegeneric \leq \eps^5 \cdot \parens*{\frac{\beta}{d \hypergeneric}}^{c \log d}$ or something like that if we like. It seemed like this would require the exposition to become a little messier so I didn't do it, but if the difference in the parameters matters to us, we can go that route. (I don't think it will matter but didn't trace through 100\% to confirm.)}
}


\section{Noise insensitivity of the statistics: Proof of \Cref{thm:noise-insensitivity}} \label{sec:noise-insensitivity}

\begin{remark} \label{rem:our-noise-insensitivity-of-statistics-versus-kanes}
Before entering into the proof, we note that \Cref{thm:noise-insensitivity} is analogous to Lemma~11 of \cite{kane11focs}, which shows that for every $i,j,$ for most $\bx \sim \normal(0,1)^n$, the $(i,j)$-th statistic at $\bx$ is multiplicatively close to the $(i,j+1)$-th statistic at $\bx$.  \cite{kane11focs}'s proof of Lemma~11 uses his \Cref{lem:kane-lemma9-alt} (i.e., \cite[Lem.~9]{kane11focs}) together with hypercontractivity, but as we discussed earlier, this incurs a $2^{O(d)}$ factor.

Our arguments in this section also use Lemma~9 of \cite{kane11focs}, but they additionally use our Local Hyperconcentration Theorem  and our notions of attenuation and hyperconcentration (specifically \Cref{prop:atten-hyperconcy}).\ignore{  and a quantitative strengthening of Corollary~4 of \cite{kane11focs} for attenuated polynomials which we state and prove below.} These  new ingredients let us avoid the $2^{O(d)}$ factor which is incurred at this point in the \cite{kane11focs} argument.
\end{remark}

\ignore{
\subsection{Attenuated polynomials are strongly concentrated}

The following lemma\ignore{, which we will use in the proof of \Cref{thm:noise-insensitivity},} says that any attenuated polynomial is strongly concentrated.  This may be viewed as a significant quantitative strengthening of Corollary~4 of \cite{kane11focs}, which states that for any degree-$d$ polynomial $g$, we have $\Prx_{\bx \sim \normal(0,1)^n}[|g(\bx)|\geq \|g\|_2/2] \geq 9^{-d}/2$; in words, $g$ is at least slightly likely to have magnitude at least half of its 2-norm.  In contrast, for attenuated polynomials \Cref{lem:our-cor4} states that with very high probability the magnitude is multiplicatively very close to the 2-norm.

\begin{lemma} [Strong concentration for attenuated polynomials\ignore{; our version of \cite{kane11focs}'s Corollary~4}] \label{lem:our-cor4} Let $0 < \eps,\tau,\delta_1 < 1$ satisfy $ \eps \leq {\frac 1 {(2e \ln(d/\tau))^3}}$ and let $q$ be an $(\hypergeneric,\eps)$-attenuated polynomial of degree at most $d$, where $\hypergeneric \geq \max\{{\frac {4 \eps^{1/3}} {\delta_1}},$ $2ed \ln {\frac d \tau}\}.$   Then
\[
\Prx_{\bx \sim \normal(0,1)^n}\bracks*{|q(\bx)| \approx_{\delta_1} \|q\|_2} \geq 1-\tau.
\]
\end{lemma}

\begin{proof}
Since $q$ is $(\hypergeneric,\eps)$-attenuated, as an immediate consequence of \Cref{def:attenuated} we have that
for each $j=1,\dots,\degree$,
\begin{equation} \label{eq:apple}
\|q^{=j}\|_2  \leq {\frac {\sqrt{\eps}}{\hypergeneric^{j}}} \cdot \|q\|_2.
\end{equation}

Decomposing $q$ into a sum of orthogonal components $q = q^{=0} + \cdots + q^{=\degree}$, by the triangle inequality we have that
\begin{equation*}
|q^{=0}|=\|q^{=0}\|_2 \in \bracks*{
\|q\|_2 - \sum_{j=1}^\degree \|q^{=j}\|_2, \|q\|_2 + \sum_{j=1}^\degree \|q^{=j}\|_2
},
\end{equation*}
which together with \Cref{eq:apple} immediately yields
\begin{equation} \label{eq:triangle}
|q^{=0}| \approx_{{\frac {2 \sqrt{\eps}}{\hypergeneric}}} \|q\|_2.
\end{equation}

Fix any $j \in \{1,\dots,\degree\}.$ We will apply the tail bound \Cref{thm:tail-bound} to the degree-$j$ polynomial $q^{=j}$, with the ``$t$'' parameter of \Cref{thm:tail-bound} set to\ignore{\rasnote{This is kind of annoying, to have to have a separate definition in the $j=1$ case.  It's because we want something like
$t_j {\frac {\sqrt{\eps}}{\hypergeneric^j}} \leq {\frac {\sqrt{\eps}}{d \hypergeneric}}$ to hold, but since $t_j$ has to be slightly largish (surely at least 1), this doesn't fly when $j=1$ because you don't get the power-of-$\hypergeneric$ helping you in the denominator.  So as a hack we settle for $t_1 {\frac {\sqrt{\eps}}{\hypergeneric^1}} \leq {\frac {\eps^{1/3}}{\hypergeneric}}$ instead. Hacky, sorry. But none of this matters I think.}}
\[
t_j \coloneqq
\begin{cases}
{\frac 1 {\eps^{1/6}}} & \text{~if~}j=1;\\
{\frac {\hypergeneric^{j-1}}d} & \text{if~}j \in \{2,\dots,\degree\}.
\end{cases}
\]
The condition on $\eps$ ensures that $t_1 \geq \parens*{{\frac {2e}{j} \ln {\frac \degree \tau}}}^{j/2} \geq \sqrt{2e}$, and the condition on $\hypergeneric$ ensures that for $j=2,\dots,\degree$ we have $t_j \geq \max\{\sqrt{2e}^j,({\frac {2e} j} \ln {\frac d \tau})^{j/2}\}$. So we may apply \Cref{thm:tail-bound}, and we get that for each $j \in \{1,\dots,d\}$,
\begin{equation} \label{eq:first-cb}
\Pr_{\bx \sim \normal(0,1)^n}\bracks*{|q^{=j}(\bx)| > t_j\|q^{=j}\|_2} \leq \exp\left(-\frac{j}{2e}t_j^{2/j} \right)
\leq {\frac \tau \degree}.
\end{equation}
Combining this with \Cref{eq:apple}, we get that

\begin{equation} \label{eq:second-cb}
\Pr_{\bx \sim \normal(0,1)^n}\bracks*{|q^{=j}(\bx)| > t_j \cdot {\frac {\sqrt{\eps}}{\hypergeneric^{j}}} \cdot \|q\|_2}
\leq {\frac \tau \degree}.
\end{equation}

Now, if $j=1$ we have that $t_1 {\frac {\sqrt{\eps}}{\hypergeneric}} = {\frac {\eps^{1/3}}{\hypergeneric}}$, and if $j \in \{2,\dots,\degree\}$ we have that
$t_j {\frac {\sqrt{\eps}}{\hypergeneric^j}} = {\frac {\sqrt{\eps}}{d \hypergeneric}}$. So by a union bound over all $j=1,\dots,\degree$ and the triangle inequality, we have that with probability at least $1-\tau$ over $\bx \sim \normal(0,1)^n$,
\begin{equation} \label{eq:abra}
\abs*{\sum_{j=1}^\degree q^{=j}(\bx)}
\leq
\sum_{j=1}^\degree \abs*{q^{=j}(\bx)}  \leq {\frac {\sqrt{\eps} + \eps^{1/3}}{\hypergeneric}} \|q\|_2.
\end{equation}
Since
\[
|q^{=0}| - \sum_{j=1}^\degree\abs*{q^{=j}(x)}
\leq
|q(\bx)|
\leq
|q^{=0}| + \sum_{j=1}^\degree\abs*{q^{=j}(x)}
\]
holds for all $x$, combining \Cref{eq:triangle} and \Cref{eq:abra} we get that
\[
|q(\bx)| \approx_{{\frac {4 \eps^{1/3}}{\hypergeneric}}} \|q\|_2,
\]
\ignore{
}
which gives the lemma by the condition on $\hypergeneric$ in the lemma statement.
\end{proof}
}


\ignore{

%
}
\ignore{


\subsection{The main technical lemma for \Cref{thm:noise-insensitivity}}
As we explain in the next subsection, \Cref{thm:noise-insensitivity} follows rather directly from \Cref{lem:our-lem11}, which is the main result of this subsection.  The proof requires the following anticoncentration result from \cite{kane11focs}, which intuitively says that the value of a polynomial is likely to be very stable, in a multiplicative sense, when moving from a random point to a random zoom centered at that random point.

\begin{lemma} [Anticoncentration for degree-$\degree$ polynomials; Lemma~9 of \cite{kane11focs}] \label{lem:kane-lemma9}
Let $g$ be any degree-$\degree$ polynomial over Gaussian space.  Suppose that $0 < \noisegeneric < \tau/100$.  Then with probability at least $1-O(d^2 \tau)$ over independent $\bx,\by \sim \normal(0,1)^n$, it holds that\rodnote{Should the $\noisegeneric$ in the below in fact be $\sqrt{\noisegeneric}$?}
\[
|\zoom{g}{\noisegeneric}{\bx}(\by)| \approx_{{\frac \noisegeneric \tau}} |g(\bx)|.
\]
\end{lemma}

\rasnote{You'll start seeing the parameter ``$\hypergenerictwo$'' cropping up a bunch, starting here and going up to (but not into) the start of \Cref{sec:proof-of-thm-noise-insens-using-our-lem-11}.  I thought it might be helpful to use a different symbol for this than ``$\hypergeneric$'' because it is conceptually untethered to the other constraints we sometimes impose on ``$\hypergeneric$'' --- as I understand it, its scope is completely internal to this little piece of the argument (it doesn't have to have any particular relation to $\hyperfixed$, for instance).  Anyway, just wanted to mention this.}

\begin{lemma} [Variant of Daniel's Lemma~11, ``large-$\lambda$ regime'']\label{lem:our-lem11}
Let $(g_\bup)_{\bupsilon \sim \Upsilon}$ be a nice (in the sense of \Cref{sec:distributional}) distribution of polynomials over $\R^n$ of degree at most $\degree$.
Let $0 < \rho < {\frac 1 {100}}$, let $0 < \mu < 1$, let
\[
\noisegeneric' \leq
\kappa \cdot  \parens*{\frac{ \mu}{d \hypergenerictwo}}^{C \log d} = \rho^2 \cdot \parens*{{\frac \mu d}}^{O(\log d)}
\]
(i.e.~$\noisegeneric'$ is at most the right-hand side of the inequality constraining it in \Cref{eq:rxattenuated}), where
$\kappa,\hypergenerictwo$ are set as in \Cref{eq:set-hypergeneric}.
Then with probability at least $1-\mu$ over $\bx \sim \normal(0,1)^n$, we have that
\[
\Ex_{\by \sim \normal(0,1)^n}
\bracks*{\Ex_{\bup}\bracks*{
\zoom{(g_{\bup})}{\noisegeneric'}{\bx}(\by)^2}}
\approx_\rho
\Ex_{\bup}\bracks*{g_{\bup}(\bx)^2}.
\]
\end{lemma}

\begin{proof}
For $x \in \R^n$, we define the quantities
\begin{align}
C_x&\coloneqq \Ex_{\bup}\bracks*{g_{\bup}(x)^2}, \label{eq:Cdef}\\
r_x(y) &\coloneqq \Ex_{\bup}\bracks*{\zoom{(g_{\bup})}{\noisegeneric'}{x}(y)^2},  \text{~and~}\label{eq:rxydef}\\
q_x(y) &\coloneqq r(x,y) - C(x). \label{eq:qxy}
\end{align}
As suggested by the notation, given a fixed $x \in \R^n$, the quantity $C_x$ should be thought of as a ``constant,'' whereas $r_x(y)$ and $q_x(y)$ are polynomials in $y$ of degree at most $2d$.

For notational convenience, for $x \in \R^n$ we define the shorthand
\[
\text{``$A(x)$'' to mean the event ``}\Ex_{\by} \bracks*{\Ex_{\bup}\bracks*{\zoom{(g_{\bup})}{\noisegeneric'}{x}(\by)^2}} \not\approx_{\rho} \Ex_{\bup}\bracks*{g_{\bup}(\bx)^2}\text{,'' i.e., ``}
\Ex_{\by}[r_x(\by)] \not \approx_\rho C_x\text{,''}
\]
and the goal of the proof is to show that $p_A \leq \mu$, where $p_A \coloneqq \Pr_{\bx}[A(\bx)].$

Our argument below will use the following simple fact:

\begin{fact} \label{fact:rx-attenuated} Suppose that $\noisegeneric',\mu,\kappa,\hypergenerictwo$ are such that \rasnote{This condition was ``$
            \hypergenerictwo  \ge  \frac{C d \log d \cdot \sqrt{\log(1/\beta)}}{\beta},
$'', the new condition is propagated from \Cref{thm:HZL}}
\begin{equation}
  \noisegeneric' \leq \kappa \cdot \parens*{\frac{ \mu}{d \hypergenerictwo}}^{C \log d},
\quad \quad
 \hypergenerictwo  \ge  \frac{C \degree \log \degree \cdot \sqrt{\log \log \degree + \log(1/\mu)}}{\mu}
   \label{eq:rxattenuated}
\end{equation}
where $C$ is a suitably large absolute constant. Then with probability at least $1-\mu/2$ over $\bx \sim \normal(0,1)^n$, the polynomial $r_\bx(y)$ is $(\hypergenerictwo,\kappa)$-attenuated.
\end{fact}

\begin{proof}
This is a direct consequence of \Cref{thm:RZL}, our random zoom lemma for a single polynomial.  Define
\[
P(y) \coloneqq \Ex_{\bup}[g_{\bup}(y)^2],
\]
and observe that $P(y)$ is a polynomial of degree at most $2d$. Applying \Cref{thm:RZL} to $P$, we get that except with probability at most $\mu/2$ the polynomial $\zoom{P}{\noisegeneric'}{\bx}$ is $(\hypergenerictwo,\kappa)$-attenuated, and the fact follows since
\[
\zoom{P}{\noisegeneric'}{\bx}(y)
=
P(\sqrt{1-\noisegeneric'} \bx + \sqrt{\noisegeneric'} y)\\
=
\Ex_{\bup}[g_{\bup}(\sqrt{1-\noisegeneric'} \bx + \sqrt{\noisegeneric'} y)^2]\\
=
\Ex_{\bup}[\zoom{(g_{\bup})}{\noisegeneric'}{\bx}(y)^2]\\
= r_{\bx}(y). \qedhere
\]
\end{proof}

We will also use the following:

\begin{claim} \label{claim:r-to-q}
For $\kappa < 1/10$, if $x$ is such that both (i) $r_x(y)$ is $(\hypergenerictwo,\kappa)$-attenuated and (ii) $A(x)$ holds, then $q_x(y)$ is $(\hypergenerictwo,8\kappa/\rho^2)$-attenuated.
\end{claim}

We defer the proof of \Cref{claim:r-to-q} until \Cref{sec:r-to-q} and continue with the proof of \Cref{lem:our-lem11}. Towards the goal of upper bounding $p_A$, we will use \Cref{fact:rx-attenuated}, taking its ``$\hypergenerictwo$'' and ``$\kappa$'' parameters to be
\begin{equation}
\hypergenerictwo = {\frac{C'' \degree \log \degree \cdot \sqrt{\log \log \degree + \log(1/\mu)}}{\mu}},
\quad \quad
\kappa = {\frac {c' \rho^2}{(\log \degree)^3}},
\label{eq:set-hypergeneric}
\end{equation}
where $C''$ is a suitably large constant and $c'$ is a suitably small one.  (So $\hypergenerictwo$ is set equal to the right-hand side of the inequality constraining it in \Cref{eq:rxattenuated}.)
Since these parameter settings are such that $\hypergenerictwo$ and $\noisegeneric'$ satisfy \Cref{eq:rxattenuated},  we have that
\[
\Pr_{\bx}[A(\bx) \text{~holds and~}r_{\bx}(y) \text{~is~}(\hypergenerictwo,\kappa)\text{-attenuated}]
\geq p_A - \mu/2,
\]
and from this, by \Cref{claim:r-to-q}, we get that
\begin{equation} \label{eq:goodie}
\Pr_{\bx}[A(\bx) \text{~holds and~}q_{\bx}(y) \text{~is~}(\hypergenerictwo,8\kappa/\rho^2)\text{-attenuated}]
\geq p_A - \mu/2.
\end{equation}
In the rest of the argument we will upper bound the probability on the left-hand side of \Cref{eq:goodie} and thereby upper bound $p_A$.

So, let us fix an outcome $x$ satisfying the condition of \Cref{eq:goodie}, i.e.\ an $x$ such that $A(x)$ holds and $q_{x}(y)$ is $(\hypergenerictwo,8\kappa/\rho^2)$-attenuated.
We apply \Cref{lem:our-cor4} to the polynomial $q_x$, taking the ``$\delta_1$'' and ``$\tau$'' parameters in \Cref{lem:our-cor4} to each be $1/2$, taking its ``$\eps$'' parameter to be $8\kappa/\rho^2$ (which equals $ 8c/(\log d)^3,$ recalling the setting $\kappa = c\rho^2/(\log d)^3$), and taking its ``$\hypergenerictwo$'' parameter as above to be ${\frac{C \degree \log \degree \cdot \sqrt{\log \log \degree + \log(1/\mu)}}{\mu}}$. It is straightforward to verify that for these settings, the $\hypergenerictwo$ and $\eps$ of \Cref{lem:our-cor4} indeed satisfy the bounds required in that lemma statement.
Since $A(x)$ holds we have $|\Ex_{\by}[q_x(\by)]| \geq \rho \Ex_{\bup}[g_\bup(x)^2]$, and since $\|q_x\| \geq |\wh{q_x}(0^n)| = |\Ex_{\by}[q_x(\by)]|,$ it follows that
$\|q_x\| \geq \rho \Ex_{\bup}[g_\bup(x)^2]$, so we get from \Cref{lem:our-cor4} that
\[
|q_x(\by)| =
\abs*{\Ex_{\bup}\bracks*{\zoom{(g_\bup)}{\noisegeneric'}{x}(\by)^2} - \Ex_{\bup}\bracks*{g_\bup(x)^2}}
\geq  \|q_x\|/2 \geq (\rho/2)\Ex_{\bup}[g_\bup(x)^2]
\]
with probability at least $1/2$ over $\by$.

For $x,y \in \R^n$ let us define the shorthand
\[
\text{``$B(x,y)$'' to mean the event ``}\Ex_{\bup}\bracks*{\zoom{(g_\bup)}{\noisegeneric'}{x}(y)^2} \not\approx_{\rho/2} \Ex_{\bup}\bracks*{g_\bup(x)^2}\text{,'' i.e. ``}
r_x(\by)] \not\approx_{\rho/2} C_x\text{.''}
\]
Summarizing the preceding paragraph, it showed that if $x$ is such that $A(x)$ holds and $q_x$ is $(\hypergenerictwo,8\kappa/\rho^2)$-attenuated, then with probability at least
$1/2$ over $\by$, we have that $B(x,\by)$ holds.

Now we apply \Cref{lem:kane-lemma9} to get an upper bound on $\Pr[B(\bx,\by)]$, with its ``$\noisegeneric$'' parameter set to $\noisegeneric'$ and its ``$\tau$'' parameter set to $2 \noisegeneric'/\rho$, and its polynomial ``$p$'' being $\Ex_{\bup}[g_\bup^2].$  (Note that the ``$\noisegeneric/\tau$'' of \Cref{lem:kane-lemma9} is $\rho/2 < {\frac 1 {200}},$ so \Cref{lem:kane-lemma9} can indeed be applied.) Applied in this way we see that \Cref{lem:kane-lemma9} gives $\Pr[B(\bx,\by)] \leq O(d^2 \tau) = O(d^2 \noisegeneric'/\rho)$.
Combining this with the last sentence of the previous paragraph, it must hold that
\begin{equation} \label{eq:gumdrops}
 \Pr_{\bx}[A(\bx) \text{~holds and~}q_{\bx}\text{~is $(\hypergenerictwo,8\kappa/\rho^2)$-attenuated}] \leq  2 \cdot O(d^2 \noisegeneric'/\rho) =
O(d^2) \cdot \noisegeneric'/\rho.
 \end{equation}
Combining \Cref{eq:goodie} and \Cref{eq:gumdrops}, we get that $p_A \leq O(d^2) \cdot \noisegeneric'/\rho + \mu/2$. Observing that $\noisegeneric' \ll  {\frac {\rho \mu} {O(d^2)}}$, we have that $p_A \leq \mu$, and \Cref{lem:our-lem11} is proved.
\end{proof}

\subsubsection{Proof of \Cref{claim:r-to-q}} \label{sec:r-to-q}

Recall that by the assumption of the claim, we have that $r_x$ is $(\hypergenerictwo,\kappa)$-attenuated and that $A(x)$ holds.  The latter can be rephrased as
\[
\E[r_x(\by)] \not\approx_{\rho} C_x.
\]
Recalling that $q_x(y)=r_x(y)-C_x,$ it follows that $\wh{q_x}(0^n) = \E[q_x(\by)] \notin [-\rho C_x, \rho C_x].$
Now we consider two possibilities.  First, if $0 \leq \E[r_x(\by)] \leq (1-\rho)C_x$, then since $C_x \geq {\frac {\E[r_x(\by)]}{1-\rho}}$ we have
\[
|\wh{q_x}(0^n)| \geq \rho C_x \geq {\frac \rho {1-\rho}} \E[r_x(\by)] \geq \rho |\wh{r_x}(0^n)|
\]
(where we used the nonnegativity of $r_x$ for the last inequality).
Second, if $(1+\rho)C_x \leq \E[r_x(\by)]$, then since $-C_x \geq {\frac {-\E[r_x(\by)]}{1+\rho}}$ we have
\[
\wh{q_x}(0^n)=\E[q_x(\by)] = \E[r_x(\by)] - C_x \geq \E[r_x(\by)] \left( 1 -  {\frac 1 {1+\rho}}\right) \geq (\rho/2) \E[r_x(\by)] = (\rho/2)|\wh{r_x}(0^n)|
\]
(where the second inequality uses $0 < \rho < 1/100$). So in both cases we have $|\wh{q_x}(0^n)| \geq (\rho/2)|\wh{r_x}(0^n)|$, and hence we have that
\begin{equation} \label{eq:qxemptysetnottoosmall}
|\wh{q_x}(0^n)|  \geq (\rho/2) |\wh{r_x}(0^n)|
\end{equation}
In words, this says that  the constant Hermite coefficient of $q_x$ is (multiplicatively) smaller than the constant Hermite coefficient of $r_x$ by at most a  $\rho/2$ multiplicative factor.

To finish the proof, we observe that since the polynomial $r_x$ is $(\hypergenerictwo,\kappa)$-attenuated where $\kappa \leq 1/10$, it must be the case that the magnitude of the constant Hermite coefficient $|\wh{r_x}(0^n)|$ contributes the lion's share of  $\|r_x\|$ --- to be more precise, it must hold that
$\wh{r_x}(0^n)^2 \geq \|r_x\|^2/2.$  Since $\HV_\hypergenerictwo[q_x^{>0}]=\HV_\hypergenerictwo[r_x^{>0}]$ (since $q_x$ and $r_x$ differ only in the constant term $C_x$), and $\|q_x\| \geq |\wh{q_x}(0^n)|$, recalling \Cref{eq:qxemptysetnottoosmall} we can infer from $r_x$ being $(\hypergenerictwo,\kappa)$-attenuated that $q_x$ is $(\hypergenerictwo,8\kappa/\rho^2)$-attenuated.
This proves \Cref{claim:r-to-q}. \qed

}

\ignore{
\subsection{Proof of \Cref{thm:noise-insensitivity} using \Cref{lem:our-lem11}} \label{sec:proof-of-thm-noise-insens-using-our-lem-11}

Fix an $i \in \{0,\dots,\degree\}$ and a $j \in \{0,\dots,\maxcol-1\}$.  We apply \Cref{lem:our-lem11} taking the nice distribution $(g_\bup)_{\bup \sim \Upsilon}$ to be ${\cal F}_{i,j}$, taking its parameter ``$\noisegeneric'$'' to be
$\noisefixed$, taking its parameter ``$\rho$'' to be $\horizclose$, and taking its parameter ``$\mu$'' to be $\epsprg/(8(\degree+1)\maxcol)$. Recalling the setting of $\noisefixed$ from \Cref{eq:noisefixed}, we see that the conditions of \Cref{lem:our-lem11} are indeed satisfied. \rasnote{CHECK/ENFORCE THIS WHEN $\horizclose$ is fixed --- this requires that $\noisefixed \leq \horizclose^2 \cdot (\epsprg/d)^{O(\log d)}$}  By \Cref{eq:alt-sij} we have that $\st_{i,j}(x)=\Ex_{\bup}[g_{\bup}(x)^2]$, and by \Cref{eq:stats-as-zooms} we have that
\[
\Ex_{\by \sim \normal(0,1)^n}
\bracks*{\Ex_{\bup}\bracks*{
\zoom{(g_{\bup})}{\noisegeneric'}{x}(\by)^2}}=
\Ex_{\bup}
\bracks*{\Ex_{\by}\bracks*{
\zoom{(g_{\bup})}{\noisegeneric'}{x}(\by)^2}} =
\Ex_{\bup}
\bracks*{\norm*{\zoom{(g_{\bup})}{\noisegeneric'}{x}}_2^2
}=\st_{i,j+1}(x).
\]
So \Cref{lem:our-lem11} gives that except with probability at most $\epsprg/(8(\degree+1)\maxcol)$ over $\bx \sim \normal(0,1)^n,$ we have
\[
\st_{i,j}(\bx) \approx_{\horizclose} \st_{i,j+1}(\bx),
\]
and \Cref{thm:noise-insensitivity} is proved. \qed
}

We proceed with the proof of \Cref{thm:noise-insensitivity}.  We begin by recording a simple corollary of \Cref{lem:kane-lemma9-alt}:

\begin{corollary}                                       \label{cor:kane-lemma-9}
    In the setting of \Cref{lem:kane-lemma9-alt}, say that $x \in \R^n$ is ``good'' if
    \[
        \text{except with probability at most $0.1$ over $\by \sim \normal(0,1)^n$ we have} \quad \zoom{g}{\noisegeneric}{x}(\by) \approx_{\gamma} g(x).
    \]
    Then $\bx \sim \normal(0,1)^n$ is good except with probability $10\beta$.
\end{corollary}

Next, combining \violet{\Cref{main theorem}}\ignore{\rasnote{Was ``our Local Hyperconcentration Theorem, \Cref{thm:RZL}''}} with \Cref{prop:atten-hyperconcy},\ignore{\rodnote{We should already have this combination as a Corollary somewhere; fix later.}} we derive the following:
\begin{proposition}
    Let $g : \R^n \to \R$ be a degree-$\degree$ polynomial and let $\hypergeneric \geq \sqrt{2}$, $\theta \leq 1$, $\gamma \leq 1$. Say that $x \in \R^n$ is ``well-behaved''  if
    \[
        \text{except with probability $(2\sqrt{\theta}/\gamma)^{\frac12 \hypergeneric^2 + 1}$ over $\by \sim \normal(0,1)^n$ we have} \quad \zoom{g}{\noisegeneric}{x}(\by) \approx_\gamma \E[\zoom{g}{\noisegeneric}{x}].
    \]
    Then $\bx \sim \normal(0,1)^n$ is well-behaved except with probability $\beta$, provided
    \violet{$
        \noisegeneric \leq \frac{c \theta}{\hypergeneric^2} \cdot \parens*{\frac{\beta}{\degree}}^{9/2}.
    $}
    \ignore{\rasnote{Was 
    ``\violet{$
        \noisegeneric \leq \frac{\theta}{\hypergeneric^2} \cdot \parens*{\frac{\beta}{\degree}}^{C \log \degree}.
    $}''}}
\end{proposition}
Given some $\gamma \leq 1$, let us take
\[
    \theta = .01\gamma^2, \quad \hypergeneric = \sqrt{2}.
\]
It follows that if $x \in \R^n$ is both good and well-behaved, then
\[
    \Pr[\zoom{g}{\noisegeneric}{x}(\by) \not \approx_{\nu} g(x)] \leq .1, \quad  \Pr[\zoom{g}{\noisegeneric}{x}(\by) \not \approx_\gamma \E[\zoom{g}{\noisegeneric}{x}]] \leq .04.
\]
Since $.1 + .04 < 1$, the only way this can happens is that $g(x) \approx_{\nu + \gamma} \E[\zoom{g}{\noisegeneric}{x}]$.  Thus the above two propositions imply that except with probability at most $10\beta + \beta = 11\beta$ over $\bx \sim \normal(0,1)^n$, we have
\[
    g(\bx) \approx_{\nu + \gamma} \E[\zoom{g}{\noisegeneric}{\bx}],
\]
provided \violet{$\noisegeneric \leq \gamma^2 (\beta/d)^{O(1)}$}\ignore{\rasnote{Was ``\violet{$\noisegeneric \leq \gamma^2 (\beta/d)^{O(\log d)}$}''}}.  Selecting $\gamma = \horizclose/2$ and $\beta = \epsprg/(88(\degree+1)\maxcol)$, we conclude (recalling that $\maxcol = \poly(d)$) that
\ignore{\rasnote{ The part before the $\implies$ was ``$\noisegeneric \leq \violet{\horizclose^2 (\epsprg/d)^{O(\log d)}}$''}}
\[
    \noisegeneric \leq \violet{\horizclose^2 (\epsprg/d)^{O(1)}} \quad\implies\quad \Pr[g(\bx) \approx_{\horizclose} \E[\zoom{g}{\noisegeneric}{\bx}]] \leq \epsprg/(8(\degree+1)\maxcol).
\]
Applying this with $g = \st_{i,j}$ completes the proof of \Cref{thm:noise-insensitivity}. \qed 



\section{Proof of \Cref{thm:hybridstep}: one step of the Replacement Method} \label{sec:hybridstep}

In this section we define a collection of ``analysis checks,'' which are inequalities among the statistics, and explain the high-level structure of the proof of \Cref{thm:hybridstep}. The analysis checks  play a crucial role in the proof of \Cref{thm:hybridstep}: as we explain in \Cref{sec:high-level-one-step}, two very different arguments (corresponding to \Cref{lem:page14} and \Cref{lem:page13}) are used to establish the conclusion of \Cref{thm:hybridstep} at a given $x \in \R^n$, depending on whether or not all of the analysis checks hold at that $x$.

\subsection{Analysis checks} \label{sec:analysischecks}

In this subsection we define our set of ``analysis checks,'' which we denote $\AnalysisChecks$.  They are related to, but somewhat different from, the mollifier checks $\MollifierChecks$ that were used to define the mollifier $\allsoftchecks$ in \Cref{sec:defining-allsoftchecks}.

One difference between the analysis checks and the mollifier checks is that since the mollifier checks needed to be ``actually encoded into the mollifier,'' each one needed to consist of both an inequality $\ineq$ among the statistics and a ``softness'' parameter $\delta$.  In contrast, the analysis checks only play a role in our analysis and do not need to be encoded in the mollifier, and for this reason each analysis check consists only of an inequality $\ineq$ among the statistics.  Other than this, the difference between the analysis checks and mollifier checks is that the analysis checks essentially correspond to the mollifier checks ``shifted right by one in the grid.''

Below we describe the analysis checks in more detail and highlight the difference between them and the mollifier checks.

\begin{definition}
    The set $\AnalysisChecks$ contains the following checks (inequalities among statistics):
    \begin{itemize}
        \item \textbf{The horizontal checks:} for every $0 \leq i \leq  d$, for every $1 \leq j \leq \maxcol-1$, we check that
                    \begin{equation}    \label{eqn:horz-checks}
                \st_{i,j}(x) \approx_{\horzanal} \st_{i,j+1}(x)
            \end{equation}
(so $2(\degree+1)(\maxcol-1)$ inequalities in total for the horizontal checks), where we set
\begin{equation} \label{eq:horzanal}
\horzanal \coloneqq \frac{1}{100d\maxcol}.
\end{equation}
(Looking ahead, we note that this choice of $\horzanal$ is less than the upper bound on $\gamma$ imposed by the noise insensitivity extension lemma, \Cref{lem:magicker-lemma}, which is the main result of \Cref{sec:magic-lemma}.) Note also that while the mollifier checks defined in \Cref{sec:MollifierChecks-def} check $\st_{i,j}(x)$ against $\st_{i,j+1}(x)$ for $j=0,\dots,\maxcol-2$, here we are checking $\st_{i,j}(x)$ against $\st_{i,j+1}(x)$ for $j=1,\dots,\maxcol-1$.  Thus these checks correspond precisely to the mollifier's noise-insensitivity checks, but ``shifted to the right by one.''
        \item \textbf{The diagonal checks:} for \emph{just} the $1$st column (note, \emph{not} the $0$th column), for all $i=0,\dots,d-1,$ we check that
\begin{equation} \label{eqn:diag-checks}
                \st_{i+1,1}(x) \leq 100 \anticoncgap \st_{i,2}(x)
\end{equation}
(so $d$ diagonal checks in total).
Note that while the strong anticoncentration checks defined in \Cref{sec:MollifierChecks-def} check $\st_{i+1,0}(x)$ against $\st_{i,1}(x)$, here we are checking $\st_{i+1,1}$ against $\st_{i,2}(x).$ So similar to the previous bullet, these checks correspond precisely to the mollifier's strong anticoncentration checks, but again ``shifted to the right by one.''

    \end{itemize}

    Below we give an illustration of the  analysis checks.     \myfig{0.75}{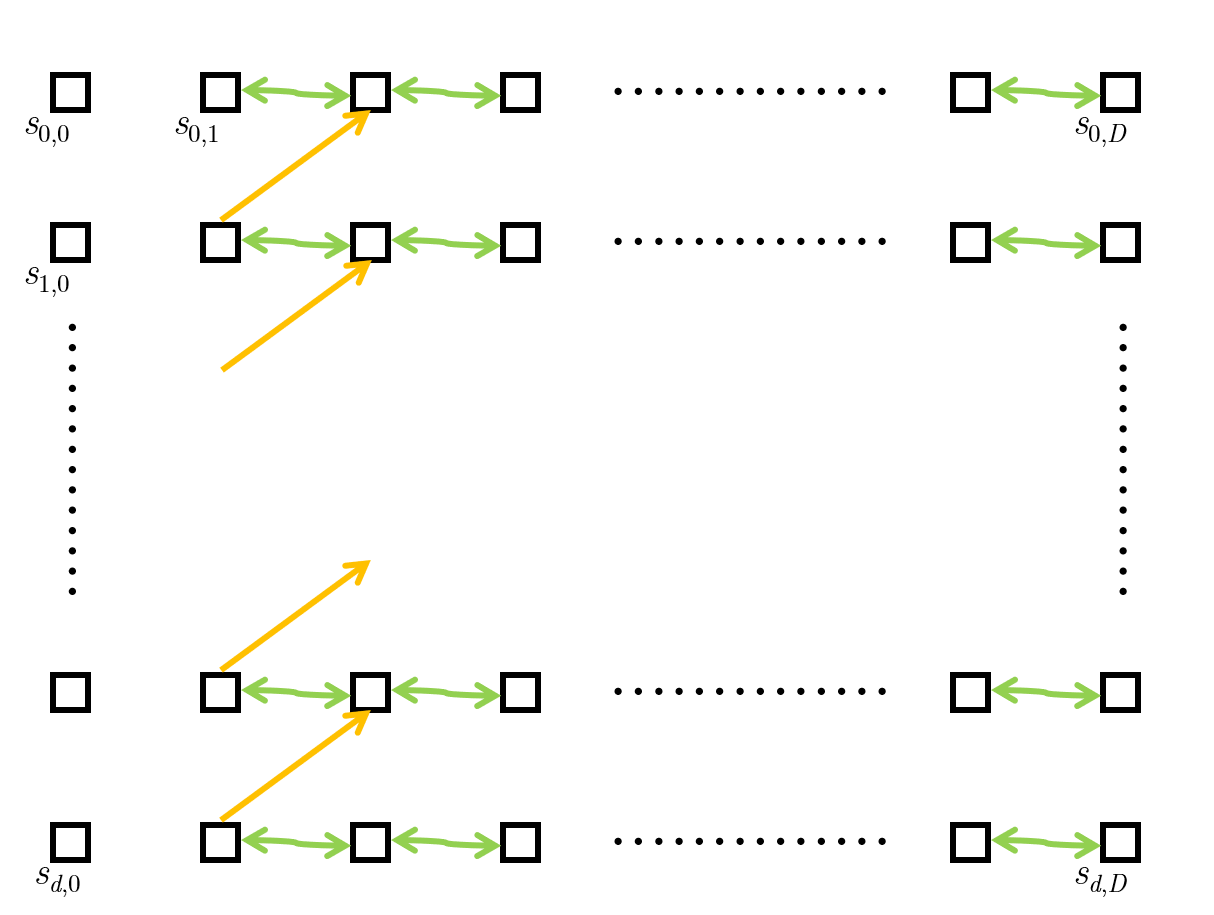}{The analysis checks. Horizontal checks are depicted in green and diagonal checks are depicted in yellow.}{}

\end{definition}

\subsection{High-level structure of the proof of \Cref{thm:hybridstep}} \label{sec:high-level-one-step}

%
%
%
%


Let us define the following small integer parameter,
\begin{equation} \label{eq:taylor}
    \taylor = 4,
\end{equation}
which will be the degree out to which we use Taylor's theorem.

\Cref{thm:hybridstep} will  be an immediate consequence of the below two lemmas, since by \Cref{eq:anticoncgap-horizclose}, \Cref{eq:len-wise}
 and \Cref{eq:noisefixed} we have that
$\poly((\taylor d)^\taylor) \cdot \anticoncgap^{\taylor/2} \ll \epsprg/(8\len) $:


\begin{lemma} [If analysis checks all pass, $\wise$-wise moments determine mollifier's value] \label{lem:page14}
Suppose that $x \in \R^n$ is such that all of the checks in $\AnalysisChecks$ hold at $x$, and that the random vector $\bz$ is a $\wise$-wise independent  $n$-dimensional Gaussian.  Then $\E\bracks*{\Indplus\parens*{\sqrt{1-\noisefixed} \cdot x + \sqrt{\noisefixed}\cdot \bz}}$ is determined up to an additive $\pm \poly((\taylor d)^\taylor) \cdot \anticoncgap^{\taylor/2}$.\ignore{\rodnote{this bound may not be quite right; also the $\poly()$ also hides some factors like $\taylor^\taylor$; gotta find a way to say this briefly}}
\end{lemma}

\begin{lemma} [If an analysis check fails, mollifier is close to zero] \label{lem:page13}
Suppose that $x \in \R^n$ is such that some check in $\AnalysisChecks$ does not hold at $x$, and that the random vector $\bz$ is a $\wise$-wise independent  $n$-dimensional Gaussian.  Then $\E\bracks*{\Indplus\parens*{\sqrt{1-\noisefixed} \cdot x + \sqrt{\noisefixed}\cdot \bz}} \in [0,\poly((\taylor d)^\taylor) \cdot \anticoncgap^{\taylor/2}]$.
\end{lemma}

As we will see in the following sections, two very different arguments are used to prove \Cref{lem:page14} and \Cref{lem:page13}. \Cref{lem:page14}, which corresponds to the case in which $x$ is such that all of the analysis checks pass, is based on a Taylor's theorem argument.  In contrast, \Cref{lem:page13}, which corresponds to the case in which $x$ is such that some analysis check fails, employs a delicate argument, which takes advantage of the careful way that the analysis checks are structured vis-a-vis the mollifier checks, to argue that  in this case almost all outcomes of $\bz$ result in $\Indplus\parens*{\sqrt{1-\noisefixed} \cdot x + \sqrt{\noisefixed}\cdot \bz}=0$.

Before we can enter into the proofs of \Cref{lem:page14} and \Cref{lem:page13}, there are several intermediate technical results which will be used in both proofs which we need to establish.\ignore{\rasnote{If we are going to stick in some exposition about what these technical results --- Desideratum \#2, Magic Lemma ---  are, and why they arise in both of the \Cref{lem:page14} and \Cref{lem:page13} proofs, this might be the place to do it. But we don't necessarily need such exposition.}} We state and prove these technical results in  \Cref{sec:hypervariance-bound} and \Cref{sec:magic-lemma}, and prove \Cref{lem:page13} and \Cref{lem:page14} in \Cref{sec:page13} and \Cref{sec:page14} respectively.


\section{\ignore{Desideratum \#2:}Bounding the hypervariance of a statistic by its ``neighbors''}
\label{sec:hypervariance-bound}

The main goal of this section is to prove the following technical result which will be needed for our analysis.  For every $x \in \R^n$, it gives an upper bound on the hypervariance of the zoom-at-$x$ of our $(i,j)$-th statistic in terms of the values of some ``nearby'' statistics:

\begin{theorem} [Bounding hypervariance of zooms \kanenote{(analogue of Proposition~12 of \cite{kane11focs})}] \label{thm:desideratum2}
 For all $i \in \{0,\dots,\degree-1\}$ all $j \in \{0,\dots,\maxcol-1\}$,\ignore{all $\hypergeneric \geq 1$,\rodnote{This should actually be about $\hyperfixed$, right?  Because $\hyperfixed$ is built into the definition of the statistics.}} and all $x \in \R^n$, it holds that
    \begin{align*}
        \HV_{\frac{\sqrt{\hyperfixed}}{13}}\bracks*{\zoom{(\st_{i,j})}{\noisefixed}{x}}
        &\leq 8\parens*{\st_{i,j+1}(x) + \st_{i+1,j}(x)} \cdot  \st_{i+1,j}(x) \\
        &\leq 16\max\{\st_{i,j+1}(x)  \st_{i+1,j}(x), \st_{i+1,j}(x)^2\}.
    \end{align*}
\end{theorem}
We stress that \Cref{thm:desideratum2} holds for \emph{every} input $x \in \R^n$. This is important because \Cref{thm:desideratum2} will be used to prove \Cref{thm:hybridstep}, and in that setting we are dealing with an arbitrary $x \in \R^n$.

\ignore{
}

\begin{remark} \label{rem:our-hypervariance-bound-versus-kanes}
\Cref{thm:desideratum2} is analogous to Proposition~12 of \cite{kane11focs}, which  upper bounds the \emph{variance} of the zoom-at-$x$ of \cite{kane11focs}'s $(i,j)$-th statistic in terms of the values of the $(i+1,j)$-th and $(i,j+1)$-th statistics.  However, there is a factor of $2^\degree$ present in the bound of \cite{kane11focs} (again because of hypercontractivity) which as always is incompatible with our goal of achieving an overall quasipolynomial rather than exponential dependence on $\degree$.
\end{remark}

Before entering into the proof of \Cref{thm:desideratum2}, we record some corollaries and related results which we will use  in \Cref{sec:page13,sec:page14}.  Fix any $x\in \R^n$, let $\bz$ denote a $\wise$-wise independent $n$-dimensional Gaussian random vector, and let us write $\bw$ to denote $\sqrt{1-\noisefixed} x + \sqrt{\noisefixed} \bz$.
Let us also introduce the notation
\[
    \ul{\st}_{i,j} = \st_{i,j}(x), \qquad \bst_{i,j} = \st_{i,j}(\bw) = \zoom{(\st_{i,j})}{\noisefixed}{x}(\bz).
\]

We first record the following:

\begin{fact}
\label{fact:themean}
For all $0 \leq i \leq d$ and all $0 \leq j \leq \maxcol-1$, we have that
\begin{equation} \label{eqn:themean}
    \ul{\st}_{i,j+1} = \E[\bst_{i,j}] = \abs{\E[\bst_{i,j}]}.
\end{equation}
\end{fact}
\begin{proof}
We have
\[
\ul{\st}_{i,j+1} =
\Ex_{\by \sim \normal(0,1)^n}[\zoom{(\st_{i,j})}{\noisefixed}{x}(\by)] =
\Ex_{\by \sim \normal(0,1)^n}[\zoom{(\st_{i,j})}{\noisefixed}{x}(\bz)] =
\E[\bst_{i,j}] = \abs{\E[\bst_{i,j}]},
\]
where the first equality is \Cref{def:sij}, the second is because $\st_{i,j}$ is a polynomial of degree at most $2 \degree$ and $\wise \geq 2 \degree$, and the last is by the non-negativity of $\st_{i,j}.$
\end{proof}

Next, as a corollary of  \Cref{thm:desideratum2} we have the following:
\begin{corollary} \label{cor:usefulhvbound}
For all $0 \leq i \leq d-1$ and all $0 \leq j \leq \maxcol-1$,
\begin{equation}    \label[ineq]{ineq:useme}
   \bst_{i,j} \text{ is } (2\taylor, 4\max\{\sqrt{\zeta_{i,j}}, \zeta_{i,j}\})\text{-hyperconcentrated,}
   \quad \quad \text{where~}\zeta_{i,j} \coloneqq \frac{\ul{\st}_{i+1,j}}{\ul{\st}_{i,j+1}}
\end{equation}
(where we interpret $0/0=0$).
\end{corollary}

\begin{proof}
We begin by noting that if $\ul{\st}_{i,j+1}=0$, then (recalling \Cref{rem:semigroup} and \Cref{def:si0}) it must be the case that $\st_{i,0}$ is the constant-0 polynomial and hence $\st_{i+1,j}$ is also zero; in this case $\bst_{i,j}$ is the identically-0 random variable, which is certainly $(2\taylor,0)$-hyperconcentrated. Hence we subsequently assume that $\ul{\st}_{i,j+1} >0.$

By \Cref{thm:desideratum2}, we have that
\begin{equation} \label{eq:lolly}
    \HV_{\frac{\sqrt{\hyperfixed}}{13}}\bracks*{\zoom{(\st_{i,j})}{\noisefixed}{x}}
    \leq 16\max\{\ul{\st}_{i,j+1} \ul{\st}_{i+1,j}, \ul{\st}_{i+1,j}^2\}.
\end{equation}
We apply \Cref{lem:hypervar-to-hypercon} to the function $\zoom{(\st_{i,j})}{\noisefixed}{x}$, observing that the ``$\mu$'' of \Cref{lem:hypervar-to-hypercon} is $\Ex_{\by \sim \normal(0,1)^n}[\zoom{(\st_{i,j})}{\noisefixed}{x}(\by)]=\ul{\st}_{i,j+1} $ and hence that \Cref{eq:lolly} lets us take the ``$\theta$'' of \Cref{lem:hypervar-to-hypercon} to be
$16\max\{\zeta_{i,j},\zeta_{i,j}^2\}.$
Since $1+\hyperfixed/13 \geq {2}\taylor$,  \Cref{lem:hypervar-to-hypercon} thus gives that (for $\by \sim \normal(0,1)^n$) the random variable $\zoom{(\st_{i,j})}{\noisefixed}{x}(\by)$ is $({2}\taylor, 4\max\{\sqrt{\zeta_{i,j}}, \zeta_{i,j}\})$-hyperconcentrated. Since $\zoom{(\st_{i,j})}{\noisefixed}{x}(y)$ is a polynomial of degree at most $2\degree$ and $4 \degree \taylor = \wise$, the $2 \taylor$-th moments of $\zoom{(\st_{i,j})}{\noisefixed}{x}(\by)$ and of $\bst_{i,j}$ are identical (recall \Cref{rem:why-wise}). Now by the definition of hyperconcentration of a random variable we get that \Cref{ineq:useme} holds as desired.
\end{proof}

\ignore{
%
%
%
}

\subsection{Proof of \Cref{thm:desideratum2}} \label{sec:proof-of-desideratum2}

\Cref{thm:desideratum2} is proved using the following two results. We note that each of these results holds in a fairly general setting: in \Cref{lem:a-lemma} $(q_{\upsilon})_{\upsilon \in \Upsilon}$ can be any nice distribution of polynomials, and in fact both lemmas hold for  polynomials over either Boolean space or Gaussian space (it will be clear from the proofs that they go through essentially unchanged in the Boolean context).

\begin{lemma} \label{lem:a-lemma}
    Let $(q_{\bup})_{\bup \sim \Upsilon}$ be a nice distribution over polynomials (as defined in \Cref{sec:distributional}) over Gaussian space.  Then\ignore{ for all $\hypergeneric \geq \sqrt{3}$,} for $\hypergeneric_0 \coloneqq \tfrac{1}{13}\hyperfixed^{1/4}$, we have
    \begin{align*}
        \HV_{\hypergeneric_0}\bracks*{\E_{\bup}\bracks*{q_{\bup}^2}}
           & \leq 8  \parens*{\E_{\bup}\bracks*{\|q_{\bup}\|_2^2} + \E_{\bup}\bracks*{\HyperVar_{\hyperfixed}[q_{\bup}]}}
                         \cdot
                    \E_{\bup}\bracks*{\HyperVar_{\hyperfixed}[q_{\bup}]}.
    \end{align*}
\end{lemma}

\begin{lemma} \label{lem:LHS-RHS}
\footnote{We note at this point that \cite{kane11focs} appears to have a gap in the proof of its Proposition~12. Specifically, it is not true that the second equality following ``Notice that'' in that proof holds true (roughly speaking, because the derivative operator $D$ and the noise operator $N$ of \cite{kane11focs} do not in general commute, as can be verified by considering the polynomial $p(x)=x_1 x_2)$. The raison d'\^{e}tre of our \Cref{lem:LHS-RHS} is to fill this gap.}
For all $0 \leq i \leq \degree$ and all $0 \leq j \leq \maxcol-1$,\ignore{all $\hypergeneric \geq 1$,} and all $x \in \R^n$, we have
\[
 \E_{\boldf \sim \calF_{i,j}}[\HV_\hyperfixed[\zoom{\boldf}{\noisefixed}{x}]]
 \leq
 \st_{i+1,j}(x).
\]
\end{lemma}


\begin{proofof}{\Cref{thm:desideratum2} using \Cref{lem:a-lemma} and \Cref{lem:LHS-RHS}}
We instantiate \Cref{lem:a-lemma} by taking the nice distribution $(q_{\bup})_{\bup \sim \Upsilon}$ to be $(\zoom{\boldf}{\noisefixed}{x})_{\boldf \sim {\cal F}_{i,j}}.$
With this choice $\HV_{\hypergeneric_0}\bracks*{\E_{\bup}\bracks*{q_{\bup}^2}}$ corresponds to $\HV_{\hypergeneric_0}\bracks*{\zoom{(\st_{i,j})}{\noisefixed}{x}}$ (by \Cref{eq:alt-sij}), and $\E_{\bup}\bracks*{\|q_{\bup}\|_2^2}$ corresponds to $\st_{i,j+1}(x)$ (by \Cref{eq:stats-as-zooms}).  The final quantity on the right-hand side of \Cref{lem:a-lemma}, $\E_{\bup}\bracks*{\HyperVar_{\hyperfixed}[q_{\bup}]}$, corresponds to
$\E_{\boldf \sim {\cal F}_{i,j}}[\HyperVar_{\hyperfixed}[\zoom{\boldf}{\noisefixed}{x}]]$, so from \Cref{lem:a-lemma} we get that
\begin{align}
\HV_{\hyperfixed_0}\bracks*{\zoom{(\st_{i,j})}{\noisefixed}{x}}
&\leq
8 \parens*{\st_{i,j+1}(x) + \E_{\boldf \sim {\cal F}_{i,j}}[\HyperVar_{\hyperfixed}[\zoom{\boldf}{\noisefixed}{x}]]} \cdot \E_{\boldf \sim {\cal F}_{i,j}}[\HyperVar_{\hyperfixed}[\zoom{\boldf}{\noisefixed}{x}]] \nonumber\\
& \leq
8 \parens*{\st_{i,j+1}(x) + \st_{i+1,j}(x)} \cdot \st_{i+1,j}(x) \tag{\Cref{lem:LHS-RHS}}
\end{align}
and the proof of \Cref{thm:desideratum2} is complete.
\end{proofof}


\subsection{Proof of \Cref{lem:a-lemma}}

Our main goal in this section is to prove the following:

\begin{lemma} \label{lem:a-lemma2}
    Let $(q_{\bup})_{\bup \sim \Upsilon}$ be a nice distribution over polynomials (as defined in \Cref{sec:distributional}) over Gaussian space.
Define the function
\[
            g = \E_{\bup}\bracks*{q_{\bup}^2}.
\]

    Then if $\hypergeneric_0 > {\frac {3^{1/8}}{13}}$,  we have
    \[
        \HV_{\hypergeneric_0}\bracks*{g}
            \leq 8
        \E_{\bup}\bracks*{\|\U_{\sqrt{3}} q_{\bup}\|_2^2}
            \cdot
        \E_{\bup}\bracks*{\HyperVar_\hypergeneric[q_{\bup}]}
    \]
    where $\hypergeneric = 3e^4 \hypergeneric_0^{{4}} \leq 165 \hypergeneric_0^{{4}}$.

    %
    %
    %
\end{lemma}

\Cref{lem:a-lemma} follows from \Cref{lem:a-lemma2} since $13^2 > 165$,  $\HV_{S}[p]$ is an increasing function of $S$ for all $p$, and
\begin{align*}
\|\U_{\sqrt{3}} q_{\upsilon}\|_2^2
&= \sum_{|\alpha|\geq 0} 3^{|\alpha|} \wh{q}_{\upsilon}(\alpha)^2
\tag{\Cref{eqn:U-formula} and Plancherel}\\
&= \sum_{|\alpha| \geq 0} \wh{q}_{\upsilon}(\alpha)^2 +
\sum_{|\alpha| \geq 1} (3^{|\alpha|} - 1) \wh{q}_{\upsilon}(\alpha)^2\\
&= \|q_{\upsilon}\|_2^2 + \HV_{\hyperfixed}[q_{\upsilon}]
\tag{\Cref{def:hypervar} and choice of $\hyperfixed$}
\end{align*}

We will use the following lemma in the proof of \Cref{lem:a-lemma2}:

\begin{lemma} \label{lem:HF-bound}
Given any polynomial $q$ and any $k \in \N^+$, we have
\[
 \W^{\geq k}[q^2] \leq 4 \cdot \norm*{\U_{\sqrt{3}} (q^{\geq k/2})}_2^2 \cdot \norm*{\U_{\sqrt{3}} q}_2^2.
\]
\end{lemma}
\begin{proof}
Let us write $q = \ell + h$ where $\ell = q^{< k/2}$ and $h = q^{\geq k/2}$.  Since $(\ell^2)^{\geq k}=0$, we have that
\[
    (q^2)^{\geq  k}  = (\ell^2 + 2 \ell h + h^2)^{\geq k} = (\ell^2)^{\geq  k} + (2\ell h)^{\geq  k}  + (h^2)^{\geq  k} = (h(2\ell  + h))^{\geq  k},
\]
and thus
\[
    \W^{\geq k}[q^2] = \norm*{(q^2)^{\geq k}}_2^2  = \norm*{(h(2\ell + h))^{\geq k}}_2^2  \leq \norm*{h(2\ell + h)}_2^2 = \E[h^2 (2\ell + h)^2] \leq \norm*{h^2}_2 \cdot \norm*{(2\ell + h)^2}_2,
\]
where the latter inequality is by Cauchy--Schwarz.  By $(2,4)$-hypercontractivity (recalling \Cref{thm:hypercon}), we have that
\[
    \|h^2\|_2 \leq \|\U_{\sqrt{3}} h\|_2^2 
\]
and similarly
\[
    \|(2\ell+h)^2\|_2 \leq \|\U_{\sqrt{3}} (2\ell+h)\|_2^2 = 4\|\U_{\sqrt{3}} \ell\|_2^2 + \|\U_{\sqrt{3}} h\|_2^2, 
\]
where the last equality holds because $\ell$ and $h$ are orthogonal.  Thus we have shown that
\begin{align*}
    \W^{\geq k}[q^2] &\leq
     \norm*{\U_{\sqrt{3}} (q^{\geq k/2})}_2^2 \cdot \parens*{4 \|\U_{\sqrt{3}} \ell\|_2^2 +  \|\U_{\sqrt{3}} h\|_2^2}\\
     &\leq  \norm*{\U_{\sqrt{3}} (q^{\geq k/2})}_2^2 \cdot \parens*{4 \|\U_{\sqrt{3}} \ell\|_2^2 +  4\|\U_{\sqrt{3}} h\|_2^2}\\
     &= 4 \cdot \norm*{\U_{\sqrt{3}} (q^{\geq k/2})}_2^2 \cdot \norm*{\U_{\sqrt{3}} q}_2^2.
\end{align*}

\end{proof}

\begin{proofof}{\Cref{lem:a-lemma2}}
    We have
    \begin{align*}
            \HV_{\hypergeneric_0}\bracks*{g}
        &= \sum_{k \geq 1} {\hypergeneric_0}^{{2}k} \cdot \norm*{g^{=k}}_2^2 \tag{\Cref{def:hypervar}} \\
        &\leq \sum_{k \geq 1} (e {\hypergeneric_0})^{{2}k} \cdot \norm*{g^{=k}}_1^2 \tag*{(\Cref{thm:aobf-9.22})} \\
        &= \sum_{k \geq 1} (e {\hypergeneric_0})^{{2}k} \cdot \E_{\bx\sim \normal(0,1)^n}\bracks*{\abs*{g^{=k} (\bx)}}^2.
        \tag{definition of one-norm}
    \end{align*}
    Next we observe that
    \[
           \abs*{g^{=k} (x)}
        = \abs*{\parens*{\E_{\bup}\bracks*{q_{\bup}^2}}^{=k}(x)}
        = \abs*{\E_{\bup}\bracks*{\parens*{q_{\bup}^2}^{=k}}(x)}
        \leq \E_{\bup}\bracks*{\abs*{\parens*{q_{\bup}^2}^{=k}(x)}},
    \]
  where the second equality holds because the operator which maps $p$ to $p^{=k}$ (i.e.~projecting to the $k$-th Wiener chaos) is a linear operator, and the inequality is the triangle inequality.
  Continuing the above, we deduce that
    \begin{align*}
            \HV_{\hypergeneric_0}\bracks*{g}
        &\leq \sum_{k \geq 1} (e {\hypergeneric_0})^{{2}k} \cdot \E_{\bx}\bracks*{  \E_{\bup}\bracks*{\abs*{\parens*{q_{\bup}^2}^{=k}(x)}} }^2 \\
        &= \sum_{k \geq 1} (e {\hypergeneric_0})^{{2}k} \cdot \E_{\bup}\bracks*{ \norm*{\parens*{q_{\bup}^2}^{=k}}_1 }^2 \\
        &= \sum_{k \geq 1} (e {\hypergeneric_0})^{{2}k} \cdot \E_{\bup}\bracks*{ \norm*{\parens*{q_{\bup}^2}^{=k}}_2 }^2
      \tag*{(monotonicity of norms)} \\
        &\leq \sum_{k \geq 1} (e {\hypergeneric_0})^{{2}k} \cdot \E_{\bup}\bracks*{\sqrt{\W^{\geq k}[ q_{\bup}^2]} }^2. \tag*{(as $\norm*{h^{=k}}_2^2 \leq \norm*{h^{\geq k} }_2^2$ for any $h$)}
    \end{align*}
Now we apply \Cref{lem:HF-bound} to each $q_{\upsilon}$, which lets us continue as follows:
    \begin{align*}
            \HV_{\hypergeneric_0}\bracks*{g}
        &\leq  \sum_{k \geq 1} (e {\hypergeneric_0})^{{2}k} \cdot \E_{\bup}\bracks*{2\sqrt{\norm*{\U_{\sqrt{3}} (q_{\bup}^{\geq k/2})}_2^2 \cdot \norm*{\U_{\sqrt{3}} q_{\bup}}_2^2    } }^2 \\
        &\leq  4\sum_{k \geq 1} (e {\hypergeneric_0})^{{2}k} \cdot \E_{\bup}\bracks*{\norm*{\U_{\sqrt{3}} (q_{\bup}^{\geq k/2})}_2^2} \cdot \E_{\bup}\bracks*{ \norm*{\U_{\sqrt{3}} q_{\bup}}_2^2} \tag*{(Cauchy--Schwarz)}\\
        &=  4 \E_{\bup}\bracks*{ \norm*{\U_{\sqrt{3}} q_{\bup}}_2^2}  \cdot \E_{\bup}\bracks*{\sum_{k \geq 1} (e {\hypergeneric_0})^{{2}k} \cdot \norm*{\U_{\sqrt{3}} (q_{\bup}^{\geq k/2})}_2^2} \\
        &=  4 \E_{\bup}\bracks*{ \norm*{\U_{\sqrt{3}} q_{\bup}}_2^2}  \cdot \E_{\bup}\bracks*{\sum_{k \geq 1} (e {\hypergeneric_0})^{{2}k} \cdot \sum_{j \geq k/2} 3^j \W^{=j}[q_{\bup}]} \tag{definition of $\norm*{\U_{\sqrt{3}} (q_{\bup}^{\geq k/2})}_2^2$} \\
        &\leq  8 \E_{\bup}\bracks*{ \norm*{\U_{\sqrt{3}} q_{\bup}}_2^2}  \cdot \E_{\bup}\bracks*{\sum_{j \geq 1} (3e^4 {\hypergeneric_0}^{{4}})^j \cdot \W^{=j}[q_{\bup}]} \\
        &=  8 \E_{\bup}\bracks*{ \norm*{\U_{\sqrt{3}} q_{\bup}}_2^2} \cdot \E_{\bup}\bracks*{\HV_{3e^4{\hypergeneric_0}^{{4}}}[q_{\bup}]},
    \end{align*}
    completing the proof.
\end{proofof}

\subsection{Proof of \Cref{lem:LHS-RHS}} \label{sec:sick}

%

We begin by re-expressing the right-hand side of \Cref{lem:LHS-RHS}:
\begin{align}
   \st_{i+1,j}(x)&= \Ex_{\boldf \sim \calF_{i+1, j}}[\boldf(x)^2] \tag{\Cref{eq:alt-sij}} \nonumber\\
    &= \Ex_{\boldf \sim \calF_{i+1,0}}
        \E_{\by_1, \dots, \by_j} (\opnoise_{\by_1}\cdots \opnoise_{\by_j}  \boldf)^2 \tag{definition of ${\cal F}_{i+1,j}$}\nonumber \\
    &= \E_{\boldf \sim \calF_{i+1,0}}\bracks*{\E_{\bz}  [\boldf_{(1-\noisefixed)^j|\bz}(x)^2]} \tag{semigroup property / \Cref{rem:semigroup}} \nonumber\\
    &= \E_{\boldf \sim \calF_{i,0}}\Ex_{\by,\by'}\bracks*{\E_{\bz}  \bracks*{\parens*{(\nderiv{\by}{\by'}{\hyperfixed,\noisefixed} \boldf)_{(1-\noisefixed)^j|\bz}(x)}^2}},
    \label{eq:useme}
\end{align}
where the last equality is by the definition of ${\cal F}_{i+1,0}$ in terms of ${\cal F}_{i,0}$\ignore{\rasnote{This a good justification? Or should we invoke \Cref{eq:stats-as-zooms} and \Cref{fact:deriv-hypervar}?}}.

We similarly re-express the left-hand side of \Cref{lem:LHS-RHS}:
\begin{align}
	\E_{\boldf \sim \calF_{i,j}}[\HV_\hyperfixed[\zoom{\boldf}{\noisefixed}{x}]]  &=    \E_{\boldf \sim \calF_{i,j}} \E_{\by,\by'}
\bracks*{\parens*{\nderiv{\by}{\by'}{\hyperfixed,\noisefixed} \boldf(x)}^2} \tag{\Cref{fact:deriv-hypervar}} \nonumber\\
    &=  \Ex_{\boldf \sim \calF_{i,0}}
        \E_{\by_1, \dots, \by_j} \E_{\by,\by'}
            \bracks*{\parens*{(\nderiv{\by}{\by'}{\hyperfixed,\noisefixed} \opnoise_{\by_1}\cdots \opnoise_{\by_j}  \boldf)(x)}^2}  \tag{definition of ${\cal F}_{i,j}$} \nonumber\\
    &=  \Ex_{\boldf \sim \calF_{i,0}}
        \E_{\by,\by'} \E_{\bz}
            \bracks*{\parens*{(\nderiv{\by}{\by'}{\hyperfixed,\noisefixed} \zoom{\boldf}{(1-\noisefixed)^j}{\bz})(x)}^2} ,\label{eq:abuseme}
\end{align}
where the last equality is by the semigroup property / \Cref{rem:semigroup}.
Comparing \Cref{eq:useme} and \Cref{eq:abuseme}, \Cref{lem:LHS-RHS} is an immediate consequence of \Cref{prop:every-outcome}, stated and proved below (setting its $\rho$ parameter to be $(1-\noisefixed)^j$), which states that the desired inequality holds ``outcome by outcome'' for outcomes of $\boldf \sim \calF_{i,0}.$

\begin{proposition} \label{prop:every-outcome}
For all $\hypergeneric \geq 1$ (and in particular $\hypergeneric = \hyperfixed$), for every polynomial $g: \R^n \to \R$, and every $0 < \rho < 1$, we have that

\begin{equation} \label{eq:every-outcome}
    \E_{\by,\by'} \E_{\bz}\bracks*{  (\nderiv{\by}{\by'}{\hypergeneric,\noisefixed}
\zoom{g}{\rho}{\bz})(x)^2 }
    \leq
    \E_{\by,\by'} \E_{\bz}\bracks*{  \zoom{(\nderiv{\by}{\by'}{\hypergeneric,\noisefixed}
 g)}{\rho}{\bz}(x)^2 }.
\end{equation}
\end{proposition}

Since the proof of \Cref{prop:every-outcome} is somewhat involved we explain the high-level idea underlying it before entering into the technical details. When $R>1$ the quantities in \Cref{eq:every-outcome} are somewhat difficult to work with since the Gaussian noise operator $\U_\hypergeneric$, which is involved in the definition of the $\nderiv{y}{y'}{\hypergeneric,\noisefixed}$ operator, does not admit a convenient probabilistic interpretation (recall that for $\hypergeneric > 1$ the definition of $\U_\hypergeneric$ is through \Cref{eqn:U-formula}).  The proof of \Cref{prop:every-outcome} takes advantage of the fact that for $0 \leq \hypergeneric'\leq 1$, the quantity
$\E_{\by,\by'} \E_{\bz}\bracks*{  (\nderiv{\by}{\by'}{\hypergeneric',\noisefixed}
\zoom{g}{\rho}{\bz})(x)^2 }$ corresponding to the left-hand side of \Cref{eq:every-outcome} does have a natural probabilistic interpretation, and likewise for the quantity $ \E_{\by,\by'} \E_{\bz}\bracks*{  \zoom{(\nderiv{\by}{\by'}{\hypergeneric',\noisefixed}
 g)}{\rho}{\bz}(x)^2 }$ corresponding to the right-hand side.   These probabilistic interpretations let us give  tractable expressions for each of the two quantities, and as we will see, it is evident from these expressions that the corresponding quantities correspond to polynomials in $\hypergeneric'$ of degree at most $2d$.  These polynomials can then be analyzed to show that the left-hand side is indeed at most the right-hand side for all $\hypergeneric \geq 1$, as asserted by the proposition.

\newcommand{\LHS}{\mathrm{LHS}}
\newcommand{\RHS}{\mathrm{RHS}}
\newcommand{\plhs}{p_\LHS}
\newcommand{\prhs}{p_\RHS}

\begin{proofof}{\Cref{prop:every-outcome}}
We define the function $\plhs(\hypergeneric')$ to be
\begin{equation} \label{eq:plhs}
    \plhs(\hypergeneric') \coloneqq     \E_{\by,\by'} \E_{\bz}\bracks*{  (\nderiv{\by}{\by'}{\hypergeneric',\noisefixed}
\zoom{g}{\rho}{\bz})(x)^2 }    \end{equation}
and the function $\prhs(\hypergeneric')$ to be
\begin{equation} \label{eq:prhs}
    \prhs(\hypergeneric') \coloneqq
     \E_{\by,\by'} \E_{\bz}\bracks*{  \zoom{(\nderiv{\by}{\by'}{\hypergeneric',\noisefixed}
 g)}{\rho}{\bz}(x)^2 }.
    \end{equation}

Let
 \[
    h(u) \coloneqq g(\sqrt{\rho}\hypergeneric'\sqrt{1-\noisefixed} x + \sqrt{1-\rho\hypergeneric'^2(1-\noisefixed)} u).
\]

The following two claims provided the probabilistic interpretations alluded to earlier:
\begin{claim} \label{claim:plhs}
For all $0 \leq \hypergeneric' \leq 1$, we have that
\begin{align}
\plhs(\hypergeneric') &=
    \Stab_{\sigma_{\LHS}(\hypergeneric')}[h] - \Stab_{\tau_{\LHS}(\hypergeneric')}[h] = \sum_{\alpha} \wh{h}(\alpha)^2 ((\sigma_{\LHS}(\hypergeneric'))^{|\alpha|} - (\tau_{\LHS}(\hypergeneric'))^{|\alpha|}),  \text{~where} \label{eq:plhs-stab}\\
    \sigma_{\LHS}(\hypergeneric') & \coloneqq \frac{1-\rho + \hypergeneric'^2\noisefixed\rho}{1-\rho \hypergeneric'^2(1-\noisefixed)},
    \quad \quad
    \tau_{\LHS}(\hypergeneric') \coloneqq  \frac{1-\rho}{1-\rho \hypergeneric'^2(1-\noisefixed)}. \nonumber
\end{align}
\end{claim}

\begin{claim} \label{claim:prhs}
For all $0 \leq \hypergeneric' \leq 1$, we have that
\begin{align}
\prhs(\hypergeneric') &=
    \Stab_{\sigma_{\RHS}(\hypergeneric')}[h] - \Stab_{\tau_{\RHS}(\hypergeneric')}[h] = \sum_{\alpha} \wh{h}(\alpha)^2 ((\sigma_{\RHS}(\hypergeneric'))^{|\alpha|} - (\tau_{\RHS}(\hypergeneric'))^{|\alpha|}), \text{~where} \label{eq:prhs-stab}\\
       \sigma_{\RHS}(\hypergeneric') & \coloneqq \frac{(1-\noisefixed)(1-\rho)\hypergeneric'^2 + \hypergeneric'^2\noisefixed}{1-\rho \hypergeneric'^2(1-\noisefixed)},
    \quad \quad
       \tau_{\RHS}(\hypergeneric') \coloneqq \frac{(1-\noisefixed)(1-\rho)\hypergeneric'^2}{1-\rho \hypergeneric'^2(1-\noisefixed)}. \nonumber
\end{align}
\end{claim}

\begin{proofof}{\Cref{claim:plhs}}
For all $0 \leq \hypergeneric' \leq 1$, we have that
\begin{align}
\plhs(\hypergeneric')  &= \frac12 \Ex_{\by, \by', \bz}\bracks*{\parens*{\U_{\hypergeneric'} \zoom{g}{\rho}{\bz}(\sqrt{1-\noisefixed} x + \sqrt{\noisefixed} \by) - \U_{\hypergeneric'}\zoom{g}{\rho}{\bz}(\sqrt{1-\noisefixed}x + \sqrt{\noisefixed} \by')}^2} \label{eq:delta}\\
&=  \frac12 \Ex_{\by, \by', \bz}\Bigl[\Bigl(
    \E_{\bv}\bracks*{\zoom{g}{\rho}{\bz}\parens*{
    \hypergeneric'(\sqrt{1-\noisefixed} x + \sqrt{\noisefixed} \by) + \sqrt{1-\hypergeneric'^2}\bv} } \nonumber \\
    & \qquad \quad {} -
    \E_{\bv'}\bracks*{\zoom{g}{\rho}{\bz}\parens*{
    \hypergeneric'(\sqrt{1-\noisefixed} x + \sqrt{\noisefixed} \by') + \sqrt{1-\hypergeneric'^2}\bv'} }
    \Bigr)^2\Bigr] \label{eq:gno}\\
    &= \frac12 \Ex_{\by, \by', \bz}\Bigl[\Bigl(
    \E_{\bv}\bracks*{g(\sqrt{1-\rho} \bz +
    \sqrt{\rho}\hypergeneric'\sqrt{1-\noisefixed} x + \sqrt{\rho}\hypergeneric' \sqrt{\noisefixed} \by + \sqrt{\rho} \sqrt{1-\hypergeneric'^2}\bv)}  \nonumber \\
    & \qquad \quad {} -
    \E_{\bv'}\bracks*{g(\sqrt{1-\rho} \bz +
    \sqrt{\rho}\hypergeneric'\sqrt{1-\noisefixed} x + \sqrt{\rho}\hypergeneric' \sqrt{\noisefixed} \by' + \sqrt{\rho} \sqrt{1-\hypergeneric'^2}\bv')}
    \Bigr)^2\Bigr], \label{eq:the-zoom}
\end{align}
where \Cref{eq:delta} is by definition of $\nderiv{\by}{\by'}{\hypergeneric',\noisefixed}$, \Cref{eq:gno} is by \Cref{def:gaussian-noise-operator} (the probabilistic definition of $\U_{\hypergeneric'}$, valid when $0 \leq \hypergeneric' \leq 1$), and \Cref{eq:the-zoom} is by definition of the zoom.
Let us define
 \begin{equation} \label{eq:def-of-h}
    h(u) \coloneqq g(\sqrt{\rho}\hypergeneric'\sqrt{1-\noisefixed} x + \sqrt{1-\rho\hypergeneric'^2(1-\noisefixed)} u),
\end{equation}
so expanding the square, we may re-express \Cref{eq:the-zoom} as

{\small
\begin{align}
&\phantom{=}
\frac12 \Ex_{\by, \bz, \bv}\bracks*{
    h\parens*{\frac{\sqrt{1-\rho} \bz +
    \sqrt{\rho \hypergeneric'^2 \noisefixed} \by + \sqrt{\rho} \sqrt{1-\hypergeneric'^2}\bv}{\sqrt{1-\rho\hypergeneric'^2(1-\noisefixed)}}}
    \cdot h\parens*{\frac{\sqrt{1-\rho} \bz + \sqrt{\rho\hypergeneric'^2\noisefixed} \by + \sqrt{\rho} \sqrt{1-\hypergeneric'^2}\bv}{\sqrt{1-\rho\hypergeneric'^2(1-\noisefixed)}}} } \label{eq:first}\\
&-
\Ex_{\by, \by', \bz, \bv,  \bv'}\bracks*{
    h\parens*{\frac{\sqrt{1-\rho} \bz +
    \sqrt{\rho \hypergeneric'^2 \noisefixed} \by + \sqrt{\rho} \sqrt{1-\hypergeneric'^2}\bv}{\sqrt{1-\rho\hypergeneric'^2(1-\noisefixed)}}}
    \cdot h\parens*{\frac{\sqrt{1-\rho} \bz + \sqrt{\rho\hypergeneric'^2\noisefixed} \by' + \sqrt{\rho} \sqrt{1-\hypergeneric'^2}\bv'}{\sqrt{1-\rho\hypergeneric'^2(1-\noisefixed)}}} } \label{eq:second}\\
    &+
\frac12 \Ex_{\by', \bz, \bv'}\bracks*{
    h\parens*{\frac{\sqrt{1-\rho} \bz +
    \sqrt{\rho \hypergeneric'^2 \noisefixed} \by' + \sqrt{\rho} \sqrt{1-\hypergeneric'^2}\bv'}{\sqrt{1-\rho\hypergeneric'^2(1-\noisefixed)}}}
    \cdot h\parens*{\frac{\sqrt{1-\rho} \bz + \sqrt{\rho\hypergeneric'^2\noisefixed} \by' + \sqrt{\rho} \sqrt{1-\hypergeneric'^2}\bv'}{\sqrt{1-\rho\hypergeneric'^2(1-\noisefixed)}}} } \label{eq:third},
\end{align}
}where all the random variables above are distributed as $\normal(0,1)^n$.  It is easy to see that $\eqref{eq:first}=\eqref{eq:third}$, and inspection reveals that both quantities are equal to $\frac 1 2 \Stab_{\sigma_{\LHS}}[h].$ Inspection also reveals that $\eqref{eq:second} = \Stab_{\tau_{\LHS}}[h]$, giving the first equality of \Cref{eq:plhs-stab}.  The second equality of \Cref{eq:plhs-stab} follows from the Hermite formula for $\Stab$ given in \Cref{def:gaussian-noise-stability}, and the proof of \Cref{claim:plhs} is complete.
 \end{proofof}

The proof of \Cref{claim:prhs} is very similar to the above proof so we omit it.

\ignore{
\begin{proofof}{\Cref{claim:prhs}}
The analysis is similar to the preceding proof.  For all $0 \leq \hypergeneric' \leq 1$, we have that
\begin{align}
&\phantom{=} \prhs(\hypergeneric') \nonumber\\
  &= \Ex_{\by, \by', \bz}\bracks*{(\nderiv{\by}{\by'}{\hypergeneric',\noisefixed}
 g)(\sqrt{1-\rho} \bz + \sqrt{\rho} x)^2}  \label{eq:def-of-zoom}\\
 &= \frac12 \Ex_{\by, \by', \bz}\bracks*{\parens*{\U_{\hypergeneric'} g(\sqrt{1-\noisefixed}(\sqrt{1-\rho} \bz + \sqrt{\rho} x) + \sqrt{\noisefixed} \by) - \U_{\hypergeneric'} g(\sqrt{1-\noisefixed}(\sqrt{1-\rho} \bz + \sqrt{\rho} x) + \sqrt{\noisefixed} \by')}^2}
 \label{eq:def-of-delta}\\
     &= \frac12 \Ex_{\by, \by', \bz}\Big[ \Big( \Ex_{\bv} \Big[g(\hypergeneric'\sqrt{1-\noisefixed}\sqrt{1-\rho} \bz + \hypergeneric'\sqrt{1-\noisefixed}\sqrt{\rho} x + \hypergeneric' \sqrt{\noisefixed} \by + \sqrt{1-\hypergeneric'^2}\bv) \Big] \nonumber \\
    & \qquad \qquad  - \Ex_{\bv'} \Big[g(\hypergeneric'\sqrt{1-\noisefixed}\sqrt{1-\rho} \bz + \hypergeneric'\sqrt{1-\noisefixed}\sqrt{\rho} x + \hypergeneric' \sqrt{\noisefixed} \by' + \sqrt{1-\hypergeneric'^2}\bv') \Big]\Big)^2 \Big], \label{eq:def-of-U}
\end{align}
where \Cref{eq:def-of-zoom} is by definition of the zoom, \Cref{eq:def-of-delta} is by definition of $\nderiv{\by}{\by'}{\hypergeneric',\noisefixed}$, and \Cref{eq:def-of-U} is by definition of the Gaussian noise operator. Recalling \Cref{eq:def-of-h} and expanding the square, we may re-express \Cref{eq:def-of-U} as

{\small
\begin{align}
&\phantom{=}
\frac12\Ex_{\by,  \bz, \bv}
    h\parens*{\frac{\hypergeneric'\sqrt{1-\noisefixed}\sqrt{1-\rho} \bz +
    \hypergeneric'\sqrt{\noisefixed} \by +  \sqrt{1-\hypergeneric'^2}\bv}{\sqrt{1-\rho\hypergeneric'^2(1-\noisefixed)}}}
    \cdot  h\parens*{\frac{\hypergeneric'\sqrt{1-\noisefixed}\sqrt{1-\rho} \bz +
    \hypergeneric'\sqrt{\noisefixed} \by +  \sqrt{1-\hypergeneric'^2}\bv}{\sqrt{1-\rho\hypergeneric'^2(1-\noisefixed)}}} \label{eq:first-RHS}\\
&-     \Ex_{\by, \by', \bz, \bv,  \bv'}
    h\parens*{\frac{\hypergeneric'\sqrt{1-\noisefixed}\sqrt{1-\rho} \bz +
    \hypergeneric'\sqrt{\noisefixed} \by +  \sqrt{1-\hypergeneric'^2}\bv}{\sqrt{1-\rho\hypergeneric'^2(1-\noisefixed)}}}
    \cdot  h\parens*{\frac{\hypergeneric'\sqrt{1-\noisefixed}\sqrt{1-\rho} \bz +
    \hypergeneric'\sqrt{\noisefixed} \by' +  \sqrt{1-\hypergeneric'^2}\bv'}{\sqrt{1-\rho\hypergeneric'^2(1-\noisefixed)}}}  \label{eq:second-RHS}\\
&+ \frac12 \Ex_{\by', \bz, \bv'}
    h\parens*{\frac{\hypergeneric'\sqrt{1-\noisefixed}\sqrt{1-\rho} \bz +
    \hypergeneric'\sqrt{\noisefixed} \by' +  \sqrt{1-\hypergeneric'^2}\bv'}{\sqrt{1-\rho\hypergeneric'^2(1-\noisefixed)}}}
    \cdot  h\parens*{\frac{\hypergeneric'\sqrt{1-\noisefixed}\sqrt{1-\rho} \bz +
    \hypergeneric'\sqrt{\noisefixed} \by' +  \sqrt{1-\hypergeneric'^2}\bv'}{\sqrt{1-\rho\hypergeneric'^2(1-\noisefixed)}}}. \label{eq:third-RHS}
\end{align}
}Verifying that $\eqref{eq:first-RHS}=\eqref{eq:third-RHS}= \frac12 \Stab_{\sigma_{\RHS}}[h]$ and that $\eqref{eq:second-RHS}=\Stab_{\tau_{\RHS}}[h]$, the proof of \Cref{claim:prhs} is complete.
\end{proofof}
}

To complete the proof of \Cref{prop:every-outcome}, we must show that $\plhs(\hypergeneric) \leq \prhs(\hypergeneric)$ for all $\hypergeneric \geq 1.$ By \Cref{claim:plhs} and \Cref{claim:prhs}, this would follow immediately from showing that
\[
(\sigma_{\LHS}(\hypergeneric))^{|\alpha|}-(\tau_{\LHS}(\hypergeneric))^{|\alpha|} \leq   (\sigma_{\RHS}(\hypergeneric))^{|\alpha|} - (\tau_{\RHS}(\hypergeneric))^{|\alpha|}.
\]
Letting $a$ plays the role of $|\alpha|$ and clearing the common denominator of $1-\rho\hypergeneric^2(1-\noisefixed)$ that is present in all of $\sigma_{\LHS},\tau_{\LHS},\sigma_{\RHS},\tau_{\RHS}$, it remains to show the following: for all natural numbers $a$ and all real $\hypergeneric \geq 1$,
\begin{equation} \label{eqn:jigsaw}
(\hypergeneric^2\noisefixed\rho-\rho+1)^a-(1-\rho)^a
\leq
(\hypergeneric^2(1-\noisefixed)(1-\rho)+\hypergeneric^2\noisefixed)^a-(\hypergeneric^2(1-\noisefixed)(1-\rho))^a.
\end{equation}
\Cref{eqn:jigsaw} is a consequence of the following stronger inequality (obtained by replacing the quantity $(\hypergeneric^2(1-\noisefixed)(1-\rho))^a$ in \Cref{eqn:jigsaw} by the larger quantity $(\hypergeneric^2(1-\rho))^a$):
\begin{equation} \label{eqn:jigsaw2}
(\hypergeneric^2\noisefixed\rho-\rho+1)^a -(1-\rho)^a
\leq
(\hypergeneric^2(1-\noisefixed)(1-\rho)+\hypergeneric^2\noisefixed)^a - (\hypergeneric^2(1-\rho))^a .
\end{equation}
\Cref{eqn:jigsaw2} can be rewritten as
\begin{equation} \label{eqn:jigsaw3}
(\hypergeneric^2\noisefixed\rho + (1-\rho))^a -(1-\rho)^a
\leq
(\hypergeneric^2\noisefixed \rho + \hypergeneric^2(1-\rho))^a -(\hypergeneric^2(1-\rho))^a,
\end{equation}
which is of the form
\begin{equation} \label{eqn:jigsaw4}
(x + y_0)^a - y_0^a
\leq
(x + y_1)^a - y_1^a
\end{equation}
where $x=\hypergeneric^2 \noisefixed \rho \geq 0$,  $y_1 = \hypergeneric^2(1-\rho)$, and $y_0 = (1-\rho)$; recalling that $\hypergeneric \geq 1$, we have $y_0 \leq y_1$.  Expanding out both sides of \Cref{eqn:jigsaw4} using the binomial theorem, the right-hand side is at least as large as the left-hand side term by term, and the proof of \Cref{prop:every-outcome}, and hence also \Cref{lem:LHS-RHS}, is complete.
\end{proofof}


\section{Noise insensitivity extension lemma} \label{sec:magic-lemma}

For technical reasons our analysis will require a technical result which we state and prove below.  Intuitively, this result says that if $x \in \R^n$ is an input to a degree-$\degree$ polynomial $r(\cdot)$ at which $\poly(\degree)$ many successive ``noisifications'' of $r$, at increasing but all small noise rates, are all multiplicatively close to each other, then they are all multiplicatively close to the value $r(x)$.

Recall that $\maxcol=(2\degree+1)^2$ and that \violet{$\noisefixed =  \parens*{{\frac \epsprg \degree}}^{O(1)}.$}\ignore{\rasnote{Was ``$\noisefixed =  \parens*{{\frac \epsprg \degree}}^{O(\log \degree)}.$''}} The lemma is as follows (recall that the notation
``$a \approx_\gamma b$'' means that $\exp(-\gamma) \leq {\frac a b} \leq \exp(\gamma)$):

\begin{lemma} [Noise insensitivity extension lemma \kanenote{(analogue of Corollary~16 of \cite{kane11focs})}] \label{lem:magicker-lemma}
Let $r_0: \R^n \to \R$ be a non-negative degree-$(2\degree)$ polynomial.  Let $a+b=1$ and suppose $0 \leq b \leq {\frac 1 {C{\degree}^{10}}}$ for a suitable large absolute constant $C$. For $1 \leq j \leq \maxcol$ write $r_j(x)$ to denote $\U_{a^{j/2}} r_0$.

Suppose $x \in \R^n$ is a point such that for all $1 \leq j \leq \maxcol-1$ we have $r_j(x) \approx_\gamma r_{j+1}(x),$ where $\gamma \leq {\frac 1 {12\maxcol(2d+1)}}$.
Then $r_0(x) \approx_1 r_1(x)$, i.e. $e^{-1} \leq {\frac {r_0(x)}{r_1(x)}} \leq e.$
\end{lemma}

We note that later when we apply this lemma it will be with the polynomial $r_0$ instantiated to be a zeroth-column statistic $s_{i,0}$, of degree $2 \degree$, and with $a=1-\noisefixed$, so we will have that \violet{$b=\noisefixed=\parens*{{\frac \epsprg \degree}}^{O(1)}$}\ignore{\rasnote{Was ``$b=\noisefixed=\parens*{{\frac \epsprg \degree}}^{O(\log \degree)}$''}} satisfies $b \leq {\frac 1 {C \degree^{10}}}$ with room to spare. Recalling \Cref{rem:semigroup}, \Cref{lem:magicker-lemma} implies that if the statistics $s_{i,1}(x),\dots,s_{i,\maxcol}(x)$ are all multiplicatively close to each other then $s_{i,0}(x)$ is also multiplicatively (fairly) close to this common value.

It is interesting to contast \Cref{lem:magicker-lemma} with Corollary~16 of \cite{kane11focs}. That corollary gives a qualitatively similar result, also establishing constant-factor multiplicative closeness of $r_0(x)$ as its conclusion, but is quantitatively very different in the assumptions it uses to reach that conclusion.  In Corollary~16 of \cite{kane11focs} only $O(d)$ many noisifications $r_1,\dots,r_{O(d)}$ are considered, but they are assumed to be much closer to each other, multiplicatively $(1 \pm \exp(-d))$-close (and it can be shown that such a strong assumption is required if only $O(d)$ many noisifications are considered). In contrast, \Cref{lem:magicker-lemma} assumes closeness now of $\poly(d)$ rather than $O(d)$ many noisifications, but the closeness that we need to assume is much weaker, only multiplicative $(1 \pm {\frac 1 {\poly(d)}})$-closeness; this is crucial for our overarching goal of ``getting rid of all factors of $2^d$.''  \ignore{We further remark that while the proof of \cite{kane11focs}'s Corollary~16 is based on polynomial interpolation, the proof of \Cref{lem:magicker-lemma} has more of an approximation-theoretic flavor (though interpolation is also used in the argument as will be evident below).}

\begin{proof}[Proof of \Cref{lem:magicker-lemma}]
Recall from \Cref{eqn:U-formula} that for any fixed $x$ and varying $\rho$, the quantity
\[
\U_{\rho} r_0 (x)= \sum_{|\alpha| \leq \degree} \rho^{|\alpha|} \wh{r_0}(\alpha) h_\alpha(x) \coloneqq A(\rho)
\]
is a polynomial in $\rho$ of degree at most $2\degree$.  Let $A_\ast \coloneqq A((1-q)^{1/2})=r_1(x)$. The hypothesis of \Cref{lem:magicker-lemma} tells us that for all $j=1,\dots,\maxcol$, we have
$
A((1-q)^{j/2}) \approx_{\gamma \maxcol} A_\ast;
$ defining the polynomial $B(\rho) \coloneqq {\frac {A(\rho)}{A_\ast}} - 1$, we get that
\[
-\gamma \maxcol \leq e^{-\gamma \maxcol}-1 \leq B((1-q)^{j/2}) \leq e^{\gamma \maxcol} -1 \leq 2 \gamma \maxcol
\]
for all $j=1,\dots,\maxcol.$  Next, let us define the degree-$2\degree$ polynomial $C(\rho)$ by
\[
C((2/q) \rho) \coloneqq B(1-\rho), \quad \quad \text{so} \quad \quad
C((2/q)(1-(1-q)^{j/2})) = B((1-q)^{j/2}) \in [-\gamma \maxcol,2\gamma \maxcol]
\]
for $j=1,\dots,\maxcol$. For notational convenience, for $j =1,\dots,\maxcol$ we write ``$j'$'' to denote the value $(2/q)(1-(1-q)^{j/2})$, and we observe that
\begin{equation} \label{eq:primeisclose}
j\parens*{1 - {\frac 1 {d^5}}} \leq j' \leq j,
\end{equation}
where the upper bound is immediate and the lower bound holds (with room to spare) since by assumption we have $q \leq {\frac 1 {C\degree^{10}}}.$ So intuitively, we have that $C(1'),\dots,C(\maxcol')$ are all very close to zero --- between $-\gamma \maxcol$ and $2 \gamma \maxcol$ --- and to prove the lemma it suffices to show that $C(0) \in [-1/2, 1/2].$

We do this using Lagrange interpolation. Recall that the Lagrange interpolation formula tells us that for any degree-$2\degree$ polynomial $C$ and any $2d+1$ points $x_1,\dots,x_{d+1}$, we have
\begin{equation} \label{eq:lagrange}
C(x) =\sum_{j=1}^{2\degree+1} C(x_j)\ell_j(x),
\quad \quad \text{where} \quad \quad
\ell_j(x) = \prod_{m \in [1,\degree+1] \setminus j} {\frac {x-x_m}{x_j-x_m}}.
\end{equation}

We apply this formula at $x=0$ where we take the $2\degree+1$ values $x_i$ to be $x_i = (i^2)'.$  Fix a $j \in [1,\degree]$ and let us consider $\ell_j(0)$; it is equal to
\begin{align*}
\ell_j(0) &= {\frac {(1^2)' \cdot (2^2)' \cdots ((j-1)^2)'\cdot ((j+1)^2)' \cdots ((2\degree+1)^2)' }{((1^2)' - (j^2)')\cdot
((2^2)' - (j^2)') \cdots (((j-1)^2)' - (j^2)') \cdot
(((j+1)^2)'- (j^2)') \cdots
(((2\degree+1)^2)' - (j^2)')}}
\end{align*}
Note that in the preceding expression, every multiplicand in the numerator is of the form $a'$ for some integer $a \in [1,\maxcol]$ and every multiplicand in the denominator is of the form $(a' - b')$ for distinct integers $a,b \in [1,\maxcol].$  It follows straightforwardly from this and from \Cref{eq:primeisclose} that $\ell_j(0)$ is within a multiplicative $\bracks*{1 - {\frac 1 {d}}, 1 + {\frac 1 {d}}}$ factor of the above expression ``without the primes'', i.e. of
\begin{equation} \label{eq:clean}
{\frac {1^2 \cdot 2^2 \cdots (j-1)^2 \cdot (j+1)^2 \cdots (2\degree+1)^2 }{(1^2 - j^2)\cdot(2^2 - j^2) \cdots ((j-1)^2 - j^2) \cdot
((j+1)^2- j^2) \cdots
((2\degree+1)^2 - j^2)}}.
\end{equation}

Now we require the following bound on the above fraction, which we prove after using it to finish the proof of \Cref{lem:magicker-lemma}:
\begin{claim} \label{claim:fraction}
For all $j \in \{1,\dots,2d+1\}$ it holds that $|\eqref{eq:clean}| \leq
2.$
\end{claim}
It follows that for each $j$ we have $|\ell_j(0)| \leq 2(1 + {\frac 1 \degree})<3$, and hence by \Cref{eq:lagrange} we have that
$|C(0)| \leq 6(\degree+1)\gamma \maxcol < 1/2.$ This proves \Cref{lem:magicker-lemma}.
\end{proof}

\begin{proofof}{\Cref{claim:fraction}}
We have that
\begin{align*}
|\eqref{eq:clean}| &= \abs*{\prod_{i=1}^{2\degree+1} {\frac {i^2}{i^2 - j^2}}}
=
\overbrace{\parens*{\prod_{i=1}^{j-1} {\frac {i^2}{(j-i)(j+i)}}}}^{=A} \cdot
\overbrace{\parens*{\prod_{i=j+1}^{2\degree+1} {\frac {i^2}{(i-j)(i+j)}}}}^{=B},
\quad \text{where by inspection}\\
A &= {\frac {1 \cdot 2 \cdot \cdots \cdot (j-1)}{(j+1) \cdot (j+2) \cdot \cdots \cdot (2j-1)}}, \quad \text{and}\\
B &\leq B' \coloneqq \prod_{i=j+1}^\infty {\frac {i^2}{(i-j)(i+j)}} =
{\frac {(j+1)(j+2) \cdot \cdots \cdot (2j)}{1 \cdot 2 \cdot \cdots \cdot j}}, \quad \text{so}\\
|\eqref{eq:clean}| &= A \cdot B \leq A \cdot B' = {\frac {2j}{j}}=2. \qedhere
\end{align*}
\end{proofof}


\section{Proof of \Cref{lem:page13}:  if some analysis check fails, then with high probability some mollifier check fails} \label{sec:page13}

As per the assumptions of \Cref{lem:page13}, in this section we completely fix an $x = \ul{x} \in \R^n$ which is such that some check in $\AnalysisChecks$ does not hold at $x$, and we let~$\bz$ denote a $\wise$-wise independent $n$-dimensional Gaussian random vector.  We recall the notation from the start of \Cref{sec:hypervariance-bound},
\[
  \bw = \sqrt{1-\noisefixed}{\ul{x}} + \sqrt{\noisefixed}\bz, \qquad  \ul{\st}_{i,j} = \st_{i,j}(\ul{x}), \qquad \bst_{i,j} = \st_{i,j}(\bw) = \zoom{(\st_{i,j})}{\noisefixed}{\ul{x}}(\bz),
\]
and we remark that we will be making extensive use of \Cref{cor:usefulhvbound} in the arguments that follow.

Recalling the statement of \Cref{lem:page13}, we assume through the rest of \Cref{sec:page13} that $\ul{x}$ causes some analysis check to fail, and our goal is to show that
\[
    0 \leq \E[\Indplus(\sqrt{1-\noisefixed} \cdot \ul{x} + \sqrt{\noisefixed} \cdot \bz)] \leq \poly(\taylor \degree^\taylor) \cdot \anticoncgap^{\taylor/2}.
\]
Recalling that $\Indplus$ is the product of functions bounded in $[0,1]$ (namely the indicator of $\sign(p) = 1$ and all of the  $\softcheck_\chk$ functions as $\chk$ ranges over $\MollifierChecks$), to prove \Cref{lem:page13} it suffices to establish the following: 
\begin{multline}
    \text{There exists some } \chk \in \MollifierChecks \text{ ``$s_u \geq c s_v$ with softness~$\softness$''}, \\
     \text{such that } \softcheck_\chk(\sqrt{1-\noisefixed} \cdot \ul{x} + \sqrt{\noisefixed} \cdot \bz) = 0 \text{ --- equivalently, } \bst_{u} < \exp(-\softness) \cdot c \bst_v \text{ ---} \\
         \text{except with probability at most } \poly((\taylor \degree)^\taylor) \cdot( \anticoncgap)^{\taylor/2} \text{~over~}\bz. \label[ineq]{ineq:muffin}
\end{multline}

To establish \Cref{ineq:muffin}, we will consider the checks in $\AnalysisChecks$ in a careful order, specifically, the order shown below.  All subsequent references to the ``first'' analysis check that fails, ``earlier'' or ``later'' analysis checks, etc.~are with respect to this ordering.

\myfig{.75}{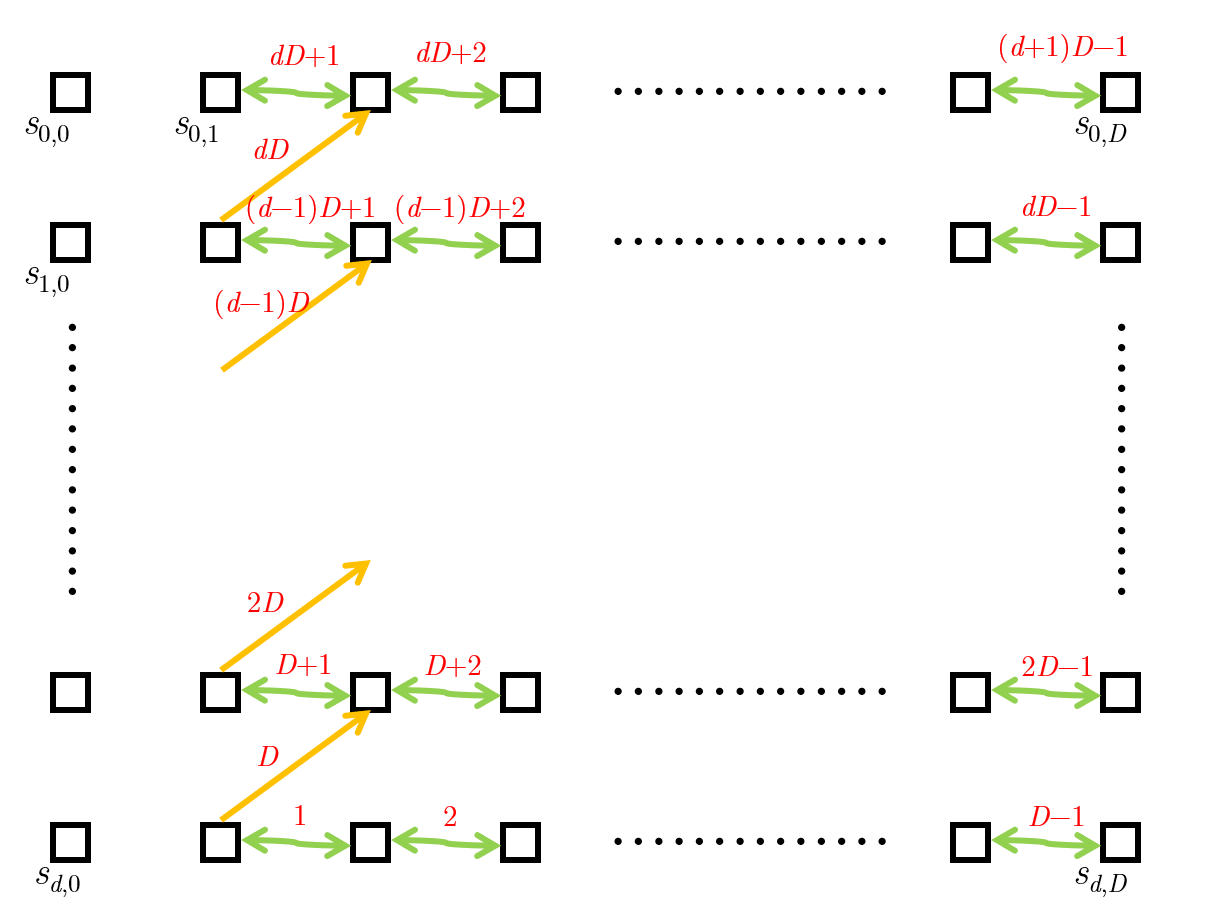}{Ordering of analysis checks (indicated in red)}{}


\medskip

The argument has three cases depending on where is the first analysis check that fails for~$\ul{x}$ (a horizontal check in the bottom row; a horizontal check in a higher row; or a diagonal check).
Before entering into the case analysis, which involves detailed and careful arguments, we stress two high level points.  First, the overall \emph{qualitative} structure of the following arguments follows \cite{kane11focs} quite closely (in particular page 13 of that paper).  Second, to obtain our \emph{quantitative} improvement over \cite{kane11focs} (essentially, getting a $\poly(\degree^\taylor)=\poly(\degree)$ factor in the failure probability of \Cref{ineq:muffin} rather than the $2^{(O(\degree)}$ factor that is present in \cite{kane11focs}), crucially requires the technical tools that we developed in \Cref{sec:hypervariance-bound} and \Cref{sec:magic-lemma}.

\subsection{The first failing analysis check is horizontal and is in the bottom row ($i=d$).}

Recall that the horizontal analysis checks in the bottom row are $\ul{\st}_{\degree,j} \approx_{\horizanal} \ul{\st}_{\degree,j+1}$ for all $1 \leq j \leq \maxcol-1$. But \Cref{def:si0} and \Cref{fact:deriv-deg} imply that $\st_{\degree,0}$ is a constant function, and it follows from \Cref{def:sij} that $\st_{\degree,j}$ is the same constant function for all~$j$.  Thus all numbers $\ul{\st}_{\degree,j}$ are equal to the same constant, and hence the analysis checks in the bottom row cannot actually fail.  So this case cannot occur.

\subsection{The first failing analysis check is horizontal and in some row $0\le i<\degree$}
Suppose that the first analysis check to fail is one of the two implicit in the statement
\begin{equation}    \label{eqn:horzfail}
    \ul{\st}_{i^*,j^*} \not \approx_{\horizanal} \ul{\st}_{i^*,j^*+1},
\end{equation}
for some $0 \leq i^* \leq \degree-1$ and $1 \leq j^* \leq \maxcol - 1$.
We first note that (similar to the beginning of the proof of \Cref{cor:usefulhvbound}) if one of the two quantities $\ul{\st}_{i^*,j^*}, \ul{\st}_{i^*,j^*+1}$ is zero then $\st_{i,0}$ must be the constant-0 polynomial and hence the other quantity must be zero as well. But since \Cref{eqn:horzfail} holds, it cannot be the case that  $\ul{\st}_{i^*,j^*}$ and $\ul{\st}_{i^*,j^*}$ are both zero.  Hence in the rest of the proof we assume that $\ul{\st}_{i^*,j^*}, \ul{\st}_{i^*,j^*+1}>0.$ 

In this case we will analyze the random variables $\bst_{i^*, j^*-1}$ and $\bst_{i^*, j^*}$. By \Cref{fact:themean} we have that
\begin{equation} \label{eqn:damean}
    \E[\bst_{i^*,j^*-1}] = \ul{\st}_{i^*,j^*}, \qquad \E[\bst_{i^*,j^*}] = \ul{\st}_{i^*,j^*+1}, \qquad \text{and let us write } M \coloneqq \max\{ \ul{\st}_{i^*,j^*}, \ul{\st}_{i^*,j^*+1} \}.
\end{equation}

Since \Cref{eqn:horzfail} is the first analysis check to fail, it must be the case that all horizontal analysis checks in the $(i^*+1)$-th row passed, i.e.~
\[
\text{for all~}1 \leq j \leq \maxcol-1, \text{~we have~}    \ul{\st}_{i^*+1,j} \approx_{\horizanal} \ul{\st}_{i^*+1,j+1},
\]
which immediately gives that
\[
\text{for all~}1 \leq j \leq \maxcol, \text{~we have~}   \ul{\st}_{i^*+1,j} \approx_{\maxcol \horizanal} \ul{\st}_{i^*+1,1}.
\]
Now we apply the noise insensitivity extension lemma, \Cref{lem:magicker-lemma}, to the degree-$(2\degree)$ polynomial $\st_{i^*,0}$ in place of ``$r_0$'', with $\horizanal$ in place of ``$\gamma$,'' and with $\noisefixed$ in place of ``$b$.'' Since $\horizanal$ is indeed less than ${\frac 1 {12\maxcol(2\degree+1)}}$ as required by that lemma, \Cref{lem:magicker-lemma} gives us that we can extend the above closeness (with slightly weaker parameters) even to the $j = 0$ case:
\begin{equation} \label{ineq:even-zero-is-close}
\text{for all~}0 \leq j \leq \maxcol, \text{~we have~}     \ul{\st}_{i^*+1,j} \approx_{1} \ul{\st}_{i^*+1,1} .
\end{equation}

As another consequence of the fact that  \Cref{eqn:horzfail} is the first analysis check to fail, we have that the diagonal analysis check relating the $(i^*+1)$-th to the $i^*$-th row passed, i.e.~we have that
\begin{equation} \label{iplusoneiscoolvisavisi}
    \ul{\st}_{i^*+1,1} \leq 100 \anticoncgap \ul{\st}_{i^*,2}.
\end{equation}
Combining \Cref{ineq:even-zero-is-close,iplusoneiscoolvisavisi}, we conclude that
\begin{equation} \label{eq:bounded-by-two}
\text{for all~}0 \leq j \leq \maxcol, \text{~we have~}    \ul{\st}_{i^*+1,j} \leq 100 e \anticoncgap \ul{\st}_{i^*,2}.
\end{equation}

Now, recalling \Cref{cor:usefulhvbound}, we have that
\begin{align*}
\bst_{i^*,j^*-1} \text{ is } (\taylor, 4\max\{\sqrt{\zeta_{i^*,j^*-1}}, \zeta_{i^*,j^*-1}\})\text{-hyperconcentrated, where~} \zeta_{i^*,j^*-1} = \frac{\ul{\st}_{i^*+1,j^*-1}}{\ul{\st}_{i^*,j^*}}
    & \leq \frac{100e \anticoncgap \ul{\st}_{i^*,2}}{\ul{\st}_{i^*,j^*}},\\
    \bst_{i^*,j^*} \text{ is } (\taylor, 4\max\{\sqrt{\zeta_{i^*,j^*}}, \zeta_{i^*,j^*}\})\text{-hyperconcentrated, where~} \zeta_{i^*,j^*} = \frac{\ul{\st}_{i^*+1,j^*}}{\ul{\st}_{i^*,j^*+1}}  &\leq \frac{100e \anticoncgap \ul{\st}_{i^*,2}}{\ul{\st}_{i^*,j^*+1}},
\end{align*}
where both inequalities are by \Cref{eq:bounded-by-two}. (Note that the above ratios are well-defined since $\ul{\st}_{i^*,j^*}, \ul{\st}_{i^*,j^*+1}>0.$)

To analyze the  $\ul{\st}_{i^*,2}$ factor which appears in both numerators above,   we consider two cases.  If $j^* = 1$, then $\ul{\st}_{i^*, 2} = \ul{\st}_{i^*, j^*+1} \leq M$ (recalling \Cref{eqn:damean}).  Otherwise, by virtue of the fact that all preceding analysis checks in the $i^*$-th row passed, we conclude that $\ul{\st}_{i^*,2} \approx_{\maxcol \horizanal} \ul{\st}_{i^*,j^*}$ and hence $\ul{\st}_{i^*,2} \leq eM$ (since $\maxcol \horizanal \leq e$).  Either way, we conclude that
\begin{equation} \label[ineq]{ineq:pumpkin}
    \zeta_{i^*,j^*-1}  \leq \frac{100e^2 \anticoncgap M}{\ul{\st}_{i^*,j^*}}, \quad \zeta_{i^*,j^*}  \leq \frac{100e^2 \anticoncgap M}{\ul{\st}_{i^*,j^*+1}}.
\end{equation}

For the rest of the analysis of this case, we will reason the exact same way about $\bst_{i^*,j^*-1}$ and about $\bst_{i^*,j^*}$.  Let us write $\bst$ to denote either $\bst_{i^*,j^*-1}$ or $\bst_{i^*,j^*}$, and we similarly write just $\zeta$ for either $\zeta_{i^*,j^*-1}$ or $\zeta_{i^*,j^*},$ and write $\mu = \E[\bst]$ (note that recalling \Cref{eqn:damean}, we have that $\mu$ is either $\ul{\st}_{i^*,j^*}$ or $\ul{\st}_{i^*,j^*+1}$).   \Cref{ineq:pumpkin} tells us that in either case we have
\begin{equation} \label{eq:zeta-ub}
\zeta \leq 100 e^2 \anticoncgap (M / \mu).
\end{equation}  Now we apply \Cref{prop:hypermarkov} to the random variable $\bst$ (which was shown above to be $(\taylor,4 \max\{\sqrt{\zeta},\zeta\})$-hyperconcentrated), taking its ``$t$'' parameter to be $\frac{M/\mu}{\degree^{C}}$ for a large absolute constant $C$, to deduce that
\begin{equation} \label{eq:both}
    |\bst - \mu| \leq M/\degree^{C}  \text{ except with probability at most } \parens*{\frac{4\max\{\sqrt{\zeta},\zeta\}\cdot \degree^{C}}{M/\mu}}^\taylor \leq \poly(\degree^\taylor) \cdot \anticoncgap^{\taylor/2}.
\end{equation}
(To justify the last inequality, we observe that if $\zeta \geq 1$ then $\max\{\sqrt{\zeta},\zeta\}/(M/\mu)  =
\zeta/(M/\mu) \leq 100e^2 \anticoncgap$ by \Cref{eq:zeta-ub}, which is at most $O(\sqrt{\anticoncgap})$ since $\anticoncgap \leq 1$.  On the other hand, if $\zeta < 1$ then $\max\{\sqrt{\zeta},\zeta\}/(M/\mu)  = \sqrt{\zeta}/(M/\mu)$, which is at most $10 e \sqrt{\anticoncgap}/\sqrt{M/\mu}$ by \Cref{eq:zeta-ub}, which is in turn at most $10e\sqrt{\anticoncgap}$ using $M= \max\{ \ul{\st}_{i^*,j^*}, \ul{\st}_{i^*,j^*+1} \} \geq \mu$.)

Thus except with a $\poly(\degree^\taylor) \cdot \anticoncgap^{\taylor/2}
$ failure probability, each of $\bst_{i^*,j^*-1}, \bst_{i^*,j^*}$ is within an additive $\pm M/\degree^{C}$ of its mean. Without loss of generality (the other case is entirely similar), let us assume that $\bst_{i^*,j^*}$ has the larger mean, so $\mu = M = \Ex[\bst_{i^*,j^*}] = \st_{i^*,j^*+1}$; given this, \Cref{eqn:horzfail} tells us that
\[
    \ul{\st}_{i^*,j^*} \leq \exp(-\horizanal) \ul{\st}_{i^*,j^*+1} = \exp(-\horizanal) M.
\]
Now, even if $\ul{\st}_{i^*,j^*} = \E[\bst_{i^*,j^*-1}]$ were as large as possible (by the above, this largest possible value is $\exp(-\horizanal) M$), except with overall failure probability at most $\poly(\degree^\taylor) \cdot \anticoncgap^{\taylor/2}$, we have that both
\[
\bst_{i^*,j^*-1} \leq \st_{i^*,j^*} + {\frac M {\degree^C}} \leq \exp(-\horizanal) M + {\frac M {\degree^C}}
\]
(by \Cref{eq:both} applied with its $\bst$ being $\bst_{i^*,j^*-1}$ and its $\mu$ being $\st_{i^*,j^*}$) and
\[
\bst_{i^*,j^*} \geq M - {\frac M {\degree^C}}
\]
(by \Cref{eq:both} applied with its $\bst$ being $\bst_{i^*,j^*}$ and its $\mu$ being $M$). Recalling the definitions of $\horizclose$ and $\horzanal$ from \Cref{eq:anticoncgap-horizclose} and \Cref{eq:horzanal}, and that the $C$ above is a large constant, the two preceding inequalities imply that \Cref{ineq:muffin} holds for the $(i^*,j^*-1)$~vs.~$(i^*,j^*)$ noise insensitivity mollifier check (see \Cref{sec:checking-insensitivity}), as desired.


\subsection{The first failing analysis check is diagonal}
Finally, the last case we must consider is that the first analysis check to fail is the diagonal check
\[
    \ul{\st}_{i^*+1,1} \not \leq 100 \anticoncgap \ul{\st}_{i^*,2}
\]
for some $0 \leq i^* \leq d-1$.  In this case we analyze the random variables $\bst_{i^*+1, 0}$ and $\bst_{i^*, 1}$.

We first observe that by \Cref{eqn:themean} and the above inequality we can lower bound the expectation of
$\bst_{i^*+1,0}$ by
\begin{equation}    \label[ineq]{ineq:firstdiag}
    \E[\bst_{i^*+1,0}] = \ul{\st}_{i^*+1,1} > 100 \anticoncgap \ul{\st}_{i^*,2} = 100 \anticoncgap \E[\bst_{i^*,1}],
\end{equation}
which will be useful for us later.  Next we give a high-probability lower bound on $\bst_{i^*+1,0}$:

\begin{claim}   \label{claim:imhungry}
    $\bst_{i^*+1,0} \geq \frac12 \ul{\st}_{i^*+1,1}$ except with probability at most $(80e)^{\taylor} \cdot \anticoncgap^{\taylor/2}$.
\end{claim}
\begin{proof}
    The claim is immediate from \Cref{prop:hypermarkov}, taking its ``$t$'' to be $1/2$, once we establish the following:
    \begin{equation}    \label[ineq]{ineq:pizza}
        \bst_{i^*+1, 0} \text{ is } (\taylor, 40e\sqrt{\anticoncgap})\text{-hyperconcentrated.}
    \end{equation}
    To establish \Cref{ineq:pizza}, first suppose that $i^*+1 = d$. In this case $\st_{i^*+1,0}$ is a constant function, so $\bst_{i^*+1,0}$ is a constant random variable, and \Cref{ineq:pizza} is clearly true.  The other possibility is that $i^*+1 < d$.  In this case since $i^* < d-1$, there must have been at least one earlier  diagonal analysis check, and it succeeded, meaning that
    \[
        \ul{\st}_{i^*+2,1} \leq 100 \anticoncgap \ul{\st}_{i^*+1,2}.
    \]
We further have that the horizontal analysis checks in rows $i^*+1$ and $i^*+2$ all succeeded, and hence
    \[
        \ul{\st}_{i^*+1,2} \approx_{\horizanal} \ul{\st}_{i^*+1,1}, \qquad \ul{\st}_{i^*+2,1} \approx_{1} \ul{\st}_{i^*+2,0},
    \]
where the second inequality is by an application of \Cref{lem:magicker-lemma} (which we may apply because $\horizanal \leq {\frac 1 {12 \maxcol (2 \degree + 1)}}$) to the degree-$2\degree$ polynomial $\st_{i^\ast + 2,0}.$
Upper-bounding $e^{\horizanal}$ by $e$ for simplicity and combining these relations, we deduce that
    \[
        \ul{\st}_{i^*+2,0} \leq 100 e^2 \anticoncgap \ul{\st}_{i^*+1,1}.
    \]
    Now \Cref{ineq:pizza} follows from \Cref{cor:usefulhvbound} applied to $\bst_{i,j} = \bst_{i^*+1,0}$.\end{proof}

Next we establish a high-probability upper bound on $\bst_{i^*,1}$:

\begin{claim} \label{claim:imthirsty}
    $\bst_{i^*,1} \leq  \ul{\st}_{i^*+1,1}/(20\anticoncgap)$ except with probability at most $16^\taylor \cdot \anticoncgap^{\taylor/2}$.
\end{claim}
\begin{proof}
    Applying \Cref{ineq:useme} to $\bst_{i^*,1}$, we get that $\bst_{i^*,1}$ is $(\taylor, 4\max\{\sqrt{\zeta_{i^*,1}},\zeta_{i^*,1}\})$-hyperconcentrated, and  \Cref{ineq:firstdiag} tells us that $\zeta_{i^*,1}/ \anticoncgap > 100$.
    Applying \Cref{prop:hypermarkov} to $\bst_{i^*,1}$ with its ``$t$'' parameter set to $\zeta_{i^*,1}/(20\anticoncgap) - 1 > \zeta_{i^*,1}/(40\anticoncgap)$, we get that except with failure probability at most

    \begin{equation} \label{eq:fpbound}
    \parens*{{\frac \eta t}}^\taylor \le
        \parens*{\frac{4\max\{\sqrt{\zeta_{i^*,1}},\zeta_{i^*,1}\}}{\zeta_{i^*,1}/(40\anticoncgap)}}^\taylor = (160\anticoncgap \max\{\zeta_{i^*,1}^{-1/2},1\})^\taylor,
    \end{equation}
    we have
    \[
        \bst_{i^*,1} \leq (t+1)  \E[\bst_{i^*,1}] = (\zeta_{i^*,1}/(20\anticoncgap)) \cdot \E[\bst_{i^*,1}]  =
         (\zeta_{i^*,1}/(20\anticoncgap)) \cdot \ul{\st}_{i^*,2}
        = \ul{\st}_{i^*+1,1}/(20\anticoncgap),
    \]
where the second equality is by \Cref{fact:themean} and the third is by the definition of $\zeta_{i^*,1}$ (recalling \Cref{ineq:useme}).    Using again $\zeta_{i^*,1} > 100 \anticoncgap$, the failure probability bound \Cref{eq:fpbound} is at most $16^\taylor \cdot \anticoncgap^{\taylor/2}$.
\end{proof}

Putting \Cref{claim:imhungry} and \Cref{claim:imthirsty} together, we conclude that
\[
    \bst_{i^*+1,0} \geq 10 \anticoncgap \bst_{i^*,1} \quad \text{except with probability } O(\anticoncgap)^{\taylor/2},
\]
which establishes \Cref{ineq:muffin} for the $i^*$-th local hyperconcentration mollifier check (see \Cref{sec:checking-local-hyperconcentration}), as desired.\ignore{\rasnote{This had said ``(using $10 \geq e^2$)'' but I'm not sure how $e^2$ comes into it?  It's fine though, this is minor.}
This concludes the analysis of the last case, so \Cref{ineq:muffin} is established and with it the proof of \Cref{lem:page13} is complete.
}


\section{Proof of \Cref{lem:page14}:  Using a Taylor-based argument if all analysis checks pass} \label{sec:page14}

As per the assumptions of \Cref{lem:page14}, in this section we completely fix an $x = \ul{x} \in \R^n$ such that all of the analysis checks pass, and we let~$\bz$ denote a $\wise$-wise independent $n$-dimensional Gaussian random vector.  We recall the notation from the start of \Cref{sec:hypervariance-bound},
\[
  \bw = \sqrt{1-\noisefixed}{\ul{x}} + \sqrt{\noisefixed}\bz, \qquad  \ul{\st}_{i,j} = \st_{i,j}(\ul{x}), \qquad \bst_{i,j} = \st_{i,j}(\bw) = \zoom{(\st_{i,j})}{\noisefixed}{\ul{x}}(\bz),
\]
and we note that we will again be making use of \Cref{cor:usefulhvbound} in the arguments that follow. We further introduce the notation
\[
\mu_p \coloneqq \E[p(\bw)].
\]

To prove \Cref{lem:page14} we must show that the expectation of the random variable
\begin{equation} \label{eq:indplusbw}
    \Indplus(\bw) = \allsoftchecks\parens*{\bw} \cdot \bone[\sign(\p(\bw))],
\end{equation}
is determined up to an additive $\pm \poly((\taylor \degree)^\taylor) \cdot \anticoncgap^{\taylor/2}$ just by virtue of $\bz$ being $\wise$-wise independent.
To do this it is useful to observe that $\allsoftchecks(\bw)$ depends only on the $\bst_{i,j}$ random variables; more precisely, we may rewrite \Cref{eq:indplusbw} as
\begin{equation} \label{eq:indplusbw-rewritten}
\Indplus(\bw) = \wt{\allsoftchecks}\parens*{\{\bst_{i,j}\}_{0 \leq i \leq \degree, 0 \leq j \leq \maxcol-1}} \cdot \bone[\sign(\p(\bw))],
\end{equation}
where the function $\wt{\allsoftchecks}: (\R^{\geq 0})^{(\degree+1) \times \maxcol} \to [0,1]$ is defined in the obvious way,\ignore{\rasnote{Is it okay to have the domain being with $\R^{\geq 0}$ rather than $\R^+$?  I guess if the ratio $\st_u/\st_v$ is $0/0$ we defined $0/0=+\infty$, and if $\st_u/\st_v =$ (something positive but tiny)/0 then again the ratio is $+\infty$ and we're okay. So I think it's well defined and smooth. Do you agree?}}
\begin{equation} \label{eq:wt-allsoftchecks}
\wt{\allsoftchecks}\parens*{\{\st_{i,j}\}_{0 \leq i \leq \degree, 0 \leq j \leq \maxcol-1}}
\coloneqq
\prod_{(\ineq = (c, \st_u, \st_v), \softness) \in \MollifierChecks} \sigma\parens*{\softness^{-1} \ln\parens*{\frac{\st_u}{c \st_v}}}.
\end{equation}

We begin to analyze \Cref{eq:indplusbw-rewritten} by first analyzing the simpler random variable $\bone[\sign(\p(\bw))]$ which is a part of $\Indplus(\bw)$:

\begin{claim} \label{claim:sign-p-almost-constant}
$\Pr[\sign(p(\bw))] \neq \sign(\mu_p)] \leq O(\anticoncgap)^{\taylor/2}$.
\end{claim}
\begin{proof}
Recalling that $s_{0,0}=p^2$, by \Cref{def:sij} we have that
\[
    \ul{\st}_{0,1} = \Ex_{\by \sim \normal(0,1)^n}[\zoom{p}{\noisefixed}{\ul{x}}(\by)^2],
\]
and by \Cref{fact:deriv-hypervar} and \Cref{def:si0}\ignore{ and the $\wise$-wise independence of $\bz$,} we have that
\[
    \ul{\st}_{1,0} = \HyperVar_{\hyperfixed}[\zoom{p}{\noisefixed}{\ul{x}}\ignore{(\bw)}].
\]
Furthermore, since by the assumption of \Cref{lem:page14} we have that all the analysis checks pass at~$\ul{x}$, we may draw the following conclusions:
\[
   \ul{\st}_{1,1} \leq 100 \anticoncgap \ul{\st}_{0,2}, \quad \ul{\st}_{0,1} \approx_{\horzanal} \ul{\st}_{0,2}, \quad \ul{\st}_{1,0} \approx_{1} \ul{\st}_{1,1},
\]
where the first of these is by the $i=1$ diagonal analysis check passing, the second is by the $(i,j)=(0,1)$ horizontal analysis check passing, and the last of these follows from the noise insensitivity extension lemma \Cref{lem:magicker-lemma} and the passing of the horizontal analysis checks
$
    \st_{1,j}(\ul{x}) \approx_{\horzanal} \st_{1,j+1}(\ul{x})
$ for all $1 \leq j \leq \maxcol - 1.$
Combining these bounds we deduce that
\[
    \HyperVar_{\hyperfixed}[\zoom{p}{\noisefixed}{\ul{x}}\ignore{(\bw)}] = \ul{\st}_{1,0} \leq 100 e \anticoncgap \ul{\st}_{0,2} = 100 e^2 \anticoncgap  \Ex_{\by \sim \normal(0,1)^n}[\zoom{p}{\noisefixed}{\ul{x}}(\by)^2],
\]
or in other words $\zoom{p}{\noisefixed}{\ul{x}}$ is $(\hyperfixed, 100 e^2 \anticoncgap)$-attenuated. Since $\taylor \leq 1 + {\frac 1 2} \hyperfixed^2$, by \Cref{lem:attenuated-hyperconcentrated} we may conclude that for $\by \sim \normal(0,1)^n$, the random variable $\zoom{p}{\noisefixed}{\ul{x}}(\by)$ is $(\taylor, 10e \sqrt{\anticoncgap})$-hyperconcentrated.  Since the definition of $(\taylor, 10e \sqrt{\anticoncgap})$-hyperconcentration only uses $\taylor$-th moments, and $\zoom{p}{\noisefixed}{\ul{x}}$ has degree at most~$\degree$ and $\taylor \degree \leq \wise$, recalling that $\bz$ is $\wise$-wise independent we may conclude that the random variable $p(\bw)$ is also $(\taylor, 10e \sqrt{\anticoncgap})$-hyperconcentrated. (Recall \Cref{rem:why-wise}.)
Now applying \Cref{prop:hypermarkov} with its ``$t$'' parameter set to $1/2$, we get that \[
    \abs{p(\bw) - \mu_p} \leq (1/2) |\mu_p|, \quad \text{except with probability at most } O(\anticoncgap)^{\taylor/2},
\]
which is easily seen to imply the claim if $\mu_p \neq 0$.  (If $\mu_p=0$, then the above condition translates into $p(\bw)=0$ with probability at least (say) $0.9$, which can only be the case if $p$ is identically 0, in which case the claim holds trivially.)
\end{proof}

Since $\wt{\allsoftchecks}$ is bounded in $[0,1]$, as an immediate consequence of \Cref{claim:sign-p-almost-constant} we have that
$\Ex[\Indplus(\bw)]$ and $\Ex[\wt{\allsoftchecks}\parens*{\{\bst_{i,j}\}_{0 \leq i \leq \degree, 0 \leq j \leq \maxcol-1}}]$ can differ by at most an additive $O(\anticoncgap)^{\taylor/2}$.
Hence the remaining task, to prove \Cref{lem:page14}, is to show that
\begin{equation} \label{eq:end-of-suffering}
\wt{\allsoftchecks}\parens*{\{\bst_{i,j}\}_{0 \leq i \leq \degree, 0 \leq j \leq \maxcol-1}} \text{~is determined up to an additive~}\pm \poly((\taylor d)^\taylor) \cdot \anticoncgap^{\taylor/2}.
\end{equation}
In the rest of this section we do this as follows: first, in \Cref{sec:establishing-hyperconcentration}, we use  the assumption that all analysis checks pass at~$\ul{x}$ to show that the statistics $\bst_{i,j}$ are suitably hyperconcentrated. Next, in \Cref{sec:mollifier-fn-of-stats} we show that the $\wt{\allsoftchecks}$ function satisfies a certain technical ``relaxedness'' condition (essentially a bound on the magnitude of its derivatives). Finally, in \Cref{sec:taylor} we use hyperconcentration of the $\bst_{i,j}$ statistics and relaxedness of $\allsoftchecks$ in an argument based on Taylor's theorem to establish \Cref{eq:end-of-suffering}.

\subsection{Establishing hyperconcentration of the $\bst_{i,j}$'s} \label{sec:establishing-hyperconcentration}

The proof of the following lemma uses the assumption that all analysis checks pass at~$\ul{x}$:

\begin{lemma} \label{lem:we-got-hyperconcentration}
For $0 \leq i \leq \degree$, $0 \leq j \leq \maxcol-1$, the random variable $\bst_{i,j}$ is $(2\taylor, 40e \sqrt{\anticoncgap})$-hyperconcentrated.
\end{lemma}
\begin{proof}
Fix any $0 \leq i \leq \degree-1$.  Since the diagonal analysis check between the $i$-th and $(i+1)$-th row passes at~$\ul{x}$, we have (recalling \Cref{eqn:diag-checks}) that $\ul{\st}_{i+1,1} \leq 100 \anticoncgap \ul{\st}_{i,2}.$
Similar to the beginning of the proof of \Cref{cor:usefulhvbound}, if $\ul{\st}_{i,2}=0$ then $\bst_{i,j}$ must be the constant-0 random variable and the lemma holds. So we assume that $\ul{\st}_{i,2} > 0$ and we have that
\begin{equation} \label{eq:ulst}
    \frac{\ul{\st}_{i+1,1}}{\ul{\st}_{i,2}} \leq 100 \anticoncgap.
\end{equation}
Further, since the horizontal analysis checks in rows $i$ and $i+1$ all passed at~$\ul{x}$, we can apply \Cref{lem:magicker-lemma} to $\st_{i,0}$ and $\st_{i+1,0}$ to get that
\begin{equation} \label{eq:ulst2}
\text{for all~}0 \leq j \leq \maxcol, \quad \quad    \ul{\st}_{i,j} \approx_{1} \ul{\st}_{i,2}, \qquad \ul{\st}_{i+1,j} \approx_{1} \ul{\st}_{i+1,1}.
\end{equation}
Combining \Cref{eq:ulst} and \Cref{eq:ulst2}, and recalling the notation $\zeta_{i,j} = \frac{\ul{\st}_{i+1,j}}{\ul{\st}_{i,j+1}}$ from \Cref{ineq:useme}, we may conclude that
\[
  \text{for all~}0 \leq i \leq \degree - 1, 0 \leq j \leq \maxcol-1 \text{~we have that~} \zeta_{i,j} \leq 100 e^2 \anticoncgap.
\]
Putting this into \Cref{ineq:useme}, we get that
\[
    \text{for $0 \leq i \leq \degree-1$, $0 \leq j \leq \maxcol-1$,} \quad \bst_{i,j} \text{ is } ({2}\taylor, 40e \sqrt{\anticoncgap})\text{-hyperconcentrated}.
\]
This hyperconcentration trivially extends to $\bst_{\degree,j}$ because $\bst_{d,j}$ is a constant random variable, and \Cref{lem:we-got-hyperconcentration} is proved.
\end{proof}

\subsection{$\wt{\allsoftchecks}$ is relaxed} \label{sec:mollifier-fn-of-stats}

For notational simplicity in the remainder of this section, rather than using the $(i,j)$-indexing for the elements of $\bst = \parens*{\{\bst_{i,j}\}_{0 \leq i \leq \degree, 0 \leq j \leq \maxcol-1}}$, we will use generic indices $1 \leq u \leq m \coloneqq (\degree+1)\maxcol = \poly(\degree)$, so we write $\bst$ as $\bst = (\bst_1,\dots,\bst_m)$.

\begin{definition}
For $a \in \N$, $B \geq 1$, we say a function $\psi : (\R^{\geq 0})^m \to [0,1]$ is \emph{$(a,B)$-relaxed} if it is smooth and satisfies\ignore{\rodnote{We can either say that we take $1/0 = \infty$, an extended real, or we could annoyingly rewrite as $|s|^{\alpha} \cdot \abs{\partial_\alpha \psi(s)} \leq B$.}}
    \[
\text{for all $\alpha \in \N^n$ with~} 0 \leq |\alpha| \leq a \text{~and all~}s \in (\R^{\geq 0})^m, \quad \quad         |s|^{\alpha} \cdot \abs{\partial_\alpha \psi(s)} \leq B.
    \]
\end{definition}
\begin{lemma}
[$\wt{\allsoftchecks}$ is relaxed]
\label{lem:mollifier-is-relaxed}
  The function $\psi = \wt{\allsoftchecks}$ is $(2\taylor,\poly((\taylor\degree)^\taylor))$-relaxed.
\end{lemma}
\begin{proof}
Let us write $B_0$ to denote the maximum of $1/\delta$ where $\delta$ ranges over all of the ``softness parameters'' involved in the definition of $\psi=\wt{\allsoftchecks}$ (recall \Cref{eq:wt-allsoftchecks}). Recalling \Cref{sec:checking-local-hyperconcentration} and \Cref{sec:checking-insensitivity} we have that $B_0 = 1/\horizclose = \poly(d).$

    Say that a function $\phi(s)$ is ``$\psi$-like'' if it takes the form of the right-hand side of \Cref{eq:wt-allsoftchecks}, except that some of the multiplicands may have derivatives $\sigma'$, $\sigma''$, $\sigma'''$, etc., in place of~$\sigma$.  We show by induction on $a = |\alpha|$ that $\partial_\alpha \psi(s)$ consists of a sum of at most $(m+1)^a$ terms, each being of the form $b \frac{\phi(s)}{s^\alpha}$, where $\phi(s)$ is $\psi$-like and~$b$ is a constant that is at most $(aB_0)^a$ in magnitude.
The base case $a = 0$ is immediate.  For the induction step, we consider differentiating a term $b \frac{\phi(s)}{s^\alpha}$ with respect to some $s_u$.  We view this term as a product of up to~$m+1$ factors involving~$s_u$, namely the multiplicands in~$\phi(s)$ involving~$s_u$, and also any power $s_u^{j}$ in the denominator.  Now we use the calculus product rule. Differentiating a multiplicand of the form $\sigma^{(i)}(\pm \softness^{-1} \ln(s_u/s_v) + \text{const.})$ with respect to~$s_u$ gives a similar factor, but with a higher derivative $\sigma^{(i+1)}$ and picking up a factor of $\pm \softness^{-1}/s_u$.  Thus we indeed get another $\psi$-like term, with an extra factor of~$s_u$ in the denominator and a constant factor increased in magnitude by at most~$B_0 \leq aB_0$, as is sufficient for the induction.  Similarly, differentiating the factor of $1/s_u^j$ picks up a constant factor of~$j \leq a \leq aB_0$ in magnitude, as well as an extra factor of $s_u$ in the denominator.  This completes the induction.

Next, we observe that in all the $\psi$-like terms that are present in $\partial_\alpha \psi(s)$, the maximum-order derivative on~$\sigma$ that arises is at most $a = |\alpha|$.  Recalling \Cref{def:sigma}, for $a>0$ all of these $\psi$-like terms are uniformly bounded in magnitude on $(\R^{\geq 0})^m$ by $a^{O(a)}$ (and when $a=0$ we have that the desired inequality holds since $|s|^0 \cdot |\psi(s)| = |\psi(s)| \leq 1$).
    Given the induction and these observations, it is clear that $\psi$ is $(2\taylor, B)$-relaxed, where $B = (m+1)^{2\taylor} \cdot (2\taylor B_0)^{2\taylor} \cdot O(2\taylor)^{2\taylor} = \poly((\taylor \degree)^\taylor)$.
\end{proof}

\begin{remark}
\Cref{lem:mollifier-is-relaxed} is analogous to Lemma~18 of \cite{kane11focs}, which similarly gives an upper bound on the partial derivatives of the mollifier of \cite{kane11focs}.  The upper bound given in \cite{kane11focs} is exponential in $d$ because of $\exp(d)$-type factors which are involved in the definition of the mollifier in that work.
\end{remark}

\subsection{The core Taylor's theorem argument} \label{sec:taylor}

In this subsection we prove the following lemma:
\begin{lemma}                                     \label{lem:taylor}
    Let $\bst = (\bst_1, \dots, \bst_m)$ be a vector of nonnegative random variables, each of which is $(2\taylor,\eta)$-hyperconcentrated. Let $\psi : (\R^{\geq 0})^m \to [0,1]$ be $(2\taylor,C)$-relaxed.  Assume $\eta \leq \frac{1}{4m}$.  Then up to an additive error of $C \cdot O(m)^\taylor \cdot \eta^\taylor$, the expectation $\E[\psi(\bst)]$ is  determined by the moments of $\bst$ of degree up to~$2\taylor$.
\end{lemma}
Given \Cref{lem:we-got-hyperconcentration} and \Cref{lem:mollifier-is-relaxed}, observing that for our mollfier $\psi = \allsoftchecks$ we have $m=\poly(d)$ and hence $40e \sqrt{\anticoncgap} \ll {\frac 1 {4m}}$, we can indeed combine these results with \Cref{lem:taylor}. Recalling that the moments of $\bst$ are determined by our assumptions and parameter settings (since each statistic $\st_{i,j}$ has degree at most $2\degree$, the random variable $\bz$ is a $\wise$-wise independent Gaussian, and $2 \degree \cdot 2 \taylor \leq \wise$), this establishes \Cref{eq:end-of-suffering} as desired.

\begin{proofof}{\Cref{lem:taylor}}
We write $\mu$ to denote $\E[\bst]$. For a generic $s \in (\R^\geq 0)^m$,
Taylor's theorem implies that:
    \[
        \psi(s) = P(s) + \text{Err}(s), \quad \text{where } P(s) = \sum_{0 \leq |\alpha| < \taylor} \tfrac{1}{\alpha!} \partial_{\alpha} \psi(\mu) \cdot (s-\mu)^\alpha, \quad \text{Err}(s) = \sum_{|\alpha| = \taylor} \tfrac{1}{\alpha!}  \partial_\alpha \psi(s^*) \cdot (s-\mu)^\alpha
    \]
    for some $s^*$ on the open line segment from $\mu$ to~$s$.
As $P$ is a polynomial in~$s$ of degree at most $\taylor-1$, we have that $\E[P(\bst)]$ is exactly determined by the moments of~$\bst$ of degree up to~$\taylor-1$.  It therefore suffices to bound
    \begin{equation}    \label[ineq]{ineq:1-2}
        \abs{\E\bracks*{\text{Err}(\bst)}}
        = \abs*{\E\bracks*{\text{Err}(\bst)} \cdot \parens*{\bone_{\{\bst \approx_1 \mu\}}+ \bone_{\{\bst \not \approx_1 \mu\}}}}
        \leq \E\bracks*{\abs{\text{Err}(\bst)} \cdot \bone_{\{\bst \approx_1 \mu\}}} + \E\bracks*{\abs*{\text{Err}(\bst)} \cdot \bone_{\{\bst \not \approx_1 \mu\}}},
    \end{equation}
    where $\bst \approx_1 \mu$ means that $\bst_i$ is within a multiplicative factor of~$e$ of~$\mu_i$ (i.e., $\bst_i \approx_1 \mu_i$) for each $1 \leq i \leq m$.   Note that when this event occurs we also have $\bst^* \approx_1 \mu$, and hence by the $(2\taylor,C)$-relaxed property of $\psi,$ for all $\alpha$ such that $|\alpha|=T$, we have
     \begin{equation} \label{eq:copacetic}
       |\mu|^{\alpha} \cdot \abs{\partial_\alpha \psi(\bst^*)} \leq e^{\taylor} C .
    \end{equation}
    We proceed to analyze $ \E\bracks*{\abs{\text{Err}(\bst)} \cdot \bone_{\{\bst \approx_1 \mu\}}}$ as follows: we have
    \begin{align}
     \E\bracks*{\abs{\text{Err}(\bst)} \cdot \bone_{\{\bst \approx_1 \mu\}}}
     \le \sum_{|\alpha| = \taylor}  \abs{\partial_\alpha \psi(\bst^*)}
     \E\bracks*{|(\bs - \mu)^\alpha| \cdot \bone_{\{\bst \approx_1 \mu\}}}. \label{eq:nifty}
     \end{align}
     Now, if an $\alpha$ with $|\alpha|=\taylor$ has $|\mu|^\alpha=0$, this means that there is an index $\alpha_i>0$ such that $\mu_i = 0$. Recalling the definition of hyperconcentration (\Cref{def:hyperconcentrated}), it must be the case that $\bs_i$ is zero with probability 1 and hence any such $\alpha$ contributes zero to the right-hand side of \Cref{eq:nifty}. Thus we have that
     \begin{align}
        \E\bracks*{\abs{\text{Err}(\bst)} \cdot \bone_{\{\bst \approx_1 \mu\}}} &\leq
        \sum_{|\alpha| \leq \taylor, |\mu|^\alpha > 0}  \abs{\partial_\alpha \psi(\bst^*)}
     \E\bracks*{|(\bs - \mu)^\alpha| \cdot \bone_{\{\bst \approx_1 \mu\}}} \nonumber\\
&\le
        e^\taylor C \sum_{|\alpha| \leq \taylor, |\mu|^\alpha > 0} |\mu|^{-\alpha} \E\bracks*{|\bst - \mu|^\alpha} \tag{using \Cref{eq:copacetic}} \nonumber\\
        & \leq e^\taylor C \cdot \#\{\alpha : |\alpha| = \taylor\}\cdot \eta^\taylor  \tag{\Cref{lem:hyperholder}}
        \nonumber\\
        &  \leq C (em)^\taylor  \cdot \eta^\taylor.
        \label[ineq]{ineq:yo1}
    \end{align}
For the second term in \Cref{ineq:1-2} we use Cauchy--Schwarz:
     \begin{multline}   \label[ineq]{ineq:yo2a}
        \E\bracks*{\abs*{\text{Err}(\bst)} \cdot \bone_{\{\bst \not \approx_1 \mu\}}} \leq \sqrt{\E\bracks*{\abs*{\text{Err}(\bst)}^2}}\sqrt{\E\bracks*{\bone_{\{\bst \not \approx_1 \mu\}}^2}}
        =  \sqrt{\E\bracks*{\abs*{\psi(\bst) - P(\bst)}^2}} \sqrt{\Pr[\bst \not \approx_1 \mu]} \\ \leq \parens*{1+\sqrt{\E\bracks*{P(\bst)^2}}} \cdot \sqrt{m} \cdot (2\eta)^\taylor,
     \end{multline}
where in the last step we used $|\psi(\bst)| \leq 1$ for the first factor, and \Cref{prop:hypermarkov} (plus a union bound over the $m$ coordinates of $\bst$) for the second factor.
Our handling of $\E\bracks*{P(\bst)^2}$ will be similar to \Cref{ineq:yo1}: we have that
 \begin{align}
        \E\bracks*{P(\bst)^2}  &\leq  \sum_{0 \leq |\alpha|,|\beta| < \taylor} \abs{\partial_\alpha\psi(\mu)}\cdot \abs{\partial_\beta\psi(\mu)} \cdot \E\bracks*{\abs{\bst-\mu}^{\alpha+\beta}}\nonumber \\
        &=  \sum_{0 \leq |\alpha|,|\beta| < \taylor: |\mu|^{\alpha+\beta}>0} \abs{\partial_\alpha\psi(\mu)}\cdot \abs{\partial_\beta\psi(\mu)} \cdot \E\bracks*{\abs{\bst-\mu}^{\alpha+\beta}}\nonumber \\
        & \leq  \sum_{0 \leq |\alpha|,|\beta| < \taylor: |\mu|^{\alpha+\beta}>0}  C^2 |\mu|^{-\alpha - \beta} \cdot \E\bracks*{\abs{\bst-\mu}^{\alpha+\beta}} \tag{$\psi$ is $(2\taylor,C)$-relaxed}\\
        & \leq C^2 \sum_{0 \leq |\alpha|,|\beta| < \taylor} \eta^{|\alpha+\beta|} \tag{\Cref{lem:hyperholder}} \\
        & \leq C^2 \sum_{k=0}^\taylor (2m\eta)^k \leq 2C^2,  \label[ineq]{ineq:yob}
    \end{align}
    where the last inequality used the assumption $\eta \leq \frac{1}{4m}$ and the equality uses reasoning similar to our earlier analysis of $|\alpha|=\taylor$ such that $|\mu|^\alpha=0.$  Putting \Cref{ineq:yo1,ineq:yo2a,ineq:yob} into \Cref{ineq:1-2} yields
    \[
        \abs{\E\bracks*{\text{Err}(\bst)}} \leq C (em)^\taylor \cdot \eta^\taylor + (1+\sqrt{2}C) \cdot \sqrt{m}  \cdot (2\eta)^\taylor = C \cdot O(m)^\taylor \cdot \eta^\taylor,
    \]
    as claimed.

    This concludes the proof of \Cref{lem:taylor} and thus also of \Cref{lem:page14}.
\end{proofof}

\section*{Acknowledgments}
We thank Avi Wigderson for a key conceptual suggestion toward the proof of \Cref{thm:HZL}. 

R.O.~is supported by NSF grant CCF-1717606.  R.A.S.~is supported by NSF grants CCF-1814873, IIS-1838154, CCF-1563155, and by the Simons Collaboration on Algorithms and Geometry.  L.-Y.T.~is supported by NSF grant CCF-1921795.  This material is based upon work supported by the National Science Foundation under grant numbers listed above. Any opinions, findings and conclusions or recommendations expressed in this material are those of the author and do not necessarily reflect the views of the National Science Foundation (NSF).

\bibliography{allrefs}{}

\begin{thebibliography}{OSTK21}

\bibitem[ABI86]{ABI85}
Noga Alon, L\'{a}szl\'{o} Babai, and Alon Itai.
\newblock A fast and simple randomized parallel algorithm for the maximal
  independent set problem.
\newblock {\em Journal of Algorithms}, 7(4):567--583, 1986.

\bibitem[CW01]{CW:01}
Anthony Carbery and James Wright.
\newblock {Distributional and $L^q$ norm inequalities for polynomials over
  convex bodies in $\mathbb{R}^n$}.
\newblock {\em Mathematical Research Letters}, 8(3):233--248, 2001.

\bibitem[DKN10]{DKN10}
Ilias Diakonikolas, Daniel Kane, and Jelani Nelson.
\newblock Bounded independence fools degree-2 threshold functions.
\newblock In {\em Proceedings of the 51st Annual Symposium on Foundations of
  Computer Science (FOCS)}, pages 11--20, 2010.

\bibitem[Jof74]{Jof74}
Anatole Joffe.
\newblock On a set of almost deterministic {$k$}-independent random variables.
\newblock {\em The Annals of Probability}, 2(1):161--162, 1974.

\bibitem[Kan11a]{Kane11ccc}
Daniel Kane.
\newblock {$k$}-independent {G}aussians fool polynomial threshold functions.
\newblock In {\em Proceedings of the 26th Conference on Computational
  Complexity (CCC)}, pages 252--261, 2011.

\bibitem[Kan11b]{kane11focs}
Daniel Kane.
\newblock A small {PRG} for polynomial threshold functions of {G}aussians.
\newblock In {\em Proceedings of the 52nd Annual Symposium on Foundations of
  Computer Science (FOCS)}, pages 257--266, 2011.

\bibitem[Kan12]{Kane12}
Daniel Kane.
\newblock A structure theorem for poorly anticoncentrated {G}aussian chaoses
  and applications to the study of polynomial threshold functions.
\newblock In {\em Proceedings of the 53rd Annual Symposium on Foundations of
  Computer Science (FOCS)}, pages 91--100, 2012.

\bibitem[Kan14]{Kane14-subpoly}
Daniel Kane.
\newblock A pseudorandom generator for polynomial threshold functions of
  {G}aussians with subpolynomial seed length.
\newblock In {\em Proceedings of the 29th Annual {C}onference on
  {C}omputational {C}omplexity (CCC)}, pages 217--228, 2014.

\bibitem[Kan15]{Kane15}
Daniel Kane.
\newblock A polylogarithmic {PRG} for degree 2 threshold functions in the
  {G}aussian setting.
\newblock In {\em Proceedings of the 30th {C}onference on {C}omputational
  {C}omplexity (CCC)}, pages 567--581, 2015.

\bibitem[KM15]{KothariMeka15}
Pravesh Kothari and Raghu Meka.
\newblock Almost optimal pseudorandom generators for spherical caps.
\newblock In {\em Proceedings of the 47th Annual on Symposium on Theory of
  Computing (STOC)}, pages 247--256, 2015.

\bibitem[KM21]{KelleyMeka21}
Zander Kelley and Raghu Meka.
\newblock {Random restrictions and PRGs for PTFs in Gaussian Space}.
\newblock Available at
  \href{https://arxiv.org/abs/2103.14134}{https://arxiv.org/abs/2103.14134},
  2021.

\bibitem[MZ10]{MZstoc10}
Raghu Meka and David Zuckerman.
\newblock Pseudorandom generators for polynomial threshold functions.
\newblock In {\em Proceedings of the 42nd ACM Symposium on Theory of Computing
  (STOC)}, pages 427--436, 2010.

\bibitem[MZ13]{MZ13}
Raghu Meka and David Zuckerman.
\newblock Pseudorandom generators for polynomial threshold functions.
\newblock {\em SIAM Journal on Computing}, 42(3):1275--1301, 2013.

\bibitem[O'D14]{ODbook}
Ryan O'Donnell.
\newblock {\em Analysis of Boolean Functions}.
\newblock Cambridge University Press, 2014.
\newblock Available at \url{http://analysisofbooleanfunctions.net/}.

\bibitem[OST20]{OST20}
Ryan O'Donnell, Rocco~A. Servedio, and Li-Yang Tan.
\newblock Fooling gaussian ptfs via local hyperconcentration.
\newblock In {\em Proceedings of the 52nd Annual ACM SIGACT Symposium on Theory
  of Computing (STOC)}, page 1170–1183, 2020.

\bibitem[OSTK21]{OSTK21arxiv-v1}
Ryan O'Donnell, Rocco~A. Servedio, Li-Yang Tan, and Daniel Kane.
\newblock {Fooling Gaussian PTFs via Local Hyperconcentration}.
\newblock Available at
  \href{https://arxiv.org/abs/2103.07809}{https://arxiv.org/abs/2103.07809},
  2021.

\end{thebibliography}
\bibliographystyle{alpha}

\appendix
\section{Omitted proofs}    \label{app:omitted}

\begin{proofof}{\Cref{prop:2}}
Our goal will be to establish the following identity:
\begin{equation}    \label{eqn:verif}
    \restr{g}{\noisegeneric}{x}(y) = g(\sqrt{1-\noisegeneric} x + \sqrtnoise y) = \sum_{\alpha,\beta \in \N^n} \wh{g}(\alpha+\beta) 
    \sqrt{\Pr[\textnormal{Bin}(\alpha+\beta,\noisegeneric) = \beta]}\,h_\alpha(x)\,h_\beta(y).
\end{equation}
From this we immediately deduce
\[
    \wh{\restr{g}{\noisegeneric}{x}}(\beta) = \sum_{\gamma \geq \beta} \wh{g}(\gamma) \sqrt{ \Pr[\textnormal{Bin}(\gamma,\noisegeneric) = \beta]}\, h_{\gamma - \beta}(x).
\]
Now Parseval's identity (i.e., taking the expected square over $\bx \sim \normal(0,1)^n$ and using orthonormality of the $h_{\gamma-\beta}$'s) yields \Cref{prop:2}.

It remains to verify \Cref{eqn:verif}.  This identity is a direct consequence of the following univariate special case:
\begin{equation}    \label{eqn:add-hermite}
    h_m(\sqrt{1-\noisegeneric} x + \sqrtnoise y) = \sum_{i+j=m} \sqrt{\Pr[\textnormal{Bin}(m,\noisegeneric) = j]}\,h_i(x)\,h_j(y).
\end{equation}
To obtain \Cref{eqn:verif}, one simply substitutes \Cref{eqn:add-hermite} into the multivariate Hermite expansion $g(z) = \sum_\gamma \wh{g}(\gamma) h_\gamma(z)$.

Finally, \Cref{eqn:add-hermite} is a variant of the standard identity~\cite[Ex.~11.11]{ODbook} concerning \mbox{$H_m(x+y)$}.  To prove \Cref{eqn:add-hermite}, we recall (\cite[Eq.~(11.8)]{ODbook})  the generating function definition of $h_m$:
\[
    \exp(tz - \tfrac12 t^2) = \sum_{m=0}^\infty \tfrac{1}{\sqrt{m!}} h_m(z) t^m.
\]
Substitute $z = \sqrt{1-\noisegeneric}x + \sqrtnoise y$, and use $-\tfrac12 = -\tfrac{1-\noisegeneric}{2} -\tfrac{\noisegeneric}{2}$ on the left-hand side.  This yields
\[
    \exp(t\sqrt{1-\noisegeneric} x - \tfrac{1-\noisegeneric}{2}t^2) \exp(t\sqrt{\noisegeneric}y - \tfrac{\noisegeneric}{2}t^2)  = \sum_{m=0}^\infty \tfrac{1}{\sqrt{m!}} h_m(\sqrt{1-\noisegeneric} x + \sqrtnoise y) t^m.
\]
Now on the left we use the generating function twice again (with $\sqrt{1-\noisegeneric} t$ and $\sqrtnoise t$ replacing~$t$), yielding
\[
    \parens*{\sum_{i=0}^\infty \tfrac{1}{\sqrt{i!}}h_i(y)(\sqrtnoise t)^i} \parens*{\sum_{j=0}^\infty \tfrac{1}{\sqrt{j!}}h_j(x)(\sqrt{1-\noisegeneric}t)^j} = \sum_{m=0}^\infty \tfrac{1}{\sqrt{m!}} h_m(\sqrtnoise y + \sqrt{1-\noisegeneric}x) t^m.
\]
\Cref{eqn:add-hermite} now follows by considering the coefficient on $t^m$ on both sides.
\end{proofof}

\begin{proofof}{\Cref{thm:kanes-prg}}
As mentioned, this result essentially appears in Section~6 of~\cite{kane11focs}.

The following construction of $k$-wise independent tuples is well known~\cite{Jof74,ABI85}:
\begin{theorem}                                     \label{thm:k-wise-basic}
    For any $k, n, M \in \N$, there is an efficient deterministic algorithm that takes in a uniformly random bit-string of length $O(k \max\{M,\log n\})$ and outputs a sequence $(\bX_1, \dots, \bX_n) \in \{0, 1, \dots, 2^M-1\}^n$ that is $k$-wise independent with respect to the uniform distribution on \mbox{$\{0, 1, \dots, 2^M-1\}$.}
\end{theorem}
We could convert each $M$-bit $\bX_i$ to a near-Gaussian using \Cref{eqn:bits-to-gaussians}; the result would be a $k$-wise independent vector in with respect to an ``$O(\frac{1}{\sqrt{M}})$-near-Gaussian'' distribution.  However one can convert $M$-bit strings to near-Gaussians with exponentially better accuracy, via the Box--Muller transform; as Kane~\cite[Proof of Cor.~2]{kane11focs} shows, one can deterministically and efficiently convert $\bX_i$ to a random variable $\by_i$ whose distribution can be coupled to a true Gaussian $\bz_i \sim \normal(0,1)$] such that $\Pr[|\by_i - \bz_i| \leq \delta] \geq 1-\delta$ for $\delta = \Theta(2^{-M/2})$.

Given this, we may then use the following lemma explicitly proven by Kane~\cite{kane11focs}, which is a relatively straightforward consequence of the Carbery--Wright theorem (\Cref{thm:CW}):
\begin{lemma} \label{lem:close-coords-fool}
    (\cite[Lem.~21]{kane11focs}.) Let $\bz$ be a $k$-wise independent $n$-dimensional Gaussian random vector.  Suppose that $\by$ is an $n$-dimensional random vector such that $\by$ and $\bz$ may be coupled so that $\Pr[|\by_i - \bz_i| \leq \delta] \geq 1-\delta$.  Finally, suppose one can show that for $\len = 1/\noisegeneric$ and $\bz^{(1)}, \dots, \bz^{(\len)}$ i.i.d.\ copies of $\bz$, the sum $\sqrt{\noisegeneric} \bz^{(1)} + \cdots + \sqrt{\noisegeneric} \bz^{(\len)}$  is $\eps/2$-fooling for degree-$d$ Gaussian PTFs.  Then for $\by^{(1)}, \dots, \by^{(\len)}$ being i.i.d.\ copies of $\by$, the sum $\sqrt{\noisegeneric} \by^{(1)} + \cdots + \sqrt{\noisegeneric} \by^{(\len)}$ is $\eps'$-fooling for degree-$d$ Gaussian PTFs, where
    \[
        \eps' = \eps/2 + O(\len n \delta) + O(\degree \sqrt{\len n} \delta^{1/\degree}\log(1/\delta)).
    \]
\end{lemma}
\Cref{thm:kanes-prg} now follows by taking  $M = O(d \log (d \len n / \eps))$ and thereby achieving $\delta = (d \len n / \eps)^{O(d)}$.
\end{proofof}

\section{Appendix by Daniel Kane}    \label{app:kane}

\newcommand{\var}{\mathbf{Var}}
\newcommand{\cov}{\mathbf{Cov}}
\renewcommand{\R}{\mathbb{R}}
\newcommand{\pr}{\mathbf{Pr}}
\newcommand{\hypvar}{\HV}

\newcommand{\D}{\mathcal{D}}

\newcommand{\bmu}{\boldsymbol{\mu}}

\begin{lemma}\label{anticoncentration lemma}
Let $p$ be a degree-$d$ polynomial with $p(x)\geq 0$ for all $x$. Then for $\bx \sim \normal(0,1)^n$ and $\eps > 0$ we have that
$$
\pr[p(\bx) < \eps \E_\bx[p(\bx)]] = O(d \eps^{1/d}).
$$
\end{lemma}
\begin{proof}
This follows from Carbery-Wright and the observation that $\|p\|_2 \geq \E_\bx[p(\bx)]$.
\end{proof}

\begin{theorem} [Local Hyperconcentration Theorem, improved version] \label{main theorem}
There exists a constant $c>0$ so that for any $1>\eps,\beta>0$ and $R\geq 1$ sufficiently small if $(g_{\bup})_{\bup \sim \Upsilon}$ is a nice distribution over degree-$d$ polynomials in $n$ variables then for
$$
\lambda \leq \frac{c\eps \beta}{R d^{9/2}}
$$
then for $\bx \sim \normal(0,1)^n$ with probability at least $1-\beta$ we have that
$$
\E_\bup[\hypvar_R((g_\bup)_{\lambda|\bx})] \leq \eps^2 \E_\bup [ \| (g_\bup)_{\lambda|\bx}\|_2^2 ].
$$
\end{theorem}
\begin{proof}
The proof of this theorem will depend on basic facts about the sequence of derivatives of $(g_\upsilon)_{\lambda|x}$. Firstly, we establish some notation, we let $D_y f(x)$ denote the directional derivative of $f$ at $x$ in the $y$ direction. We begin with the following Lemma relating the size of functions and their derivatives:

\begin{lemma}\label{strong anticoncentration lemma}
Let  $(g_{\bup})_{\bup \sim \Upsilon}$ be a nice distribution over degree-$d$ polynomials. Let $\bx$ and $\by$ be independent $\normal(0,1)^n$ Gaussian random variables. Then for $\eps>0$ we have that
$$
\pr_{\bx,\by}[\E_\bup [|g_{\bup}(\bx)|^2] \leq \eps^2 \E_{\bup} [|D_{\by} (g_{\bup}(\bx))|^2]] = O(d^2\eps).
$$
\end{lemma}
\begin{proof}
We begin with the case where $g_\upsilon$ is actually a constant family (i.e. is just a single function). This result then follows immediately from Lemma 9 of [D. Kane ``The Correct Exponent for the Gotsman-Linial Conjecture''].

From here we generalize to the case where $g_\upsilon$ is a linear polynomial in $\upsilon$. By the previous case, we have that
$$
\pr_{\bx,\by,\bup}[|g_{\bup}(\bx)|^2 \leq \eps^2 |D_{\by} (g_{\bup}(\bx))|^2] = O(d^2\eps).
$$
On the other hand since $g_\upsilon(x)$ and $D_y (g_\upsilon(x))$ are linear functions of $\upsilon$. Therefore, for any $x$ and $y$, with at least $50\%$ probability over the choice of $\bup$ we have that $\E_{\bup}[|g_{\bup}(x)|^2 ] \ll |g_{\bup}(x)|^2$ and $|D_y (g_{\bup}(x))|^2 \ll \E_{\bup}[|D_y (g_{\bup}(x))|^2]$. Therefore, whenever
$$
\E_{\bup} [|g_{\bup}(x)|^2] \leq \eps^2 \E_{\bup} [|D_y (g_{\bup}(x))|^2]
$$
there is at least a $50\%$ probability over the choice of ${\bup}$ that
$$
|g_{\bup}(x)|^2 \leq \eps^2 C |D_y (g_{\bup}(x))|^2
$$
for some positive constant $C$. However, the latter happens with probability $O(d^2 \eps)$ and is at least half of

$$
\pr_{\bx,\by}[\E_{\bup} [|g_{\bup}(\bx)|^2] \leq \eps^2 \E_{\bup} [|D_{\by} (g_{\bup}(\bx))|^2]].
$$
Therefore, this latter probability is $O(d\eps^2).$

Finally, we can handle the generic case. Let $h_a(\upsilon)$ be an orthonormal basis for the polynomials in $\upsilon$. We can write $g_\upsilon(x)$ as $\sum_a h_a(\upsilon)p_a(x)$ for some polynomials $p_a$. Define
$$
h_{\bmu}(x):= \sum_a {\bmu}_a p_a(x)
$$
where $\bmu$ is a Gaussian random variable with as many components as there are terms in the above decomposition of $g_\upsilon$. It is easy to see that for any $x$ and $y$ that
$$
\E_{\bup}[|g_{\bup}(x)|^2] = \sum_a p_a^2(x) = \E_{\bmu}[|h_{\bmu}(x)|^2]
$$
and
$$
\E_{\bup}[|D_y g_{\bup}(x)|^2] = \sum_a (D_y p_a(x))^2 = \E_\mu[|D_y h_\mu(x)|^2].
$$
Since our Lemma holds for $h$, it must therefore also hold for $g$. This completes our proof.
\end{proof}
Our theorem will rest upon the following notion of a random derivative sequence:
\begin{definition}
Given a nice distribution  $(g_{\bup})_{\bup \sim \Upsilon}$ of degree-$d$ polynomials, a \emph{random derivative sequence} for $g$ is a sequence of the form
$$
\E_{\bup}[|g_{\bup}(\bx)|^2], \E_{\bup}[|D_{\by_1}g_{\bup}(\bx)|^2], \E_{\bup}[|D_{\by_2}D_{\by_1}g_{\bup}(\bx)|^2], \ldots \E_{\bup}[|D_{\by_d}D_{\by_{d-1}}\cdots D_{\by_1}g_{\bup}(\bx)|^2],
$$
where $\bx,\by_1,\by_2,\ldots,\by_d$ are independent Gaussian random variables.

We will often denote the $k^{th}$ term by
$$
\D^k_{\bx,\by_i}(g_\upsilon) :=\E_{\bup}[|D_{\by_k}D_{\by_{k-1}}\cdots D_{\by_1}g_{\bup}(\bx)|^2].
$$
\end{definition}

\begin{corollary}\label{derivative sequence corollary}
Let $(g_{\bup})_{\bup \sim \Upsilon}$ be a nice distribution over degree-$d$ polynomials. Then with probability at least $1-\eps$ over the choice of $\bx,\by_i$ we have that:
$$
\D^{k+1}_{\bx,\by_i}(g_\upsilon) = O(d^6/\eps^2)\D^{k}_{\bx,\by_i}(g_\upsilon)
$$
for all $0\leq k \leq d$.
\end{corollary}
\begin{proof}
It follows immediately from \Cref{strong anticoncentration lemma} that for each $k$
$$
\D^{k+1}_{\bx,\by_i}(g_\upsilon) = O(d^6/\eps^2)\D^{k}_{\bx,\by_i}(g_\upsilon)
$$
with probability $1-\eps/d$. The full result follows from a union bound over $k$.
\end{proof}

This says that a typical random derivative sequence does not increase too rapidly. However, we show that there is a kind of converse for non-hyperconcentrated families:
\begin{proposition}\label{sequence increase proposition}
Let $(g_{\bup})_{\bup \sim \Upsilon}$ be a nice distribution of degree-$d$ polynomials so that for some $R,\eps>0$ we have that
$$
\E_{\bup}[\hypvar_R((g_{\bup}))] \geq \eps^2 \E_{\bup} [ \| (g_{\bup})\|_2^2 ].
$$
Then with at least $50\%$ probability over a choice of random Gaussians $\bx,\by_i$ we have that for some $0\leq k \leq d$ that
$$
\D^{k+1}_{\bx,\by_i}(g_\upsilon) > \Omega(\eps^2/(d^3 R^2))\D^{k}_{\bx,\by_i}(g_\upsilon).
$$
\end{proposition}

In order to prove this, we will need to talk about the Hermite parts of $g$. Recall that for $g$ a polynomial, the notation $g^{= k}$ denotes the degree-$k$ Hermite part of $g$.


We will make use of the following facts:
\begin{fact}
For any polynomial $g$ and $R>0$,
$$
\hypvar_R(g) = \sum_k k R^{2k} \|g^{= k}\|_2^2.
$$
\end{fact}
\begin{fact}
For any polynomial $g$ and vector $v$:
$$
(D_v g)^{=k} = D_v(g^{=k+1}).
$$
\end{fact}
\begin{fact}
For any $g$ with $g=g^{= k}$ and $\bx$ a random Gaussian
$$
\E_{\bx} [\| D_{\bx} g \|_2^2] = k\|g\|_2^2.
$$
\end{fact}

We will also need to know that with reasonable probability that some element of a random derivative sequence is not too  small.
\begin{lemma}\label{decrease lemma}
Let $(g_{\bup})_{\bup \sim \Upsilon}$ be a nice distribution of degree-$d$ polynomials and let $m\geq d$. Then with probability at least $1-d/m$ over a choice of random Gaussians $\bx,\by_i$ there exists a $0\leq k\leq d$ so that
$$
\D^{k}_{\bx,\by_i}(g_\upsilon) \geq \Omega(1/m)^{3k} \E_{\bup}[ \|g_{\bup}\|_2^2]/2.
$$
\end{lemma}
\begin{proof}
We proceed by induction on $d$. If $d=0$, then $\D^{0}_{\bx,\by_i}(g_\upsilon) = \E_{\bup}[ \|g_{\bup}\|_2^2]$ and we are done.

For the inductive step, we break into two cases. On the one hand, if
$$
\E_{\bup}[ \|g_{\bup}\|_2^2] \geq 100m \E_{\bup}[\var_{\bx}(g_{\bup}(\bx))],
$$
then we have that
\begin{align*}
& \E_{\bx}[\E_{{\bup}}[\|g_{\bup}\|_2^2 \mathbf{1}\{|g_{\bup}(\bx)| < (2/3)\|g_{\bup}\|_2\}]]\\
= &\E_{\bup}[\E_{\bx}[\|g_{\bup}\|_2^2 \mathbf{1} \{|g_{\bup}(\bx)| < (2/3)\|g_{\bup}\|_2\}]]\\
\leq & \E_{\bup}\left[9(\|g_{\bup}\|_2^2)\left(\frac{\var_{\bx}(g_{\bup}(\bx))}{\|g_{\bup}\|_2^2} \right)\right]\\
= & 9\E_{\bup}[ \|g_{\bup}\|_2^2]/2\\
\leq & (1/10m) \E_{\bup}[ \|g_{\bup}\|_2^2].
\end{align*}
However, if
$$
\D^{0}_{\bx,\by_i}(g_{\upsilon}) < \E_{\bup}[ \|g_{\bup}\|_2^2]/2,
$$
it must be the case that
$$
\E_{{\bup}}[\|g_{\bup}\|_2^2 \mathbf{1}\{|g_{\bup}(x)| < (2/3)\|g_{\bup}\|_2\}] \geq (1/6) \E_{{\bup}}[\|g_{\bup}\|_2^2],
$$
and by the Markov Inequality, this happens with probability at most $1/m$.

On the other hand, if
$$
\E_{\bup}[ \|g_{\bup}\|_2^2] \leq 100m \E_{\bup}[\var_{\bx}(g_{\bup}(\bx))],
$$
then
$$
\E_{\by_1}[\E_{\bup}[\|D_{\by_1} g_{\bup}\|_2^2]] = \E_{\bup}[\var_{\bx}(g_{\bup}(\bx))].
$$
Since $\E_{\bup}[\|D_{y_1} g_{\bup}\|_2^2]$ is a non-negative quadratic function of $y_1$, \Cref{anticoncentration lemma} implies that with probability at least $1-1/m$ we have that
$$
\E_{\bup}[\|D_{\by_1} g_{\bup}\|_2^2] \geq (1/m^2) \E_{\bup}[\var_{\bx}(g_{\bup}(\bx))] \geq \Omega(1/m^3)\E_{\bup}[ \|g_{\bup}\|_2^2].
$$
For such an outcome of $\by_1$, we can apply our inductive hypothesis to $(D_{\by_1}g_\upsilon)$.
\end{proof}

We are now prepared to prove \Cref{sequence increase proposition}.
\begin{proof}
Note that
$$
\E_{\bup}[\hypvar_R(g_{\bup})] = \sum_{k=1}^d k R^{2k} \E_{\bup}[\|g_{\bup}^{=k}\|_2^2].
$$
Therefore, under our hypothesis, there must be a $k$ so that
$$
\E_{\bup}[\|g_{\bup}^{=k}\|_2^2] \geq (\eps^2/(4R^2))^k\E_{\bup}[\|g_{\bup}\|_2^2].
$$
Notice that
$$
\E_{{\bup}}[\| D_{y_k}D_{y_{k-1}}\cdots D_{y_1}g_{\bup}^{=0}\|_2^2]
$$
is a non-negative degree $2k$ polynomial in $y_1,\ldots,y_k$ with average value at least $(\eps^2/(4R^2))^k\E_{\bup}[\|g_{\bup}\|_2^2]$. Therefore, \Cref{anticoncentration lemma} implies that with probability at least $5/6$ we have that
$$
\| D_{\by_k}D_{\by_{k-1}}\cdots D_{\by_1}g_{\bup}^{=0}\|_2^2 \geq \Omega(\eps^2/(d^2 R^2))^k \E_{\bup}[\|g_{\bup}\|_2^2].
$$
By \Cref{decrease lemma} with probability at least $5/6$ there is a $k' > k$ so that
$$
\D^{k'}_{\bx,\by_i}(g_\upsilon) \geq \Omega(1/d)^{3(k'-k)}\| D_{\by_k}D_{\by_{k-1}}\cdots D_{\by_1}g_\upsilon^{=0}\|_2^2/2 \geq \Omega(\eps^2/(d^3 R^2))^{k'}\E_{\bup}[\|g_{\bup}\|_2^2].
$$
Finally, with probability at least $5/6$ we have that
$$
\D^{0}_{\bx,\by_i}(g_\upsilon) = O(1)\E_{\bup}[\|g_{\bup}\|_2^2].
$$
Thus, if all three of these events hold (which happens with probability at least $1/2$), there will be some $d\geq k> 1$ so that
$$
\D^{k}_{\bx,\by_i}(g_\upsilon) \geq \Omega(\eps^2/(d^3 R^2))^{k}\D^{0}_{\bx,\by_i}(g_\upsilon).
$$
Therefore, there must also be a $k$ so that
$$
\D^{k+1}_{\bx,\by_i}(g_\upsilon)\geq \Omega(\eps^2/(d^3 R^2)) \D^{k}_{\bx,\by_i}(g_\upsilon). \qedhere
$$
\end{proof}

Our Theorem will now follow from the tension between \Cref{derivative sequence corollary} and \Cref{sequence increase proposition} along with the observation that
$$
\D^k_{\bx,\by_i}((g_\upsilon)_{\lambda|\bx}) = \lambda^{2k} \D^k_{\bz,\by_i}(g_\upsilon)
$$
where $\bz=\sqrt{1-\lambda}\bx'+\lambda \bx$. Note also that $\bz$ is a standard Gaussian if $\bx'$ and $\bx$ are. In particular,  \Cref{derivative sequence corollary} tells us that with probability $1-\beta/2$ that
\begin{equation}\label{derivative upper bound equation}
\D^{k+1}_{\bz,\by_i}(g_\upsilon) = O(d^6/\beta^2)\D^k_{\bz,\by_i}(g_\upsilon)
\end{equation}
for all $0\leq k \leq d$. On the other hand, if
$$
\E_{\bup}[\hypvar_R((g_{\bup})_{\lambda|\bx})] > \eps^2 \E_{\bup} [ \| (g_{\bup})_{\lambda|x}\|_2^2 ]
$$
then \Cref{sequence increase proposition} implies that with at least $50\%$ probability that there is a $0\leq k\leq d$ so that
$$
\D^{k+1}_{\bx,\by_i}((g_\upsilon)_{\lambda|x}) > \Omega(\eps^2/(d^3 R^2))\D^{k}_{\bx,\by_i}((g_\upsilon)_{\lambda|\bx}).
$$
But this is equivalent to saying that
$$
\D^{k+1}_{\bz,\by_i}(g_\upsilon) > \Omega(\lambda^{-2}\eps^2/(d^3 R^2))\D^k_{\bz,\by_i}(g_\upsilon).
$$
However, given our setting of $\lambda$ this would contradict Equation \eqref{derivative upper bound equation}. Therefore, the probability of Equation \eqref{derivative upper bound equation} being violated is at most $\beta/2$, but is at least half the probability that
$$
\E_\bup[\hypvar_R((g_\bup)_{\lambda|\bx})] > \eps \E_\bup [ \| (g_\bup)_{\lambda|\bx}\|_2^2 ].
$$
Hence we conclude that the latter probability is at most $\beta.$
\end{proof}
 
\end{document}